\documentclass[journal]{IEEEtran}
\usepackage[english]{babel}

\usepackage{ifpdf}

\usepackage{cite}
\usepackage{hyperref}

\ifCLASSINFOpdf
	\usepackage[pdftex]{graphicx}
	\graphicspath{{./figures/}}
\else
	\usepackage[dvips]{graphicx}
	\graphicspath{./figures/}
\fi
\usepackage{color}
\usepackage{pgf, tikz, pgfplots}
\usetikzlibrary{shapes, arrows, automata}

\usepackage[cmex10]{amsmath}
\usepackage{amsfonts, amssymb, amsthm}
\usepackage{mathrsfs}
\usepackage{upgreek}

\usepackage{algorithm,algpseudocode}

\usepackage{enumerate}
\usepackage{multirow}
\usepackage{rotating}
\usepackage{subcaption}
\usepackage{url}

\usepackage{needspace}

\input{mysymbol.sty}

\def\fol{f_{\textrm{ol}}}
\def\tol{\leftrightarrow}

\def\QED{$\blacksquare$}
\def\qu{\breve{u}} 
\def\qU{\breve{\ccalU}} 
\def\qcm{\breve{c}} 
\def\qV{\breve{\ccalV}} 
\def\qK{\breve{K}} 
\def\qQ{\breve{Q}} 

\definecolor{pennblue}{rgb}{0.004,0.145,0.431} 
\definecolor{penngreen}{rgb}{0,0.557,0} 
\definecolor{pennyellow}{rgb}{0.949,0.757,0} 
\definecolor{pennorange}{rgb}{0.765,0.353,0} 
\definecolor{pennred}{rgb}{0.584,0,0.102} 
\definecolor{pennpurple}{rgb}{0.290,0,0.259} 

\newtheorem{proposition}{\hspace{0pt}\bf Proposition}
\newtheorem{theorem}{\hspace{0pt}\bf Theorem}
\newtheorem{corollary}{\hspace{0pt}\bf Corollary}
\newtheorem{remark}{\hspace{0pt}\bf Remark}
\newtheorem{definition}{\hspace{0pt}\bf Definition}

\begin{document}
%
\title{Hierarchical Overlapping Clustering of Network Data Using Cut Metrics}
%
%
%

\author{Fernando~Gama,~
        Santiago~Segarra,~
        and~Alejandro~Ribeiro
\thanks{Work in this paper is supported by NSF CCF 1217963. F. Gama and A. Ribeiro are with the Department
of Electrical and Systems Engineering, University of Pennsylvania. S. Segarra is with the Institute for Data, Systems, and Society, Massachusetts Institute of Technology. Emails: fgama@seas.upenn.edu, segarra@mit.edu, aribeiro@seas.upenn.edu.}}

%
%

\markboth{IEEE TRANSACTIONS ON SIGNAL AND INFORMATION PROCESSING OVER NETWORKS (ACCEPTED)}%
{Hierarchical Overlapping Clustering of Network Data Using Cut Metrics}
%



\maketitle

\begin{abstract}
A novel method to obtain hierarchical and overlapping clusters from network data -- i.e., a set of nodes endowed with pairwise dissimilarities -- is presented. The introduced method is \emph{hierarchical} in the sense that it outputs a nested collection of groupings of the node set depending on the resolution or degree of similarity desired, and it is \emph{overlapping} since it allows nodes to belong to more than one group. Our construction is rooted on the facts that a hierarchical (non-overlapping) clustering of a network can be equivalently represented by a finite ultrametric space and that a convex combination of ultrametrics results in a cut metric. By applying a hierarchical (non-overlapping) clustering method to multiple dithered versions of a given network and then convexly combining the resulting ultrametrics, we obtain a cut metric associated to the network of interest. We then show how to extract a hierarchical overlapping clustering structure from the aforementioned cut metric. Furthermore, the so-called overlapping function is presented as a tool for gaining insights about the data by identifying meaningful resolutions of the obtained hierarchical structure. Additionally, we explore hierarchical overlapping quasi-clustering methods that preserve the asymmetry of the data contained in directed networks. Finally, the presented method is illustrated via synthetic and real-world classification problems including handwritten digit classification and authorship attribution of famous plays.
\end{abstract}

\begin{IEEEkeywords}
Clustering, Network theory, Cut metrics, Hierarchical clustering, Covering, Dithering.
\end{IEEEkeywords}

%
\IEEEpeerreviewmaketitle


%
%
%
%

 

\section{Introduction}
	\label{sec:intro}

\IEEEPARstart{C}{onsider} a dataset where each element can be represented as a node in a network with a pairwise dissimilarity function. In this setting, the general objective of clustering is to group those nodes that are more similar to each other than to the rest, according to the relationship established by the dissimilarity function \cite{jaindubes88,castro02}. Clustering and its generalizations are ubiquitous tools since they are used in a wide variety of fields such as psychology \cite{cattell43}, social network analysis \cite{lu11}, political science \cite{paulus15}, neuroscience \cite{ozdemir15}, among many others~\cite{dalton15,zao15}.

Traditional clustering methods provide only one partitioning of the node set in such a way that each data point belongs to one and only one block of the partition. An important limitation of these traditional clustering methods is that the dataset may present a complex data structure at several resolutions or levels of similarity, and outputting only one partition may not be adequate in portraying the different grouping degrees that may be present. \emph{Hierarchical} clustering solves this issue by providing a collection of partitions, indexed by a resolution parameter, that can be set by the user to determine the extent up to which nodes are considered similar or different \cite{lance67general}. In other words, each partition in the collection reflects a different degree of similarity between the nodes, ranging from considering all nodes different, to considering all nodes similar. Examples of hierarchical clustering methods are UPGMA \cite{upgma58}, Ward's method \cite{ward63}, complete \cite{defays76} or single linkage~\cite{carlsson10}. 

However, hierarchical methods still have a major limitation, namely, that nodes are assigned to one and only one category or cluster. There are many situations in which this assignment does not constitute a reasonable approach. Specifically, some nodes might inherently have traits of more than one group and hence have similarities to multiple clusters of nodes that otherwise would be considered dissimilar \cite{youssef09, machan09}. Two classes of methods have been proposed to overcome this issue. The first class corresponds to the so-called soft or fuzzy clustering methods that allow the allocation of a node to multiple clusters by assigning to each node a membership degree or probability of belonging to every cluster \cite{bezdek81}. This can also be done in a hierarchical fashion \cite{yu05, bordogna12}. Nevertheless, these methods subscribe to the idea that nodes have more affinity to either one or another subset, as illustrated by their membership degree; or that they belong to only one group but the setting of the problem is not rich enough to determine to which one, thus assigning a probability of belonging to different subsets. Fundamentally, these methods do not contemplate the idea that a node can intrinsically be part in equal terms of more than one cluster. The second class of methods that attempt to solve this issue is the class of \emph{overlapping} clustering methods, which perform a non-hierarchical deterministic assignment of nodes to more than one subset \cite{cleuziou08}. While overlapping methods enable nodes to belong to more than one cluster, they are still myopic -- just like traditional clustering methods -- to multiple affinity resolutions.

In this paper, a \emph{hierarchical overlapping} clustering method is proposed. Our method outputs a collection of groupings where each of them corresponds to different resolutions of similarity between the nodes (hierarchical), while allowing nodes in each grouping to deterministically belong to more than one cluster (overlapping). Essentially, the method is derived by generalizing the concepts of ultrametrics, equivalence relations, and dendrograms -- typical of hierarchical (non-overlapping) clustering methods -- to those of cut metrics \cite{deza96,culbertson15}, tolerance relations \cite{bartol04}, and nested collections of coverings. More precisely, we generate a family of networks derived from dithering multiple times the dissimilarity function of a given network. We then proceed to apply a hierarchical (non-overlapping) clustering method to each of the dithered networks, resulting in a collection of ultrametrics. From this collection we obtain a cut metric related to the network of interest by means of a convex combination of the ultrametrics. Finally, based on the cut metric we construct a tolerance relation that connects the nodes that are at a distance lower than a pre-specified resolution, which we then use to derive a covering of the node set. By considering all possible resolution levels, we obtain a nested collection of coverings. Hierarchical overlapping clustering methods have been previously developed in the context of community detection, but only for undirected, unweighted networks \cite{lancichinetti09,shen09}.

In Section~\ref{sec:prelim}, the concepts of partition, equivalence relation, and dendrogram are defined, and their relations to both traditional and hierarchical (non-overlapping) clustering methods are established. These concepts are then respectively generalized to those of covering, tolerance relation, and nested collection of coverings. This generalization is achieved by dropping the requirements that prevent a node from belonging to more than one group. Moreover, a formal definition of a hierarchical overlapping clustering method is presented. Section~\ref{sec:cut-metrics} relates the introduced concepts with {concrete} finite metric constructions. More specifically, ultrametrics are defined and directly linked to hierarchical clustering methods (Theorem~\ref{theo:ultrametric-dendrogram}). We then introduce cut metrics and explain how they can be used to obtain a nested collection of coverings (Theorem~\ref{theo:cut-metric-nested-covering}). The proposed algorithm to obtain cut metrics from a given network (Algorithm~\ref{a:overlapping-clustering}), and hence nested collection of coverings, is detailed in Section~\ref{sec:algorithm}. The algorithm is grounded on the fact that a convex combination of ultrametrics yields a cut metric (Proposition~\ref{prop:um-cut-metric}) and the concept of dithering \cite{schuchman64}, that can be used to obtain multiple noisy versions of the given network. Also, quasi-clustering methods that allow for a hierarchical structure and overlapping nodes are explored in Section~\ref{sec:quasi}. These methods are particularly useful when there is a need to preserve in the grouping structure the asymmetric nature of directed networks. Section~\ref{sec:syn_exs} illustrates the implementation of the proposed method in synthetic {experiments} to show that the hierarchical overlapping clustering algorithm yields sensitive results that correspond to intuition. Finally, in Section~\ref{sec:apps}, the proposed method is applied to the unsupervised classification problem of determining handwritten digits \cite{mnist} and the authorship attribution problem of identifying plays that have been co-authored \cite{segarra15}. Concluding remarks in Section~\ref{sec:conclusions} close the paper.


\section{Preliminaries} \label{sec:prelim}

Let $N =(X ,A_{X })$ be a network defined by a finite set of nodes $X $ and a dissimilarity function $A_{X }: X  \times X  \to 
\reals_{+}$ that measures how different two nodes are. This function satisfies that, for any two nodes $x,x' \in X $, $A_{X }(x,x') \ge 0$ with $A_{X }(x,x')=0$ if and only if $x=x'$. 
$A_X$ is a dissimilarity function in the sense that the greater the value of $A_{X }$, the more different two nodes are. Observe that the dissimilarity function need not satisfy the triangle inequality nor be symmetric, that is, we may have that $A_{X }(x,x') \ne A_{X }(x',x)$. This is particularly useful when modeling uneven levels of influence among nodes. We denote by $\ccalN$ the set of all possible networks.

\begin{definition}
		\label{def:partition}
A partition $P_{X }=\{B_{1},\ldots,B_{m}\}$ is a collection of subsets of $X $ such that $\cup_{i=1}^{m} B_{i}=X $ and $B_{i} \cap B_{j}=\emptyset$ for $i \ne j$, $i,j=1,\ldots,m$. Denote by $\ccalP$ the set of all possible partitions.
\end{definition}
\begin{definition}
		\label{def:equivalence-relation}
An equivalence relation $\sim$ is a binary set relation that, for any $x,x',x'' \in X $, satisfies the properties of reflexivity ($x \sim x$), symmetry ($x \sim x'$ if and only if $x' \sim x$), and transitivity (if $x \sim x'$ and $x' \sim x''$, then $x \sim x''$).
\end{definition}
\noindent Observe that partitions (Def. \ref{def:partition}) uniquely define equivalence relations (Def. \ref{def:equivalence-relation}) as follows: $x \sim x'$ if and only if $x,x' \in B_{i}$ for some $B_{i} \in P_{X }$, $i=1,\ldots,m$. Equivalently, equivalence relations uniquely define partitions \cite{halmos60}. 
An equivalence relation obtained from a specific partition $P_X$ will be denoted by~$\sim_{P_{X }}$.

With these definitions in place, a (traditional) clustering method $\ccalG$ can be defined as a structure-preserving map from the set of networks to the set of partitions, $\ccalG: \ccalN \to \ccalP$. This map is structure preserving in the sense that the output partition is defined over the same node set $X$ specific to the input network. That is, $\ccalG(N)=P_{X }=\{B_{1},\ldots,B_{m}\}$ for $N=(X,A_{X })$ with $\cup_{i=1}^{m} B_{i}=X $. Intuitively, a desirable clustering method $\ccalG$ is one in which the nodes contained in each subset $B_{i}$ are determined by the dissimilarity function $A_{X }$ so that similar nodes are grouped together.

In many situations, having only one partition as the output of a clustering method may not be appropriate. It might be desirable to output several partitions that depend on a resolution parameter that specifies the affinity required between two nodes to be deemed as similar.
\begin{definition}
		\label{def:dendrogram}
A dendrogram $D_{X }=\{D_{X }(\delta), \delta \ge 0\}$ is a collection of partitions, where $\delta$ is a resolution parameter and $D_{X }(\delta)$ is a partition of the node set $X$, satisfying the following two properties:
\begin{enumerate}[(i)]
\item $D_{X }(0)=\{ \{x\}, x \in X \}$, and there exists a $\delta_{\max}$ such that for all $\delta \ge \delta_{\max}$, $D_{X }(\delta)= \{ X \}$;
\item it is a nested collection of partitions, that is, for $\delta \le \delta'$, if $x \sim_{\delta} x'$, then $x \sim_{\delta'} x'$ for all $x, x' \in X$;
\end{enumerate}
together with a technical condition for right-continuity \cite{carlsson14}.
\end{definition}
\noindent Note that, in an effort to ease the exposition, $x \sim_{\delta} x'$ is being used as a shorthand notation for $x \sim_{D_{X}(\delta)} x'$. Observe that property (i) encodes the fact that at resolution $\delta=0$ all nodes are considered to be different from each other, thus each node forms a singleton cluster, whereas for sufficiently large $\delta$ all nodes are considered to be similar and, hence, clustered in a single block. Property (ii) enforces the agglomeration to occur in a nested fashion in the sense that if two nodes belong to the same cluster at a specific resolution then they will remain co-clustered for any larger resolution.
Denoting by $\ccalD$ the set of all possible dendrograms, we define a hierarchical clustering method $\ccalH$ as follows.
\begin{definition}
		\label{def:hierarchical-clustering}
A hierarchical clustering method $\ccalH$ is a structure-preserving map from the set of networks to the set of dendrograms
\begin{equation}
\ccalH : \ccalN \to \ccalD.
		\label{eq:hierarchical-clustering}
\end{equation}
\end{definition}

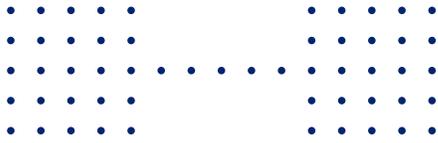
\begin{figure}[!t]
	\centering

\def \thisplotscale {0.5}
\def \unit {\thisplotscale cm}

\tikzstyle{dot} = [ellipse, 
                   inner sep=0pt, 
                   fill=pennblue, 
                   anchor = center,
                   minimum height = 0.2*\unit, 
                   minimum width  = 0.2*\unit]

{\footnotesize\begin{tikzpicture}[x = 1*\unit, y=1*\unit, font=\footnotesize]

   \foreach \i in {-2,...,2}{
      \foreach \j in {-2,...,2}{   
         \path (0.8*\i,0.8*\j) node [dot] {};
      }
   }
   \foreach \i in {-2,...,2}{  
         \path (4,0) ++ (0.8*\i,0) node [dot] {};
   }

   \foreach \i in {-2,...,2}{
      \foreach \j in {-2,...,2}{   
         \path (8,0) ++ (0.8*\i,0.8*\j) node [dot] {};
      }
   }

\end{tikzpicture}}
	\caption{Dumbbell network with dissimilarities given by the Euclidean distance between points.}
	\label{fig:dumbbell-nodes}
\end{figure}
A limitation of both traditional $\ccalG$ and hierarchical $\ccalH$ clustering methods is that each node is required to belong to one and only one cluster at any given resolution. This requirement stems from the nonintersecting nature of the blocks that form a partition. In many problems, however, some nodes may inherently have traits of more than group, making it reasonable for them to belong to more than one cluster. For instance, consider the dumbbell network in Fig.~\ref{fig:dumbbell-nodes} where the dissimilarity between two nodes is given by their Euclidean distance. It is intuitive that the point clouds on each side should form separate clusters. However, it is unclear if, e.g., {\it all} of the handle (bridge) should be a separate cluster as it is not unreasonable to assign its borders to the respective point clouds. 

Throughout the paper, we develop a clustering method that enables both deterministic overlapping of nodes and hierarchical outcomes. This is achieved by generalizing the concepts of equivalence relations, partitions, and dendrograms to those of tolerance relations, coverings, and nested collections of coverings, respectively; see Table \ref{table:parallelism}.

\subsection{Generalizing concepts: Coverings and tolerance relations}

The first step in order to allow a node to belong to more than one cluster is to drop the nonintersecting requirement of partitions.

\begin{definition}
		\label{def:covering}
A covering $Q_{X }=\{C_{1},\ldots,C_{m}\}$ of the node set $X $ is a collection of subsets such that $\cup_{i=1}^{m} C_{i}=X $, but not necessarily $C_{i} \cap C_{j} = \emptyset$ for $i \ne j$. Denote by $\ccalQ$ the set of all possible coverings.
\end{definition}
\noindent Note that partitions (Def. \ref{def:partition}) are particular cases of coverings (Def. \ref{def:covering}), i.e., $\ccalP \subset \ccalQ$. Furthermore, a non-hierarchical overlapping clustering method $\ccalM$ is defined as a structure-preserving map from the set of networks to the set of coverings $\ccalM: \ccalN \to \ccalQ$. Related to the concept of covering, we introduce the notion of a tolerance relation.
\begin{definition}
		\label{def:tolerance-relation}
A tolerance relation $\tol$ is a binary set relation that, for any $x,x' \in X $, satisfies the properties of reflexivity ($x \tol x$), and symmetry ($x \tol x'$ if and only if $x' \tol x$), but is not necessarily transitive.
\end{definition}
\noindent Observe that a tolerance relation is a generalization of an equivalence relation (Def. \ref{def:equivalence-relation}) obtained by dropping transitivity. A tolerance relation defined on a node set $X $ can be represented by an unweighted and undirected graph where an edge exists between nodes $x$ and $x'$ if and only if $x \tol x'$. There are two established ways of inducing a covering from a tolerance relation {based on this graph \cite{bartol04}. The first methodology is referred to as covering by classes, wherein each cover consists of the neighbor set of each node, together with the node itself, and where all covers that are completely contained within a larger one are removed.} The other established procedure is denominated covering by blocks, where each block of the covering is given by the maximal cliques of the graph. For the remainder of the paper, we adopt this latter procedure as the canonical way of inducing coverings from tolerance relations. {Note that,} unlike the case of equivalence relations and partitions, while any tolerance relation can induce a covering (by blocks) not all coverings of a node set $X $ can be induced by a tolerance relation \cite[Theorem 2]{bartol04}. Leveraging the connection between tolerance relations and coverings, we extend the notion of a dendrogram into the realm of overlapping clusters.

\begin{table}[tb]
	\caption{Parallelism between Hierarchical clustering $\ccalH$ and Hierarchical Overlapping clustering $\ccalO$.}
	\label{table:parallelism}
	\centering
	\begin{tabular}{l|ll}
	\textbf{Method}	& Hierarchical $\ccalH$	& Overlapping $\ccalO$ \\
			& (Def. \ref{def:hierarchical-clustering})	& (Def. \ref{def:overlapping-clustering}) \\ \hline
	\textbf{Metric}	& Ultrametric $u_{X}(x,x')$	& Cut Metric  $c_{X }(x,x')$\\
			& (Def. \ref{def:ultrametric})	& (Def. \ref{def:cut-metric}) \\
	\textbf{Relation}	& Equivalence $\sim$	& Tolerance $\tol$ \\
			& (Def. \ref{def:equivalence-relation})	& (Def. \ref{def:tolerance-relation}) \\
	\textbf{Grouping}	& Partition $P_{X }=\{B_{i}\}$	& Covering $Q_{X }=\{C_{i}\}$ \\
			& (Def. \ref{def:partition})	& (Def. \ref{def:covering}) \\
	\textbf{Hierarchy}	& Dendrogram $D_{X }$	& Nested Covering $K_{X }$ \\
			& (Def. \ref{def:dendrogram})	& (Def. \ref{def:nested-coverings}) \\ \hline
	\textbf{Construction}& Theorem \ref{theo:ultrametric-dendrogram} & Theorem \ref{theo:cut-metric-nested-covering}
	\end{tabular}
\end{table}

\begin{definition}
		\label{def:nested-coverings}
A nested covering $K_{X }=\{K_{X }(\delta), \delta \ge 0\}$ is a collection of coverings, where $\delta$ is a resolution parameter and $K_{X }(\delta)$ is a covering induced by a tolerance relation $\tol_{\delta}$, satisfying the following properties:
\begin{enumerate}[(i)] 
\item $K_{X }(0)=\{ \{x\}, x \in X \}$ and there exists a $\delta_{\max}$ such that for all $\delta \ge \delta_{\max}$, $K_{X }(\delta)= \{ X\} $;
\item For $\delta \leq \delta'$, if $x \tol_{\delta} x'$, then $x \tol_{\delta'} x'$ for all $x,x' \in X $.
\end{enumerate}
\end{definition}

\noindent Notice the resemblance between properties (i) and (ii) above and those in Def.~\ref{def:dendrogram}. In this way, we extend the notion of nestedness from its intuitive definition for (non-overlapping) clusters to the setting of overlapping clusters.
Denoting by $\ccalK$ the set of all nested coverings, we define hierarchical overlapping clustering methods as follows.
\begin{definition}
		\label{def:overlapping-clustering}
A hierarchical overlapping clustering method $\ccalO$ is a structure-preserving map from the set of networks to the set of nested coverings
\begin{equation}
\ccalO: \ccalN \to \ccalK.
		\label{eq:overlapping-clustering}
\end{equation}
\end{definition}
Given a network $N=(X, A_X)$, a clustering method $\ccalO$ outputs a structure $\ccalO(N) = K_X$ which takes into account different resolution levels of similarity and, at each resolution, allows overlapping clusters. Building on metric concepts developed in Section~\ref{sec:cut-metrics}, {we present the practical construction of one such method $\ccalO$ in Section~\ref{sec:algorithm}}.


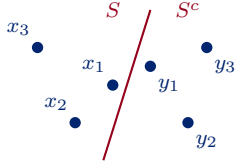
\begin{figure}[!t]
	\centering

\def \thisplotscale {0.5}
\def \unit {\thisplotscale cm}

\tikzstyle{dot} = [ellipse, 
                   inner sep=0pt, 
                   fill=pennblue, 
                   anchor = center,
                   minimum height = 0.3*\unit, 
                   minimum width  = 0.3*\unit]

{\footnotesize\begin{tikzpicture}[x = 1*\unit, y=1*\unit, font=\footnotesize]



	\path (0,0) node [dot] {}; \node [color=pennblue]at (-0.5,0.5) {$x_{1}$};
	\path (-1,-1) node [dot] {}; \node [color=pennblue]at (-1.5,-0.5) {$x_{2}$};
	\path (-2,1) node [dot] {}; \node [color=pennblue]at (-2.5,1.5) {$x_{3}$};
	\path (1,0.5) node [dot] {}; \node [color=pennblue]at (1.5,0) {$y_{1}$};
	\path (2,-1) node [dot] {}; \node [color=pennblue] at (2.5,-1.5) {$y_{2}$};
	\path (2.5,1) node [dot] {}; \node [color=pennblue] at (3,0.5) {$y_{3}$};
	\draw [thick,color=pennred] (-0.25,-2)--(1,2);
	
	\node[color=pennred] at (0,2) {$S$}; \node[color=pennred] at (2,2) {$S^{c}$};
	
\end{tikzpicture}}
	\vspace{0.3cm}
	\caption{A cut semimetric on the node set $X =\{x_{1},x_{2},x_{3},y_{1},y_{2},y_{3}\}$. A cut that partitions the space into $S$ and $S^{c}$ is drawn. The resulting cut semimetric is $\delta_{S}(x_{i},x_{j})=\delta_{S}(y_{i},y_{j})=0$ and $\delta_{S}(x_{i},y_{j})=1$ for all $i,j=1,2,3$.}
	\label{fig:cut-semimetric}
\end{figure}

\section{Nested coverings derived from cut metrics}\label{sec:cut-metrics}

Hierarchical clustering methods output dendrograms, which are nested collections of partitions. However, dendrograms are cumbersome to handle mathematically so their equivalence to ultrametric spaces proves to be useful \cite{carlsson10}.
\begin{definition}
		\label{def:ultrametric}
An ultrametric is a function $u_{X}: X  \times X  \to \reals_{+}$ that, for any $x,x',x'' \in X $ satisfies 
\begin{enumerate}[(i)]
\item $u_{X}(x,x') \ge 0$ and $u_{X}(x,x')=0$ if and only if $x=x'$;
\item $u_{X}(x,x')=u_{X}(x',x)$;
\item $u_{X}(x,x'') \le \max\{u_{X}(x,x'),u_{X}(x',x'')\}$.
\end{enumerate}
\end{definition}
\noindent Ultrametrics can be interpreted as metrics that satisfy a more stringent version of the triangle inequality [Property (iii)] where the addition in the regular triangle inequality is replaced by a maximization. Dendrograms can be equivalently represented as ultrametrics, as shown in the following theorem. 

\begin{theorem}
		\label{theo:ultrametric-dendrogram}
Given an ultrametric $u_{X}(x,x')$ defined over the node set $X $, define a collection of equivalence relations as follows
\begin{equation}
u_{X}(x,x') \le \delta \Longleftrightarrow x \sim_{\delta} x'.
		\label{eq:um-dendrogram}
\end{equation}
The resulting collection of equivalence relations is such that it generates a dendrogram $D_{X }$.

Likewise, given a dendrogram $D_{X }$ that uniquely determines a collection of equivalence relations $\sim_{\delta}$ defined on the node set $X $, an ultrametric $u_{X}$ can be obtained as
\begin{equation}
u_{X}(x,x'):=\min\{\delta:x \sim_{\delta} x'\}.
		\label{eq:dendrogram-um}
\end{equation}
\end{theorem}
\begin{myproof}
See \cite[Theorem 1]{carlsson14}.
\end{myproof}

\noindent Theorem~\ref{theo:ultrametric-dendrogram} establishes a structure-preserving bijection between ultrametrics and dendrograms. Consequently, we may reinterpret hierarchical clustering methods as structure-preserving maps from the set of networks to the set $\ccalU$ of all possible ultrametrics, i.e., $\ccalH: \ccalN \to \ccalU$ (cf. Def.~\ref{def:hierarchical-clustering}).

With the objective of deriving a relation between nested coverings $K_X$ and a metric structure akin to that between dendrograms and ultrametrics, the notion of cut metrics is introduced.
Let $S$ be a subset of a given node set $X $ and define a cut semimetric $\delta_{S}(x,x')$ as
\begin{equation}
\delta_{S}(x,x') = \ind{S   \cap \{x,x'\}\neq\emptyset}\ind{S^{c} \cap \{x,x'\}\neq \emptyset}
		\label{eq:cut-semimetric}
\end{equation}
where $\ind{}$ is the indicator function. Essentially, the cut semimetric partitions the node set into two blocks $S$ and $S^c$ and assigns a dissimilarity of $1$ to nodes that are in different sides of this partition, and $0$ to nodes that fall in the same side; see Fig. \ref{fig:cut-semimetric}.

This cut semimetric can also be understood as being induced by a binary classifier which partitions the node set in two and assigns a unit dissimilarity to nodes belonging to different categories. If multiple cuts, or partitions, are considered, then a cut metric $c_{X }$ can be obtained as the combination of all the associated dissimilarities; see Fig. \ref{fig:cut-metric-figure}.
\begin{definition}
		\label{def:cut-metric}
A cut metric $c_{X }:X  \times X  \to \reals_{+}$ is a metric that can be written as a nonnegative linear combination of cut semimetrics
\begin{equation}
c_{X }(x,x')=\sum_{S \subseteq X } \lambda_{S}\delta_{S}(x,x')
		\label{eq:cut-metric}
\end{equation}
with $\lambda_{S} \ge 0$ and where the sum ranges over all possible subsets of $X $.
\end{definition}
\noindent Note that $\lambda_{S'}=0$ for some $S'$ amounts to the associated cut not being considered in the construction of $c_{X }$. It follows from Def.~\ref{def:cut-metric} that $c_X$ is a metric, thus, in particular it satisfies that for any $x,x' \in X $, $c_{X }(x,x')=c_{X }(x',x)\ge 0$ (symmetry and nonnegativity) and $c_{X }(x,x')=0$ if and only if $x=x'$ (identity). We denote by $\ccalV$ the set of cut metrics.
\begin{figure}[!t]
	\centering

\def \thisplotscale {0.5}
\def \unit {\thisplotscale cm}

\tikzstyle{dot} = [ellipse, 
                   inner sep=0pt, 
                   fill=pennblue, 
                   anchor = center,
                   minimum height = 0.3*\unit, 
                   minimum width  = 0.3*\unit]

{\footnotesize\begin{tikzpicture}[x = 1*\unit, y=1*\unit, font=\footnotesize]



	\path (0,0) node [dot] {}; \node [pennblue] at (-0.5,0.5) {$x_{1}$};
	\path (-1,-1) node [dot] {}; \node[pennblue] at (-1.5,-0.6) {$x_{2}$};
	\path (-2,1) node [dot] {}; \node [pennblue]at (-2.5,1.5) {$x_{3}$};
	\path (1,0.5) node [dot] {}; \node [pennblue]at (1.5,0) {$y_{1}$};
	\path (2,-1) node [dot] {}; \node[pennblue] at (2.5,-1.5) {$y_{2}$};
	\path (2.5,1) node [dot] {}; \node[pennblue] at (3,0.5) {$y_{3}$};
	\draw [thick,pennred] (-0.25,-2)--(1,2); \node [pennred] at (0.25,2) {$S_{1}$}; \node [pennred] at (1.75,1.5) {$S_{1}^{c}$}; \node [pennred] at (1.25,2.5) {$\lambda_{S_{1}}$};
	\draw [thick,pennyellow] (-5,0)--(5,-0.75); \node [pennyellow] at (4.5,-1.25) {$S_{2}$}; \node[pennyellow]at (4.5,0) {$S_{2}^{c}$}; \node [pennyellow] at (5.75,-0.75) {$\lambda_{S_{2}}$};
	\draw [thick, penngreen] (-5,-1) to[in=90,out=45] (3.5,-1.5); \node [penngreen] at (-5.5,-0.5) {$S_{3}$}; \node [penngreen] at (-3.8,-0.8) {$S_{3}^{c}$}; \node [penngreen] at (-5.25,-1.25) {$\lambda_{S_{3}}$};
		
\end{tikzpicture}}
	\caption{Construction of a cut metric $c_{X }$. Consider the node set $X =\{x_{1},x_{2},x_{3},y_{1},y_{2},y_{3}\}$. The first cut separates $\{x_{i}\}$ from $\{y_{i}\}$, $i=1,2,3$ and is assigned a dissimilarity $\lambda_{S_{1}}$. The cuts associated with $S_{i}$ for $i=2, 3$ separate $\{x_{i},y_{i}\}$ from the rest of $X$ and are assigned dissimilarity $\lambda_{S_{i}}$. Finally, the cut metric is obtained as $c_{X }=\lambda_{S_{1}}\delta_{S_{1}}+\lambda_{S_{2}}\delta_{S_{2}}+\lambda_{S_{3}}\delta_{S_{3}}$.}
	\label{fig:cut-metric-figure}
\end{figure}
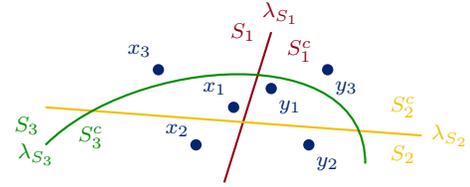

In the same way that ultrametrics induce dendrograms (Theorem \ref{theo:ultrametric-dendrogram}), cut metrics induce nested collections of coverings, as shown next.

\begin{theorem}
		\label{theo:cut-metric-nested-covering}
Let $X $ be a set of nodes and let $c_{X }$ be a cut metric defined on $X$. Let $K_{X }=\{K_{X }(\delta), \delta \ge 0\}$ be a collection of coverings. If, for each $\delta \ge 0$, the corresponding covering $K_{X }(\delta)$ is obtained from the tolerance relation given by
\begin{equation}
c_{X }(x,x') \le \delta \Longrightarrow x \tol_{\delta} x'
		\label{eq:cut-metric-tol}
\end{equation}
then $K_{X }$ is a nested collection of coverings as defined in Def. \ref{def:nested-coverings}.
\end{theorem}
\begin{myproof}
Observe that \eqref{eq:cut-metric-tol} defines a valid tolerance relation since $c_{X }(x,x)=0$ and $c_{X }(x,x')=c_{X }(x',x)$ imply reflexivity and symmetry of $\tol_{\delta}$, respectively. Using the relation $\tol_{\delta}$ in \eqref{eq:cut-metric-tol} to define a collection of coverings by blocks $K_{X }(\delta)$ \cite{bartol04}, it can be shown that $K_X$ satisfies properties (i) and (ii) in Def.~\ref{def:nested-coverings}. To be more specific, $K_{X }(0)=\{ \{x \}, x \in X \}$ follows from the fact that $c_{X }(x,x')=0$ if and only if $x=x'$ and, since there exists a $\delta_{\max}$ such that $c_{X}(x,x') \le \delta_{\max}$ for all $x,x' \in X$, it follows that $K_{X }(\delta)= \{X \}$ for all $\delta \geq \delta_{\max}$.
Finally, observe that $K_{X }$ is nested since, for any $\delta<\delta'$, $c_{X }(x,x') \le \delta < \delta'$ implies that $x \tol_{\delta} x' \Rightarrow x \tol_{\delta'} x'$.
\end{myproof}

\noindent Notice that \eqref{eq:cut-metric-tol} is similar to the connection between ultrametrics and equivalence relations in \eqref{eq:um-dendrogram} but in a unidirectional fashion; see Remark~\ref{R:difference_theorems}. For $0 < \delta < \delta_{\max}$ the resulting tolerance relations define coverings in such a way that nodes might belong to more than one cluster. If there exists at least one such node, it is said that there is \emph{overlap}. Consequently, we are interested in computing the number of overlapping nodes for each value of $\delta$.
\begin{definition}
Let $K_{X}=\{K_{X}(\delta), \delta \ge 0\}$ be a nested collection of coverings defined over node set $X$ and where each covering $K_{X}(\delta)$ is comprised of $m(\delta)$ blocks or covers, i.e. $K_{X}(\delta)=\{C_{1},\ldots,C_{m(\delta)}\}$. Then, the overlapping function $\fol:\reals_{+} \to \mathbb{Z}_{+}$ is given by
\begin{equation}
\fol(\delta)=\sum_{k=1}^{n} \mathbb{I}  \left\{ 
	x_{k} \in C_{i} \cap C_{j} , \ i \ne j, \ i,j=1,\ldots,m(\delta)\right\}.
		\label{eq:overlap-function}
\end{equation}
\end{definition}

\noindent The function $\fol(\delta)$ counts the nodes that are overlapping at each value of $\delta$, i.e. that belong to more than one block of the covering $K_{X }(\delta)$. From property (i) in Def.~\ref{def:nested-coverings}, it follows that $\fol(0)=0$ and $\fol(\delta)=0$ for all $\delta \ge \delta_{\max}$. Moreover, we use the overlapping function to define the notion of clusterability as follows.

\begin{definition}\label{def:clusterability}
Given a network $N =(X ,A_{X })$, a cut metric $c_X$, and an associated nested covering $K_X$, we say that $(X,c_{X })$ is \emph{clusterable} if there exists a $\delta$ such that $\fol(\delta)=0$, $K_{X }(\delta) \neq \{ \{x \}, x \in X \}$, and $K_{X }(\delta) \ne \{X\} $.
\end{definition}
\noindent In other words, a network is clusterable if for some resolution $\delta$ we obtain a valid partition (i.e., there is no overlap) that is different from those partitions obtained for extreme resolutions ($\delta = 0$ where all nodes are clustered separately, and $\delta \geq \delta_{\max}$ where all nodes are clustered together).
Even for non-clusterable networks, the overlapping function provides valuable information about the underlying grouping structure and may help to identify  nodes that cannot be fully classified into one subset. Also, the overlapping function is useful in determining meaningful values of the resolution $\delta$; see Sections~\ref{sec:syn_exs} and \ref{sec:apps} for additional details.

To sum up, in the same way that in hierarchical (non-overlapping) clustering, ultrametrics are used to define equivalence relations that determine dendrograms, in the proposed hierarchical overlapping clustering method, cut metrics are used to define tolerance relations that determine nested coverings; see Table~\ref{table:parallelism}.
\begin{remark}\label{R:difference_theorems} \normalfont
Recall that defining a partition on a node set is equivalent to defining an equivalence relation on the same set (cf. Defs.~\ref{def:partition} and \ref{def:equivalence-relation}). This means that specifying any of them uniquely leads to the other. The same occurs for the definitions of dendrograms and ultrametrics; see Theorem~\ref{theo:ultrametric-dendrogram}. However, in the context of hierarchical overlapping clustering these bijections no longer hold. 
While it is true that given a cut metric, we can obtain a collection of tolerance relations that leads to a nested collection of coverings (see Theorem~\ref{theo:cut-metric-nested-covering}), we cannot, in general, uniquely determine a tolerance relation from a given covering \cite[Theorem 2]{bartol04} or a cut metric from a given nested collection of coverings.
\end{remark}

\begin{remark}\label{R:different_metrics} \normalfont
In a strict sense, the proof of Theorem~\ref{theo:cut-metric-nested-covering} does not require $c_{X}$ to be a cut metric. In fact, any nonnegative, symmetric function satisfying the identity property would suffice to create a relationship between nodes as in \eqref{eq:cut-metric-tol}. {Constructing such a relation is known as the Rips complex \cite{hausmann95}}. However, since we ultimately want to group nodes that are more similar to each other than to the rest, we are looking for functions that actually reflect the relationship between nodes encoded in the dissimilarity function. The cut metric obtained as a convex combination of ultrametrics achieves this. {Moreover, due to its construction (cf. Figs.~\ref{fig:cut-semimetric} and \ref{fig:cut-metric-figure}), cut metrics are intrinsically related to the formation of coverings and partitions.}
Another recent alternative to use in \eqref{eq:cut-metric-tol} is projections on A-spaces as done in \cite{culbertson16}. It is also worth pointing out that the process of generating a Rips complex as in \eqref{eq:cut-metric-tol} involves the computation of the maximal cliques of a graph which is a computationally intensive operation; however, there exists algorithms to carry this operation more efficiently by means of parallelization \cite{bronkerbosch73}.
\end{remark}


\section{Hierarchical overlapping clustering algorithm}
	\label{sec:algorithm}

Given a network, our goal is to output a nested collection of coverings such that the grouping of nodes is done in accordance to the difference between them at varying levels of resolution. In virtue of Theorem~\ref{theo:cut-metric-nested-covering}, a cut metric can be used to induce a nested collection of coverings. Hence, the idea is to systematically construct a cut metric in such a way that it reflects the dissimilarity across nodes as given by the dissimilarity function $A_{X }$. The following proposition states a key relation between ultrametrics and cut metrics.
\begin{proposition}
		\label{prop:um-cut-metric}
A convex combination of ultrametrics yields a cut metric. 
\end{proposition}
\begin{myproof}
Given $(X , d_{X })$ where $d_{X }$ is a metric, then $d_{X }$ is in particular a tree metric -- there exists a tree graph in which it is possible to embed the distance on its edges \cite[p. 147]{deza96} -- if and only if it satisfies the four-point condition
\begin{align}
d_{X }(x,x')+d_{X }(x'',x''') \le \max \big\{ 
	& d_{X }(x,x'')+d_{X }(x',x'''),
		\nonumber \\
	& d_{X }(x,x''')+d_{X }(x',x'') \big\}
		\label{eq:4PC}
\end{align}
for all $x,x',x'',x''' \in X $. {Ultrametrics also satisfy the four-point condition so that they are particular cases of tree metrics \cite[p. 311]{deza96}. It can be shown that tree metrics, as well as convex combinations of tree metrics, are $\ell_{1}$-embeddable \cite[Proposition 11.1.4, Fact 11.1.15]{deza96}. Finally, since a metric is $\ell_{1}$-embeddable if and only if it is a cut metric \cite[Proposition 4.4.2]{deza96}, a convex combination of ultrametrics yields a cut metric.}
\end{myproof}

\noindent From Proposition~\ref{prop:um-cut-metric} it immediately follows that ultrametrics are particular cases of cut metrics. Intuitively, this feature is inherited from the facts that equivalence relations are particular cases of tolerance relations and that partitions are particular cases of coverings; see Table~\ref{table:parallelism}. Also, Proposition~\ref{prop:um-cut-metric} plays a key role in the generation of cut metrics as it gives a systematic way to obtain cut metrics from ultrametrics, which are readily available from the application of hierarchical (non-overlapping) clustering methods $\ccalH$.

\begin{algorithm}[t]
 	\caption{Hierarchical overlapping clustering algorithm~$\ccalO$.}
	\label{a:overlapping-clustering}
	
	\begin{algorithmic}[1]
 \Statex \textbf{Input:} $J$: no. of perturbations, $N$: network,
 \Statex $\qquad \ccalH$: hierarchical clustering method, perturbation$(\cdot)$
 \Statex \textbf{Output:} $K_{X}$: nested collection of coverings.
 \Statex
 \Procedure{$\ccalO$}{$J,N,\ccalH,\textrm{perturbation}$}
 \For{i\,=\,1:$J$}
   \State $\tilde{N }$\, =\, perturbation($N $);
   \State $\tilde{u}_{X}^{(i)}$\,=\,$\ccalH(\tilde{N })$;
 \EndFor
\State $c_{X }$\,=\,avg($\tilde{u}_{X}^{(1)},\ldots, \tilde{u}_{X}^{(J)}$);
\State $K_{X }$\,=\,obtain\_coverings($c_{X }$);
 \EndProcedure
 	\end{algorithmic}
\end{algorithm}

The next step is to obtain multiple ultrametrics related to the dissimilarity function $A_{X }$ so that the resulting combination of these yields a cut metric that bears a relation with $A_{X }$. In order to achieve this, the concept of dithering \cite{schuchman64} proves to be useful, where we intentionally apply random noise to the dissimilarity function $A_{X }$ to obtain $\tilde{A}_{X }$. Denoting by $\tilde{N }=(X ,\tilde{A}_{X })$ the resulting perturbed network, and by $\tilde{u}_{X}= \ccalH(\tilde{N })$ the ultrametric obtained by hierarchically clustering $\tilde{N }$, the following corollary of Proposition~\ref{prop:um-cut-metric} {results}.

\begin{corollary}
The function $c_{X }$ obtained as
\begin{equation}
c_{X }=\mbE \left[ \tilde{u}_{X} \right]
\end{equation}
is a cut metric, {where the expectation is taken with respect to the probability distribution of the randomization introduced by dithering.}
\end{corollary}

In practice, several realizations of random perturbations are implemented, generating a whole family of closely related networks. Each one of these networks yields a different ultrametric when a pre-specified hierarchical clustering method $\ccalH$ is applied to them. When all these ultrametrics are averaged, a cut metric is obtained (cf. Proposition~\ref{prop:um-cut-metric}). Formally, given a network $N = (X, A_X)$ denote by $\tilde{N}_{i}=(X ,\tdA_{X }^{(i)})$ for $i=1,\ldots,J$ the $J$ networks resulting from independently perturbing the dissimilarity function $J$ times. Moreover, denoting by $\tilde{u}_{X}^{(i)}$ the ultrametric output of hierarchically clustering $\tilde{N}_{i}$, we generate the cut metric $c_X$ given by (cf.~Proposition~\ref{prop:um-cut-metric})
\begin{equation}
c_{X }(x,x') = \frac{1}{J} \sum_{i=1}^{J} \tilde{u}_{X}^{(i)}(x,x').
		\label{eq:cut-metric-avg}
\end{equation}
Finally, we use $c_X$ to obtain an associated nested collection of coverings $K_X$ (cf. Theorem~\ref{theo:cut-metric-nested-covering}). The algorithm is summarized in Algorithm~\ref{a:overlapping-clustering}.

\begin{remark}\normalfont \label{rmk:linear-convex}
Proposition~\ref{prop:um-cut-metric} also holds for non-negative linear combinations of ultrametrics, but only a convex combination guarantees that the distance scale described by the ultrametrics is preserved in the resulting cut metric. Therefore, if we want to interpret the scaling parameter $\delta$ in the same units as the ultrametric function, then we should proceed with a convex combination. This is useful for directly comparing how clusters are formed in hierarchical methods, both overlapping and non-overlapping.
\end{remark}

\begin{remark}\normalfont
Observe that cut metrics conform a convex cone and hence, any non-negative combination of cut metrics is still a cut metric. This is not true in the case of ultrametrics: a non-negative linear combination of ultrametrics is not an ultrametric \cite[Section VII-B]{carlsson14}. In fact, it is cut metrics that form the conic hull of ultrametrics.
\end{remark}

\begin{remark}\normalfont
The choice of the hierarchical non-overlapping clustering method $\ccalH$ depends on the specific problem at hand and as such is an input to the proposed algorithm. It is observed in the experiments (Sections~\ref{sec:syn_exs} and \ref{sec:apps}) that for clusterable datasets, Algorithm~\ref{a:overlapping-clustering} outputs the same coverings as those obtained when using $\ccalH$ directly. Furthermore, the number $J$ of perturbations is not a fundamental parameter in the sense that it is only used to create a cut metric. For the experiments considered in this paper, we observed that for values of $J \ge 10$ the obtained output covers do not depend on $J$.
\end{remark}


\section{Quasi-Cut Metrics} \label{sec:quasi}

The output structure of the hierarchical overlapping clustering method $\ccalO$ is a nested collection of coverings $\ccalK$ (cf. Def~\ref{def:overlapping-clustering}). This is a symmetric structure in the sense that if node $x$ shares a block in a covering with node $x'$ then $x'$ must share a block with $x$. This feature is inherited from the fact that cut metrics are symmetric functions. When considering asymmetric (directed) networks, the symmetry in the output necessarily entails a loss of information, which might be an undesirable feature of a grouping algorithm in some applications \cite{slater84,pentney05,zhao05}.

\begin{figure*}[t]
	\captionsetup[subfigure]{justification=centering}
	\centering
	\begin{subfigure}{0.66\columnwidth}
		\centering

\def \thisplotscale {0.42}
\def \unit {\thisplotscale cm}

\tikzstyle{blue vertex} = [ellipse, 
                   inner sep=0pt, 
                   fill=pennblue!30,
                   draw=black,
                   anchor = center,
                   minimum height = 1.75*\unit, 
                   minimum width  = 1.75*\unit]

{\small
\begin{tikzpicture}[-stealth, shorten >=2, scale = \thisplotscale]

    \node [blue vertex] at (1,4) (1) {$x_2$};
    \node [blue vertex] at (5,-1) (2) {$x_3$};    
    \node [blue vertex] at (-3,-1) (3) {$x_1$};

    \path (1) edge [bend left=20, above right] node {$3$} (2);	
    \path (2) edge [bend left=20, below] node {$2$} (3);
    \path (3) edge [bend left=20, above left] node {$1$} (1);    	

    \path (2) edge [bend left=20,left, pos=0.6] node {$1$} (1);	
    \path (3) edge [bend left=20, above] node {$3$} (2);
    \path (1) edge [bend left=20, right, , pos=0.4]  node {$2$} (3);    	
    
    \node [below] at (-3,5) {$\breve{u}_{X}^{(1)}$};
  
\end{tikzpicture}
}
		\caption{Quasi-ultrametric $\qu_{X}^{(1)}$}
		\label{quasi-u1}
	\end{subfigure}
	\hspace{-0.7cm}
	\begin{subfigure}{0.66\columnwidth}
		\centering

\def \thisplotscale {0.42}
\def \unit {\thisplotscale cm}

\tikzstyle{blue vertex} = [ellipse, 
                   inner sep=0pt, 
                   fill=pennblue!30,
                   draw=black,
                   anchor = center,
                   minimum height = 1.75*\unit, 
                   minimum width  = 1.75*\unit]

{\small
\begin{tikzpicture}[-stealth, shorten >=2, scale = \thisplotscale]

    \node [blue vertex] at (1,4) (1) {$x_2$};
    \node [blue vertex] at (5,-1) (2) {$x_3$};    
    \node [blue vertex] at (-3,-1) (3) {$x_1$};

    \path (1) edge [bend left=20, above right] node {$3$} (2);	
    \path (2) edge [bend left=20, below] node {$2$} (3);
    \path (3) edge [bend left=20, above left] node {$1$} (1);    	

    \path (2) edge [bend left=20,left, pos=0.6] node {$1$} (1);	
    \path (3) edge [bend left=20, above] node {$2$} (2);
    \path (1) edge [bend left=20, right, , pos=0.4]  node {$3$} (3);    	
    
    \node [below] at (-3,5) {$\breve{u}_{X}^{(2)}$};
  
\end{tikzpicture}
}
		\caption{Quasi-ultrametric $\qu_{X}^{(2)}$}
		\label{quasi-u2}
	\end{subfigure}
	\hspace{-0.7cm}
	\begin{subfigure}{0.66\columnwidth}
		\centering

\def \thisplotscale {0.42}
\def \unit {\thisplotscale cm}

\tikzstyle{blue vertex} = [ellipse, 
                   inner sep=0pt, 
		   fill=pennblue!30,
                   draw=black,
                   anchor = center,
                   minimum height = 1.75*\unit, 
                   minimum width  = 1.75*\unit]

{\small
\begin{tikzpicture}[-stealth, shorten >=2, scale = \thisplotscale]

    \node [blue vertex] at (1,4) (1) {$x_2$};
    \node [blue vertex] at (5,-1) (2) {$x_3$};    
    \node [blue vertex] at (-3,-1) (3) {$x_1$};

    \path (1) edge [bend left=20, above right] node {$3$} (2);	
    \path (2) edge [bend left=20, below] node {$2$} (3);
    \path (3) edge [bend left=20, above left] node {$1$} (1);    	

    \path (2) edge [bend left=20,left, pos=0.6] node {$1$} (1);	
    \path (3) edge [bend left=20, above] node {$2.5$} (2);
    \path (1) edge [bend left=20, right, , pos=0.4]  node {$2.5$} (3);    	
    
    \node [below] at (-3,5) {$\breve{c}_{X}$};
  
\end{tikzpicture}
}
		\caption{Quasi-cut metric $\qcm_{X}=0.5 (\qu_{X}^{(1)}+\qu_{X}^{(2)})$}
		\label{quasi-c}
	\end{subfigure}
	
	\vspace{0.5cm}
	
	\begin{subfigure}{1.99\columnwidth}
		\centering

\def \thisplotscale {4.2}
\def \unit {\thisplotscale cm}

\def \yheight{0.75}
\def \xdisplaced{0}

\tikzstyle{dot} = 
	[ellipse,
	 inner sep = 0pt,
	 fill = pennblue,
	 anchor = center,
	 minimum height = 0.05*\unit,
	 minimum width  = 0.05*\unit]

\tikzstyle{cover} = 
	[ellipse,
	 inner sep = 0pt,
	 anchor = center]

\def\xpos{0.1}
\def\ypos{0.125}
\def\xlbl{0.07}
\def\ylbl{0.08}

{\small
\begin{tikzpicture}[scale = \thisplotscale]
	
	\node at (-0.25,0) (xbgn) {};
	\node at (0,0) (x0) {};
	\node at (1,0) (x1) {};
	\node at (2,0) (x2) {};
	\node at (2.5,0) (x2p5) {};
	\node at (3,0) (x3) {};
	\node at (3.75,0) (xend) {};

	\path (xbgn.center) edge[thick] (x0.center);
	\path (x0.center) edge[thick] node[above,midway] (x01) {$0 \le \delta < 1$} (x1.center);
	\path (x1.center) edge[thick] node[above,midway] (x12) {$1 \le \delta < 2$} (x2.center); 
	\path (x2.center) edge[thick] node[above,midway] (x22p5) {$2 \le \delta < 2.5$} (x2p5.center);
	\path (x2p5.center) edge[thick] node[above,midway] (x2p53) {$2.5 \le \delta < 3$} (x3.center);
	\path (x3.center) edge[thick,-stealth] node[above,midway] (x3end) {$3 \le \delta$} node[at end,below right] {$\delta$} (xend.center);
	
	\path (x0.center) edge[draw=black!30] node[at start, below] {$0$} ++ (0,\yheight);
	\path (x1.center) edge[draw=black!30] node[at start, below] {$1$} ++ (0,\yheight);
	\path (x2.center) edge[draw=black!30] node[at start, below] {$2$} ++ (0,\yheight);
	\path (x2p5.center) edge[draw=black!30] node[at start, below] {$2.5$} ++ (0,\yheight);
	\path (x3.center) edge[draw=black!30] node [at start, below] {$3$} ++ (0,\yheight);
	\path (-0.2,\yheight) node {$\breve{K}_{X}$};
	
	\path (x01.center) ++ (-\xdisplaced,\yheight/2) node {\input{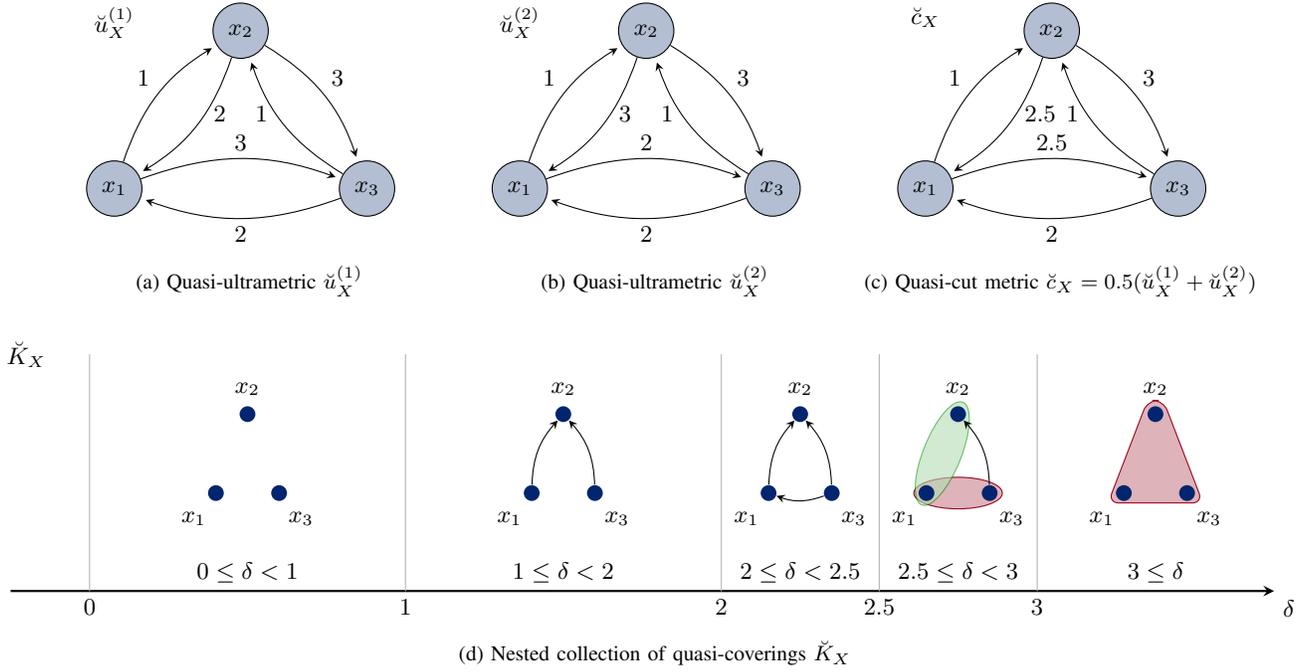}};
	\path (x12.center) ++ (-\xdisplaced,\yheight/2) node {

\begin{tikzpicture}[scale = \thisplotscale]
	
	\path (-\xpos,-\ypos) node[dot] (x11) {} ++ (-\xlbl,-\ylbl) node {$x_{1}$};
	\path (0,\ypos) node[dot] (x12) {} ++ (0,\ylbl) node {$x_{2}$};
	\path (\xpos,-\ypos) node[dot] (x13) {} ++ (\xlbl,-\ylbl) node {$x_{3}$};
	
	\path (x11) edge[bend left=20,-stealth] (x12);
	\path (x13) edge[bend right=20,-stealth] (x12);
	
\end{tikzpicture}};
	\path (x22p5.center) ++ (-\xdisplaced,\yheight/2) node {

\begin{tikzpicture}[scale = \thisplotscale]
	
	\path (-\xpos,-\ypos) node[dot] (x21) {} ++ (-\xlbl,-\ylbl) node {$x_{1}$};
	\path (0,\ypos) node[dot] (x22) {} ++ (0,\ylbl) node {$x_{2}$};
	\path (\xpos,-\ypos) node[dot] (x23) {} ++ (\xlbl,-\ylbl) node {$x_{3}$};
	
	\path (x21) edge[bend left=20,-stealth] (x22);
	\path (x23) edge[bend right=20,-stealth] (x22);
	\path (x23) edge[bend left=20,-stealth] (x21);
	
\end{tikzpicture}};
	\path (x2p53.center) ++ (-\xdisplaced,\yheight/2) node {

\begin{tikzpicture}[scale = \thisplotscale]
	
	\node[cover, 
		minimum height = 0.1*\unit,
		minimum width = 0.28*\unit,
		fill=pennred!30,
		draw=pennred] 
	at (0,-\ypos) {};
	\node[cover,
		minimum height = 0.12*\unit,
		minimum width = 0.35*\unit,
		rotate=atan(2*\ypos/\xpos),
		opacity=0.6,
		fill=penngreen!30,
		draw=penngreen] 
	at (-0.5*\xpos,0) {};
	
	\path (-\xpos,-\ypos) node[dot] (x2p51) {} ++ (-\xlbl,-\ylbl) node {$x_{1}$};
	\path (0,\ypos) node[dot] (x2p52) {} ++ (0,\ylbl) node {$x_{2}$};
	\path (\xpos,-\ypos) node[dot] (x2p53) {} ++ (\xlbl,-\ylbl) node {$x_{3}$};
	
	\path (x2p53) edge[bend right=20,-stealth] (x2p52);
	
\end{tikzpicture}};
	\path (x3end.center) ++ (-\xdisplaced,\yheight/2) node {

\begin{tikzpicture}[scale = \thisplotscale]
	
	\draw[rounded corners,fill=pennred!30,draw=pennred] (-1.5*\xpos,-1.25*\ypos)--(-0.25*\xpos,1.35*\ypos)--(0.25*\xpos,1.35*\ypos)--(1.5*\xpos,-1.25*\ypos)--cycle;
	
	\path (-\xpos,-\ypos) node[dot] (x31) {} ++ (-\xlbl,-\ylbl) node {$x_{1}$};
	\path (0,\ypos) node[dot] (x32) {} ++ (0,\ylbl) node {$x_{2}$};
	\path (\xpos,-\ypos) node[dot] (x33) {} ++ (\xlbl,-\ylbl) node {$x_{3}$};
	
\end{tikzpicture}};
	

\end{tikzpicture}
}
		\caption{Nested collection of quasi-coverings $\qK_{X}$}
		\label{quasi-nested}
	\end{subfigure}
	\caption{Illustrative example of a hierarchical overlapping quasi-clustering method $\breve{\ccalO}$. \subref{quasi-u1}-\subref{quasi-u2} Given quasi-ultrametrics. \subref{quasi-c} Quasi-cut metric obtained by averaging the given quasi-ultrametrics. \subref{quasi-nested} Nested collection of quasi-coverings obtained from $\qcm_{X}$. For $\delta=0$ all nodes belong to separate covers and the edge set is empty. At $\delta=1$ nodes $x_{1}$ and $x_{3}$ start influencing $x_{2}$. At $\delta=2$ node $x_{3}$ also exercises influence over $x_{1}$. The first two covers form at $\delta=2.5$ and they present an overlap, since $x_{1}$ belongs to both covers. Additionally, node $x_{3}$ still exercises influence over $x_{2}$. Finally, for $\delta=3$ all nodes belong to the same single cover and the edge set is empty.}
	\label{fig:quasi_example}
\end{figure*}

Our objective is to provide a framework for developing a hierarchical overlapping quasi-clustering of the node set $X$. A quasi-cluster is a structure that consists of both a grouping of the node set and a collection of edges indicating the influence between different groups of nodes \cite{carlsson14icml}. In extending this notion, we start by defining quasi-coverings, which are the basic units of the proposed method, and we also determine the restrictions required for a collection of quasi-coverings to be nested. 

\begin{definition} \label{def:quasi-covering}
The pair $\qQ_{X}=(Q_{X},F_{X})$ is a quasi-covering if $Q_{X}=\{C_{1},\ldots,C_{m}\}$ is a covering of $X$ and $F_{X} \subseteq Q_{X} \times Q_{X}$ is a set (possibly empty) of ordered pairs representing directed edges between the elements of $Q_{X}$.
\end{definition}

A quasi-covering is a directed graph whose node set is given by the blocks of a covering. The edges represent the asymmetric influence that some blocks can exercise over others. 
The conditions needed for a collection of quasi-coverings to be nested are described next.

\begin{definition}\label{def:quasi-nested}
The collection $\qK_{X}=\{\qK_{X}(\delta), \delta \ge 0\}$ is a nested collection of quasi-coverings if $\qK_{X}(\delta)=(K_{X}(\delta),F_{X}(\delta))$ is a quasi-covering for each resolution parameter $\delta$, $K_{X}=\{K_{X}(\delta),\delta \ge 0\}$ is a nested collection of coverings (cf. Def.~\ref{def:nested-coverings}), and $F_{X}=\{F_{X}(\delta),\delta \ge 0\}$ is a collection of edge sets satisfying
\begin{enumerate}[(i)]
\item $F_X(0)=\emptyset$ and there exists a $\delta_{\max}$ such that $F_{X}(\delta)=\emptyset$ for all $\delta \ge \delta_{\max}$;
\item For arbitrary resolutions $\delta < \delta'$ and covers $C_{i}, C_{j} \in K_{X}(\delta)$ such that $(C_{i}, C_{j}) \in F_{X}(\delta)$, there exist covers $C_{i}',C_{j}' \in K_{X}(\delta')$ such that $C_{i} \subseteq C_{i}'$, $C_{j} \subseteq C_{j}'$ and either $(C_{i}', C_{j}') \in F_{X}(\delta')$ or $|C_{i} \cap C_{j}| < | C_{i}' \cap C_{j}'|$.
\end{enumerate}
Denote by $\breve{\ccalK}$ the space of all nested collections of quasi-coverings.
\end{definition}

A nested collection of quasi-coverings reflects the different levels of similarity present in the network. More specifically, as $\delta$ grows, covers would typically start exercising influence one over the other, then some nodes start to overlap while the remaining ones still exercise influence, until eventually the covers merge into one; see Fig.~\ref{quasi-nested} for an example. Requirement (i) in Def.~\ref{def:quasi-nested} enforces border conditions on the edge set $F_X$ that are consistent with those required for $K_X$ in Def.~\ref{def:nested-coverings}. More precisely, at resolution $\delta = 0$, each node must belong to its own cluster \emph{and} there should be no influence relations among these. By contrast, for a large enough resolution $\delta_{\max}$ every node must belong to a single cluster, hence, no influence relation is possible. Requirement (ii) enforces nestedness in the influence structure $F_X$. Specifically, if an edge exists at resolution $\delta$ between blocks $C_i$ and $C_j$ then this edge must persist for larger resolutions $\delta'$ or the blocks must become more similar, encoded by having a larger intersection.

We consider a nested collection of quasi-coverings to be the desired output of hierarchical overlapping quasi-clustering methods. Formally, define such methods $\breve{\ccalO}$ as a structure-preserving map from the space of networks to the space of nested collections of quasi-coverings, $\breve{\ccalO}:\ccalN \to \breve{\ccalK}$.
In analogy to the development of Sections~\ref{sec:cut-metrics} and~\ref{sec:algorithm}, in which we used hierarchical (non-overlapping) clustering methods to obtain ultrametrics that are then combined to obtain cut metrics, in what follows we propose to use existing hierarchical (non-overlapping) quasi-clustering methods to obtain quasi-ultrametrics that can be combined into quasi-cut metrics (Def.~\ref{def:quasi-cm}). Quasi-ultrametrics are defined as non-negative functions $\qu_{X}:X \times X \to \reals_{+}$ such that $\qu_{X}(x,x')=0$ if and only if $x=x'$ and $\qu_{x}(x,x') \le \max\{\qu_{X}(x,x''),\qu_{X}(x'',x')\}$ for all $x,x',x'' \in X$ but not necessarily symmetric $\qu_{X}(x,x') \ne \qu_{X}(x',x)$; see \cite{carlsson14icml}. Denote by $\qU$ the set of all possible quasi-ultrametrics.

\begin{definition}\label{def:quasi-cm}
Denote as $\qV$ the conic hull of the set of quasi-ultrametrics, $\qV=\textrm{cone}\{\qU\}$. Then, given a node set $X$, an element $\qcm_{X} \in \qV$ is a quasi-cut metric on the node set $X$.
\end{definition}

Notice that, by definition, a convex combination of quasi-ultrametrics outputs a quasi-cut metric. Moreover, nested collections of quasi-coverings can be constructed from quasi-cut metrics as shown in the following theorem.

\begin{theorem} \label{thm:quasi-cm-nested}
Given a node set $X$ and a quasi-cut metric $\qcm_{X} \in \qV$, we can obtain a nested collection of quasi-coverings $\qK_{X}=\{\qK_{X}(\delta),\delta \ge 0\}$, with $\qK_{X}(\delta)=(K_{X}(\delta),F_{X}(\delta))$ a quasi-covering, as follows. The covering $K_{X}(\delta)$ is obtained from the tolerance relation $\tol_{\delta}$ determined by
\begin{equation} \label{eq:construct-tol}
	\max\{\qcm_{X}(x,x'),\qcm_{X}(x',x)\} \le \delta \Rightarrow x \tol_{\delta} x'.
\end{equation}
The edge set $F_{X}(\delta)$ is constructed such that for $C_{i},C_{j} \in K_{X}(\delta)$, we have that $(C_i, C_j) \in F_{X}(\delta)$ if
\begin{equation} \label{eq:construct-edgeset}
	\min_{\substack{x \in C_{i}\backslash C_{i} \cap C_{j}\\x' \in C_{j}\backslash C_{i} \cap C_{j}}} \qcm_{X}(x,x') \le \delta.
\end{equation}
\end{theorem}

\begin{proof}

Notice that $\max\{\qcm_{X}(x,x'),\qcm_{X}(x',x)\}$ is a nonnegative, symmetric function that satisfies the identity property. Hence, $\tol_{\delta}$ in \eqref{eq:construct-tol} is a valid tolerance relation and, from the combination of Theorem~\ref{theo:cut-metric-nested-covering} and Remark~\ref{R:different_metrics}, it follows that $K_X=\{K_{X}(\delta),\delta \ge 0\}$ is a nested collection of coverings as required by Def.~\ref{def:quasi-nested}.

We now show that the collection of edge sets $F_X$ satisfies conditions (i) and (ii) in Def.~\ref{def:quasi-nested}.
It follows from the definition of quasi-cut metrics that $\qcm_{X}(x,x')=0$ if and only if $x=x'$ so that $F_{X}(0)=\emptyset$. Also, since $K_X$ is a nested collection of coverings, it must be that for some $\delta_{\max}$ we have $K_{X}(\delta)=\{X\}$ for all $\delta \ge \delta_{\max}$. Consequently, having only one block, the edge set must be empty, i.e., $F_{X}(\delta)=\emptyset$ for all $\delta > \delta_{\max}$.

To prove that $F_X$ satisfies condition (ii), let $C_{i},C_{j} \in K_{X}(\delta)$ and $(C_{i}, C_{j}) \in F_{X}(\delta)$ at a given $\delta$. Then, there is (at least) a pair of nodes $x_{i} \in C_{i}\backslash C_{i} \cap C_{j}$ and $x_{j} \in C_{j}\backslash C_{i} \cap C_{j}$ that satisfy \eqref{eq:construct-edgeset}. 
Focus now on an arbitrary resolution $\delta' > \delta$. The fact that $K_X=\{K_{X}(\delta),\delta \ge 0\}$ is a nested collection of coverings implies the existence of blocks $C_i', C_j' \in K_X(\delta')$ such that $C_i \subseteq C'_i$ and $C_j \subseteq C'_j$. This immediately implies that $x_i \in C'_i$ and $x_j \in C'_j$. Consider then the following two alternatives: (1) if neither $x_i$ nor $x_j$ belong to $C'_i \cap C'_j$, then \eqref{eq:construct-edgeset} forces $(C'_i, C'_j) \in F_X(\delta')$; and (2) if either $x_i$ or $x_j$ belong to $C'_i \cap C'_j$, then we have that $|C'_i \cap C'_j| > |C_i \cap C_j|$. Given that for any $\delta'$ either (1) or (2) must be true, requirement (ii) in Def.~\ref{def:quasi-nested} is satisfied, completing the proof.
\end{proof}

In virtue of Theorem~\ref{thm:quasi-cm-nested} we can readily obtain a nested collection of quasi-coverings from a quasi-cut metric. The quasi-cut metric can thus be constructed by performing a conic combination of quasi-ultrametrics. The algorithm proposed is analogous to the one developed in Section~\ref{sec:algorithm}. Given a network $(X,A_{X})$, first dither $J$ times the dissimilarity function $A_{X}$. Then, obtain a quasi-ultrametric $\tilde{\qu}_{X}^{(i)}$ by applying a hierarchical (non-overlapping) quasi-clustering method -- see, e.g., \cite{carlsson14icml} -- to each of the dithered networks $\tdN_{i}=(X,\tdA_{X}^{(i)})$ for $i=1,\ldots,J$. Finally, generate the quasi-cut metric $\qcm_{X}$ by averaging the $J$ quasi-ultrametrics.

In Fig.~\ref{fig:quasi_example} we show an illustrative example of a hierarchical overlapping quasi-clustering method $\breve{\ccalO}$. We are given two quasi-ultrametrics $\qu_{X}^{(1)}$ and $\qu_{X}^{(2)}$ -- representing those obtained by dithering the same underlying network -- depicted in Figs.~\ref{quasi-u1} and \ref{quasi-u2}, respectively. We then obtain a quasi-cut metric by averaging these two quasi-ultrametrics $\qcm_{X}=0.5(\qu_{X}^{(1)}+\qu_{X}^{(2)})$; see Fig.~\ref{quasi-c}. Observe that $\qcm_{X}$ is not a quasi-ultrametric since $\qcm_{X}(x_{2},x_{3})=3$ which is larger than $\max\{\qcm_{X}(x_{2},x_{1}),\qcm_{X}(x_{1},x_{3})\}=2.5$. In Fig.~\ref{quasi-nested} we portray the nested collection of quasi-coverings resulting from the quasi-cut metric $\qcm_{X}$ as explained in Theorem \ref{thm:quasi-cm-nested}. It is observed that for $\delta=0$ all nodes belong to separate covers and the edge set is empty; and that for $\delta=3$ all nodes belong to the same single cover and the edge set is also empty [cf. (i) in Def.~\ref{def:quasi-nested}]. We can also observe the progression of influence as $\delta$ grows larger. Nodes start influencing each other ($\delta \ge 1$), until covers are created ($\delta=2.5$) based on nodes that were already exercising influence on each other. Then, some nodes overlap ($x_{1}$), while non-overlapping nodes still exercise influence ($x_{3}$ onto $x_{2}$). Finally, all covers and edges collapse into a single cover containing all nodes ($\delta \ge 3$).

\begin{remark} \normalfont
Quasi-coverings in Def.~\ref{def:quasi-covering} act as generalizations of quasi-partitions introduced in \cite{carlsson14icml}. Likewise, a nested collection of quasi-coverings is the generalization of a quasi-dendrogram. These natural extensions are obtained from the constructions in \cite{carlsson14icml} by dropping the non-intersecting requirement.
\end{remark}


\section{Synthetic Experiments}
	\label{sec:syn_exs}
	
\begin{figure*}
	\captionsetup[subfigure]{justification=centering}
	\centering
	\begin{subfigure}{0.495\columnwidth}
		\centering
		\vspace{0.1cm}

\def \thisplotscale {0.12}
\def \unit {\thisplotscale cm}

\tikzstyle{dot} = [ellipse, 
                   inner sep=0pt, 
                   fill=pennblue, 
                   anchor = center,
                   minimum height = 0.65*\unit, 
                   minimum width  = 0.65*\unit]

{\footnotesize\begin{tikzpicture}[x = 1*\unit, y=1*\unit, font=\footnotesize]

    \path (5,0) node [draw, fill = pennblue!30, rectangle, rounded corners, opacity = 1, 
                      minimum width  = 13*\unit, 
                      minimum height = 13*\unit] (Cleft) {}
       ++ (9,8.5) node [left] {$C_1$}; 
    \path (27,0) node [draw, fill = pennblue!30, rectangle, rounded corners, opacity = 1, 
                      minimum width  = 13*\unit, 
                      minimum height = 13*\unit] (Cright) {}
       ++ ( -9,8.5) node [right] {$C_2$};                       

\path(-0.01, 5.01) node [dot] {};
\path(1.15, 4.95) node [dot] {};
\path(2.12, 4.96) node [dot] {};
\path(3.23, 5.05) node [dot] {};
\path(4.03, 5.03) node [dot] {};
\path(4.95, 4.80) node [dot] {};
\path(5.92, 4.89) node [dot] {};
\path(6.91, 4.91) node [dot] {};
\path(7.96, 4.94) node [dot] {};
\path(9.04, 4.97) node [dot] {};
\path(10.02, 4.98) node [dot] {};
\path(-0.23, 4.02) node [dot] {};
\path(1.01, 3.91) node [dot] {};
\path(1.93, 3.89) node [dot] {};
\path(3.01, 3.99) node [dot] {};
\path(3.98, 3.96) node [dot] {};
\path(5.23, 3.92) node [dot] {};
\path(5.96, 3.94) node [dot] {};
\path(6.87, 4.08) node [dot] {};
\path(8.07, 4.06) node [dot] {};
\path(9.05, 3.84) node [dot] {};
\path(9.99, 4.09) node [dot] {};
\path(0.07, 3.15) node [dot] {};
\path(0.93, 3.20) node [dot] {};
\path(1.90, 3.04) node [dot] {};
\path(3.27, 2.75) node [dot] {};
\path(3.91, 2.91) node [dot] {};
\path(4.89, 3.07) node [dot] {};
\path(5.83, 3.00) node [dot] {};
\path(6.98, 3.09) node [dot] {};
\path(7.98, 3.23) node [dot] {};
\path(9.02, 3.16) node [dot] {};
\path(10.00, 2.90) node [dot] {};
\path(-0.12, 1.87) node [dot] {};
\path(1.10, 2.01) node [dot] {};
\path(1.82, 2.08) node [dot] {};
\path(2.99, 2.02) node [dot] {};
\path(3.82, 1.88) node [dot] {};
\path(4.99, 1.94) node [dot] {};
\path(6.04, 1.93) node [dot] {};
\path(7.01, 1.95) node [dot] {};
\path(7.94, 2.16) node [dot] {};
\path(8.98, 2.11) node [dot] {};
\path(10.05, 1.88) node [dot] {};
\path(0.07, 0.84) node [dot] {};
\path(0.97, 1.21) node [dot] {};
\path(2.06, 0.89) node [dot] {};
\path(3.06, 1.02) node [dot] {};
\path(3.89, 0.90) node [dot] {};
\path(5.03, 1.04) node [dot] {};
\path(5.98, 0.87) node [dot] {};
\path(7.03, 0.91) node [dot] {};
\path(7.94, 0.88) node [dot] {};
\path(9.14, 0.94) node [dot] {};
\path(10.01, 1.03) node [dot] {};
\path(0.16, -0.08) node [dot] {};
\path(1.07, -0.09) node [dot] {};
\path(1.99, -0.16) node [dot] {};
\path(2.96, 0.02) node [dot] {};
\path(4.05, -0.12) node [dot] {};
\path(5.08, -0.07) node [dot] {};
\path(5.79, 0.08) node [dot] {};
\path(6.91, 0.05) node [dot] {};
\path(8.13, -0.05) node [dot] {};
\path(9.08, -0.13) node [dot] {};
\path(10.03, -0.07) node [dot] {};
\path(-0.17, -0.99) node [dot] {};
\path(0.94, -0.98) node [dot] {};
\path(1.96, -0.86) node [dot] {};
\path(3.00, -0.98) node [dot] {};
\path(4.05, -1.15) node [dot] {};
\path(4.99, -0.90) node [dot] {};
\path(6.00, -0.98) node [dot] {};
\path(6.93, -0.84) node [dot] {};
\path(7.95, -1.15) node [dot] {};
\path(8.99, -1.25) node [dot] {};
\path(10.03, -0.96) node [dot] {};
\path(-0.18, -2.11) node [dot] {};
\path(1.01, -1.99) node [dot] {};
\path(1.85, -1.86) node [dot] {};
\path(2.90, -2.14) node [dot] {};
\path(4.02, -2.12) node [dot] {};
\path(4.99, -2.18) node [dot] {};
\path(6.01, -2.02) node [dot] {};
\path(7.07, -1.96) node [dot] {};
\path(7.97, -1.91) node [dot] {};
\path(9.13, -2.05) node [dot] {};
\path(9.98, -2.12) node [dot] {};
\path(-0.09, -2.73) node [dot] {};
\path(0.93, -3.13) node [dot] {};
\path(1.98, -3.11) node [dot] {};
\path(3.18, -3.12) node [dot] {};
\path(4.06, -2.97) node [dot] {};
\path(5.23, -3.05) node [dot] {};
\path(6.02, -2.92) node [dot] {};
\path(6.83, -2.99) node [dot] {};
\path(8.07, -2.91) node [dot] {};
\path(9.07, -3.13) node [dot] {};
\path(9.95, -3.07) node [dot] {};
\path(-0.05, -4.05) node [dot] {};
\path(1.24, -4.19) node [dot] {};
\path(2.07, -3.89) node [dot] {};
\path(2.96, -3.99) node [dot] {};
\path(3.92, -3.93) node [dot] {};
\path(4.93, -4.03) node [dot] {};
\path(6.00, -3.99) node [dot] {};
\path(7.11, -4.01) node [dot] {};
\path(8.05, -4.03) node [dot] {};
\path(8.99, -3.93) node [dot] {};
\path(10.03, -4.01) node [dot] {};
\path(0.10, -4.95) node [dot] {};
\path(1.12, -4.86) node [dot] {};
\path(2.07, -4.81) node [dot] {};
\path(3.17, -5.07) node [dot] {};
\path(3.96, -4.84) node [dot] {};
\path(4.95, -5.08) node [dot] {};
\path(6.01, -4.99) node [dot] {};
\path(6.92, -5.08) node [dot] {};
\path(8.01, -5.04) node [dot] {};
\path(8.85, -4.93) node [dot] {};
\path(10.20, -4.92) node [dot] {};
\path(22.13, 4.80) node [dot] {};
\path(22.98, 4.93) node [dot] {};
\path(24.07, 5.11) node [dot] {};
\path(25.13, 4.86) node [dot] {};
\path(25.99, 5.01) node [dot] {};
\path(26.91, 4.96) node [dot] {};
\path(27.96, 5.00) node [dot] {};
\path(29.10, 5.09) node [dot] {};
\path(29.95, 4.99) node [dot] {};
\path(31.10, 5.00) node [dot] {};
\path(31.88, 4.87) node [dot] {};
\path(21.86, 3.89) node [dot] {};
\path(23.12, 4.16) node [dot] {};
\path(24.15, 3.85) node [dot] {};
\path(25.15, 3.95) node [dot] {};
\path(26.05, 4.03) node [dot] {};
\path(27.30, 3.98) node [dot] {};
\path(28.00, 4.11) node [dot] {};
\path(29.12, 3.94) node [dot] {};
\path(29.83, 3.97) node [dot] {};
\path(30.97, 3.95) node [dot] {};
\path(31.85, 3.91) node [dot] {};
\path(22.18, 2.86) node [dot] {};
\path(23.07, 3.06) node [dot] {};
\path(24.15, 3.06) node [dot] {};
\path(24.91, 2.88) node [dot] {};
\path(26.00, 2.90) node [dot] {};
\path(27.20, 2.96) node [dot] {};
\path(27.93, 2.81) node [dot] {};
\path(28.97, 2.96) node [dot] {};
\path(30.09, 3.02) node [dot] {};
\path(31.08, 2.79) node [dot] {};
\path(32.06, 2.93) node [dot] {};
\path(21.95, 2.04) node [dot] {};
\path(23.08, 1.89) node [dot] {};
\path(24.07, 1.95) node [dot] {};
\path(25.03, 2.05) node [dot] {};
\path(26.13, 2.07) node [dot] {};
\path(27.24, 2.09) node [dot] {};
\path(28.18, 1.98) node [dot] {};
\path(29.13, 1.95) node [dot] {};
\path(30.06, 1.99) node [dot] {};
\path(31.04, 1.96) node [dot] {};
\path(32.02, 2.01) node [dot] {};
\path(22.17, 0.83) node [dot] {};
\path(23.00, 0.85) node [dot] {};
\path(23.87, 0.93) node [dot] {};
\path(24.93, 1.29) node [dot] {};
\path(26.11, 0.99) node [dot] {};
\path(26.96, 1.03) node [dot] {};
\path(28.07, 1.03) node [dot] {};
\path(28.89, 1.06) node [dot] {};
\path(30.07, 1.03) node [dot] {};
\path(31.11, 1.03) node [dot] {};
\path(31.90, 1.11) node [dot] {};
\path(22.00, 0.01) node [dot] {};
\path(23.06, 0.16) node [dot] {};
\path(23.97, 0.02) node [dot] {};
\path(25.03, -0.11) node [dot] {};
\path(26.02, -0.06) node [dot] {};
\path(26.74, 0.04) node [dot] {};
\path(28.01, -0.23) node [dot] {};
\path(28.91, 0.18) node [dot] {};
\path(30.10, -0.04) node [dot] {};
\path(31.09, 0.11) node [dot] {};
\path(32.07, -0.07) node [dot] {};
\path(21.91, -1.10) node [dot] {};
\path(22.91, -0.96) node [dot] {};
\path(23.97, -1.12) node [dot] {};
\path(25.00, -0.91) node [dot] {};
\path(25.88, -1.09) node [dot] {};
\path(27.00, -1.08) node [dot] {};
\path(28.01, -0.93) node [dot] {};
\path(29.19, -0.72) node [dot] {};
\path(30.08, -0.91) node [dot] {};
\path(31.06, -1.19) node [dot] {};
\path(32.04, -0.98) node [dot] {};
\path(21.99, -2.00) node [dot] {};
\path(23.14, -1.89) node [dot] {};
\path(23.96, -1.93) node [dot] {};
\path(25.15, -2.09) node [dot] {};
\path(25.98, -2.14) node [dot] {};
\path(27.11, -1.87) node [dot] {};
\path(27.99, -1.67) node [dot] {};
\path(29.12, -1.95) node [dot] {};
\path(29.90, -2.10) node [dot] {};
\path(31.05, -2.06) node [dot] {};
\path(32.10, -1.98) node [dot] {};
\path(21.81, -3.03) node [dot] {};
\path(22.99, -2.98) node [dot] {};
\path(23.93, -2.96) node [dot] {};
\path(25.02, -2.88) node [dot] {};
\path(25.97, -3.01) node [dot] {};
\path(27.17, -2.93) node [dot] {};
\path(28.13, -3.04) node [dot] {};
\path(28.98, -3.11) node [dot] {};
\path(29.94, -2.94) node [dot] {};
\path(30.97, -2.78) node [dot] {};
\path(31.95, -2.80) node [dot] {};
\path(21.92, -4.07) node [dot] {};
\path(22.99, -4.07) node [dot] {};
\path(24.00, -4.01) node [dot] {};
\path(25.04, -3.89) node [dot] {};
\path(25.91, -4.18) node [dot] {};
\path(27.09, -4.07) node [dot] {};
\path(27.90, -4.09) node [dot] {};
\path(28.96, -4.02) node [dot] {};
\path(29.92, -3.98) node [dot] {};
\path(31.15, -4.17) node [dot] {};
\path(31.95, -3.95) node [dot] {};
\path(21.79, -5.05) node [dot] {};
\path(22.78, -4.98) node [dot] {};
\path(24.03, -4.85) node [dot] {};
\path(24.95, -4.97) node [dot] {};
\path(26.18, -5.07) node [dot] {};
\path(27.18, -4.93) node [dot] {};
\path(27.95, -5.06) node [dot] {};
\path(28.99, -5.14) node [dot] {};
\path(30.14, -5.04) node [dot] {};
\path(31.01, -4.93) node [dot] {};
\path(31.95, -4.87) node [dot] {};

\end{tikzpicture}}
		\vspace{0.06cm}
		\caption{Two clouds\\(covers for $\delta=1.11$)}
		\label{two_clouds-nodes}
	\end{subfigure}
	\hfill
	\begin{subfigure}{0.495\columnwidth}
		\centering
		\input{figures/scaling-cover-01}
		\caption{Multiple resolutions\\(covers for $\delta=1.11$)}
		\label{scaling-cover-01}
	\end{subfigure}
	\hspace{-0.25cm}
	\begin{subfigure}{0.495\columnwidth}
		\centering
		\input{figures/scaling-cover-02}
		\vspace{0.4cm}
		\caption{Multiple resolutions\\(covers for $\delta=2.07$)}
		\label{scaling-cover-02}
	\end{subfigure}
	\hfill
	\begin{subfigure}{0.495\columnwidth}
		\centering
		\vspace{0.1cm}

\def \thisplotscale {0.14}
\def \unit {\thisplotscale cm}

\tikzstyle{dot} = [ellipse, 
                   inner sep=0pt, 
                   fill=pennblue, 
                   anchor = center,
                   minimum height = 0.4*\unit, 
                   minimum width  = 0.4*\unit]

{\footnotesize\begin{tikzpicture}[x = 1*\unit, y=1*\unit, font=\footnotesize]

	\draw [rounded corners,fill =pennred!30] (21,-5)--(21,5)--(11,5)--(11,1.5)--(5.2,1.5)--(5.2,-1.5)--(11,-1.5)--(11,-5)--(21,-5);
	\node at (20,7) {$C_{2}$};

	\draw [rounded corners,fill = pennblue!30, opacity=0.8] (-5,-5)--(-5,5)--(5,5)--(5,1)--(10.9,1)--(10.9,-1)--(5,-1)--(5,-5)--(-5,-5);
	\node at (5,6) {$C_{1}$};

   \foreach \i in {-5,...,5}{
      \foreach \j in {-5,...,5}{   
         \path (0.8*\i,0.8*\j) node [dot] {};
      }
   }
   \foreach \i in {-5,...,5}{  
         \path (8,0) ++ (0.8*\i,0) node [dot] {};
   }

   \foreach \i in {-5,...,5}{
      \foreach \j in {-5,...,5}{   
         \path (16,0) ++ (0.8*\i,0.8*\j) node [dot] {};
      }
   }

\end{tikzpicture}}
		\vspace{0.5cm}
		\caption{Dumbbell\\(covers for $\delta=2.17$)}
		\label{dumbbell-cover}
	\end{subfigure}
	
	\vspace{0.5cm}
	
	\begin{subfigure}{0.66\columnwidth}
		\centering
		\includegraphics[width=0.9\columnwidth]{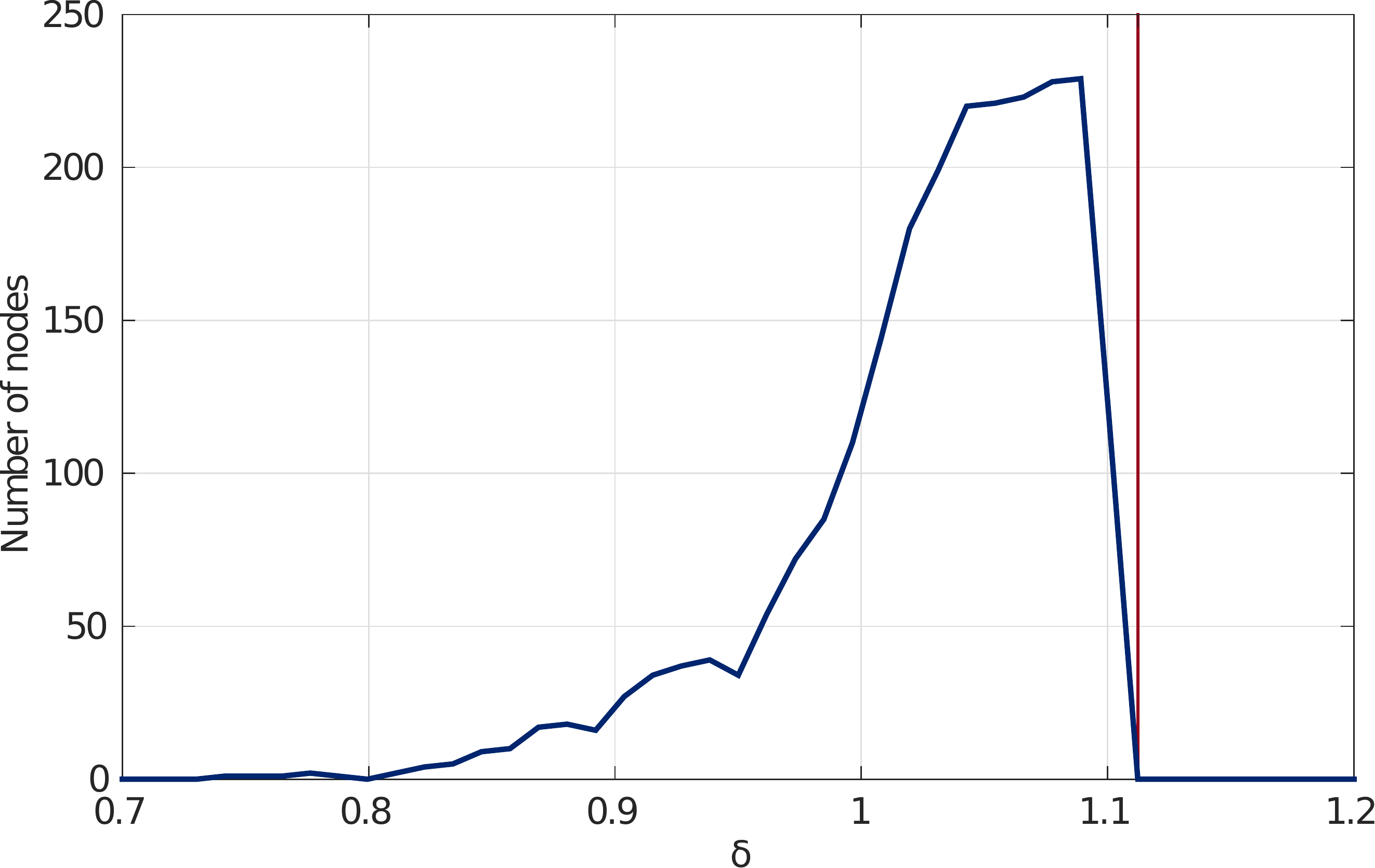}
		\caption{Two clouds: overlapping function}
		\label{two_clouds-ol}
	\end{subfigure}
	\hfill
	\begin{subfigure}{0.66\columnwidth}
		\centering
		\includegraphics[width=0.9\columnwidth]{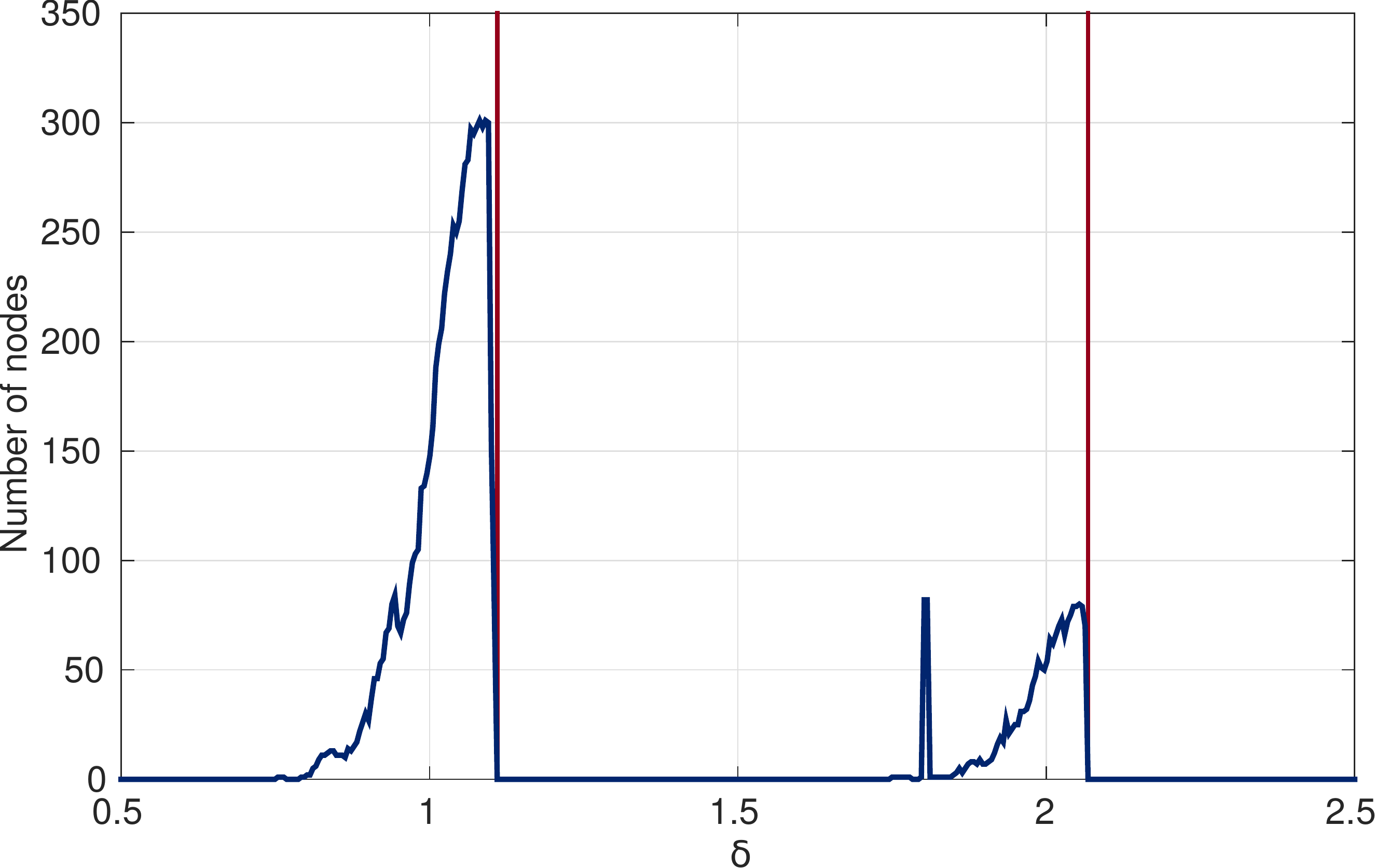}
		\caption{Multiple resolutions: overlapping function}
		\label{scaling-ol}
	\end{subfigure}
	\hfill
	\begin{subfigure}{0.66\columnwidth}
		\centering
		\includegraphics[width=0.9\columnwidth]{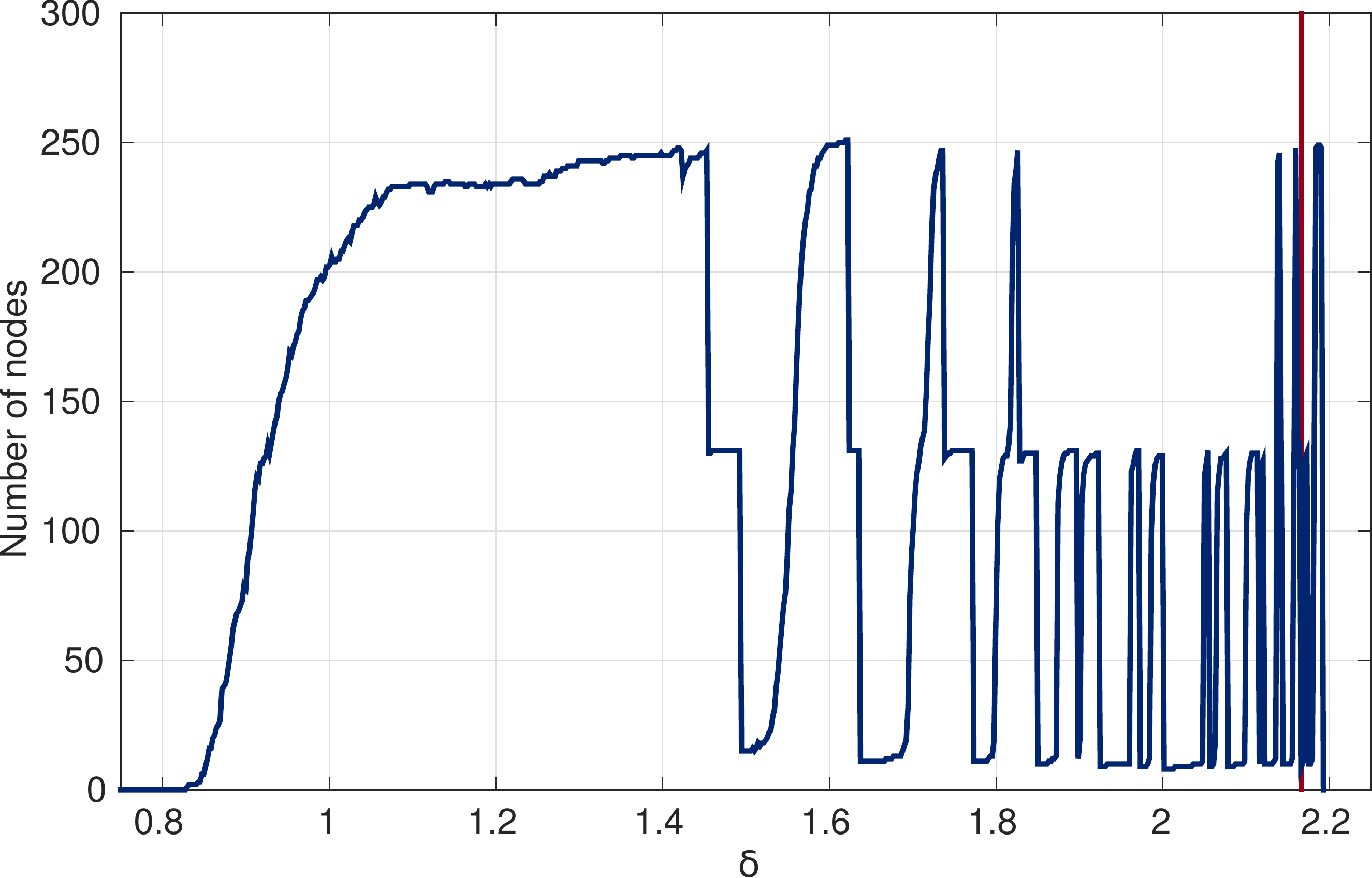}
		\caption{Dumbbell: overlapping function}
		\label{dumbbell-ol}
	\end{subfigure}

	\vspace{0.5cm}	
		
	\begin{subfigure}{0.66\columnwidth}
		\centering
		\includegraphics[width=0.9\columnwidth]{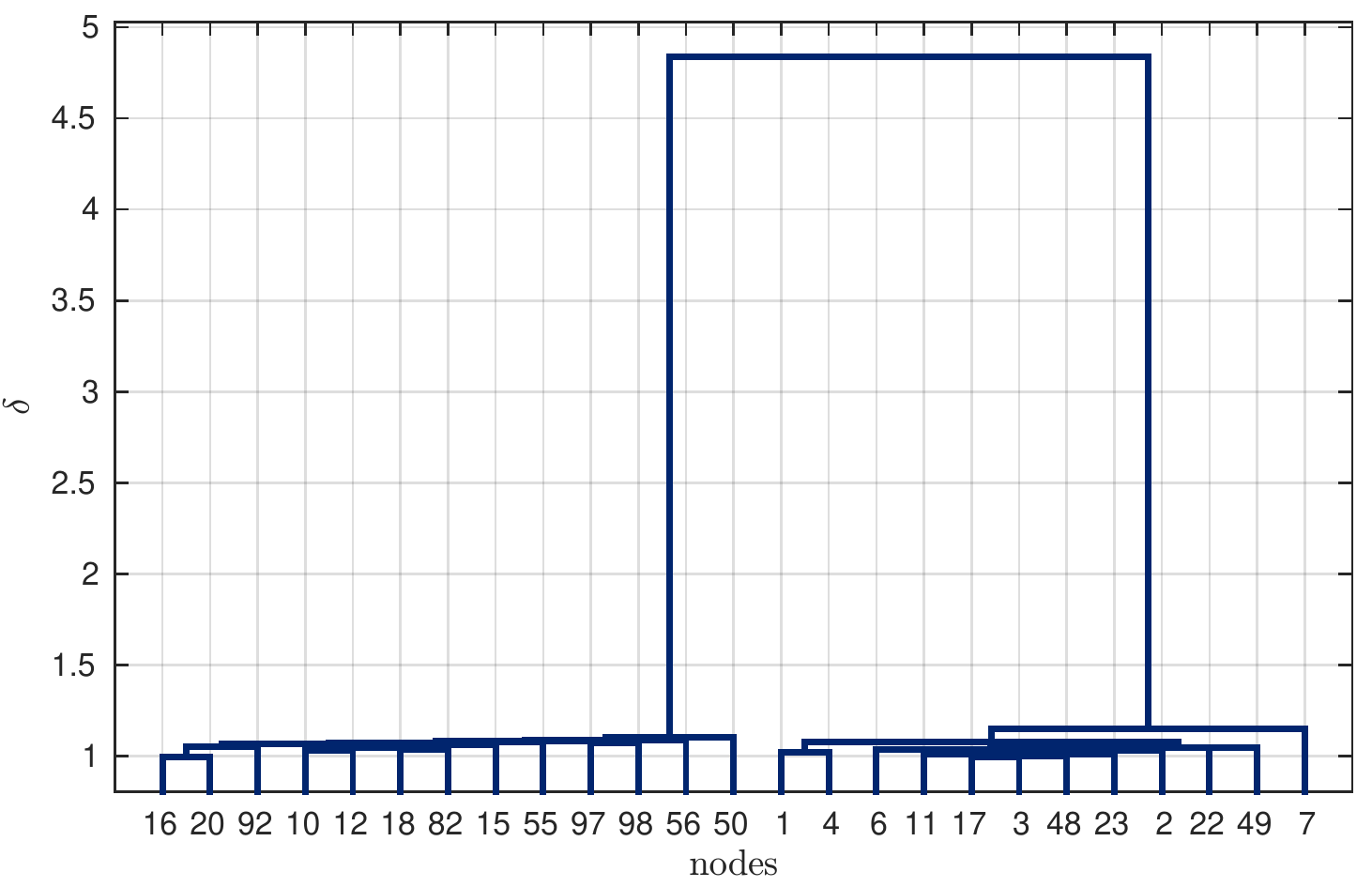}
		\caption{Two clouds: dendrogram}
		\label{two_clouds-dendro}
	\end{subfigure}
	\hfill
	\begin{subfigure}{0.66\columnwidth}
		\centering
		\includegraphics[width=0.9\columnwidth]{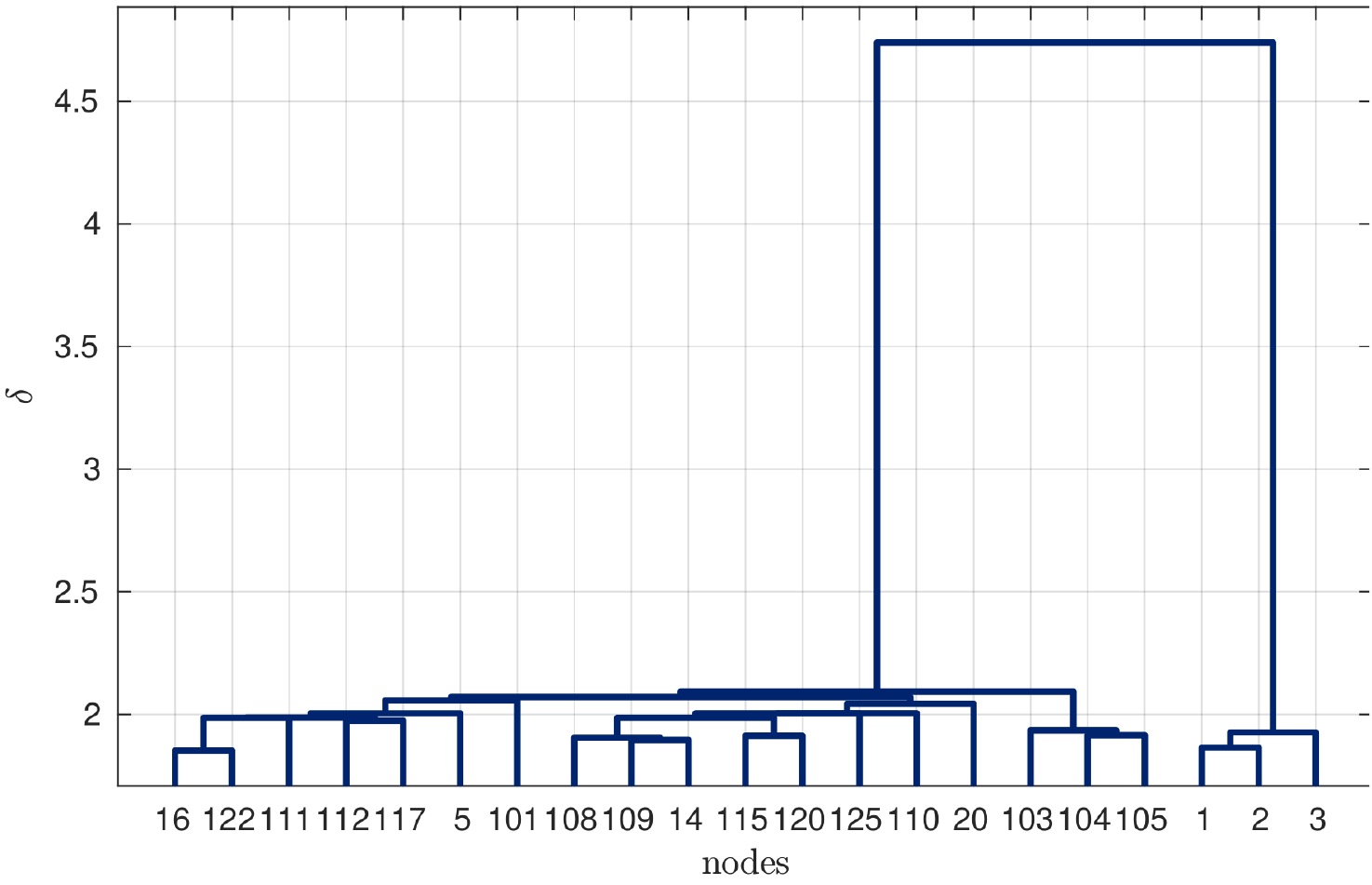}
		\caption{Multiple resolutions: dendrogram}
		\label{scaling-dendro}
	\end{subfigure}
	\hfill
	\begin{subfigure}{0.66\columnwidth}
		\centering
		\includegraphics[width=0.9\columnwidth]{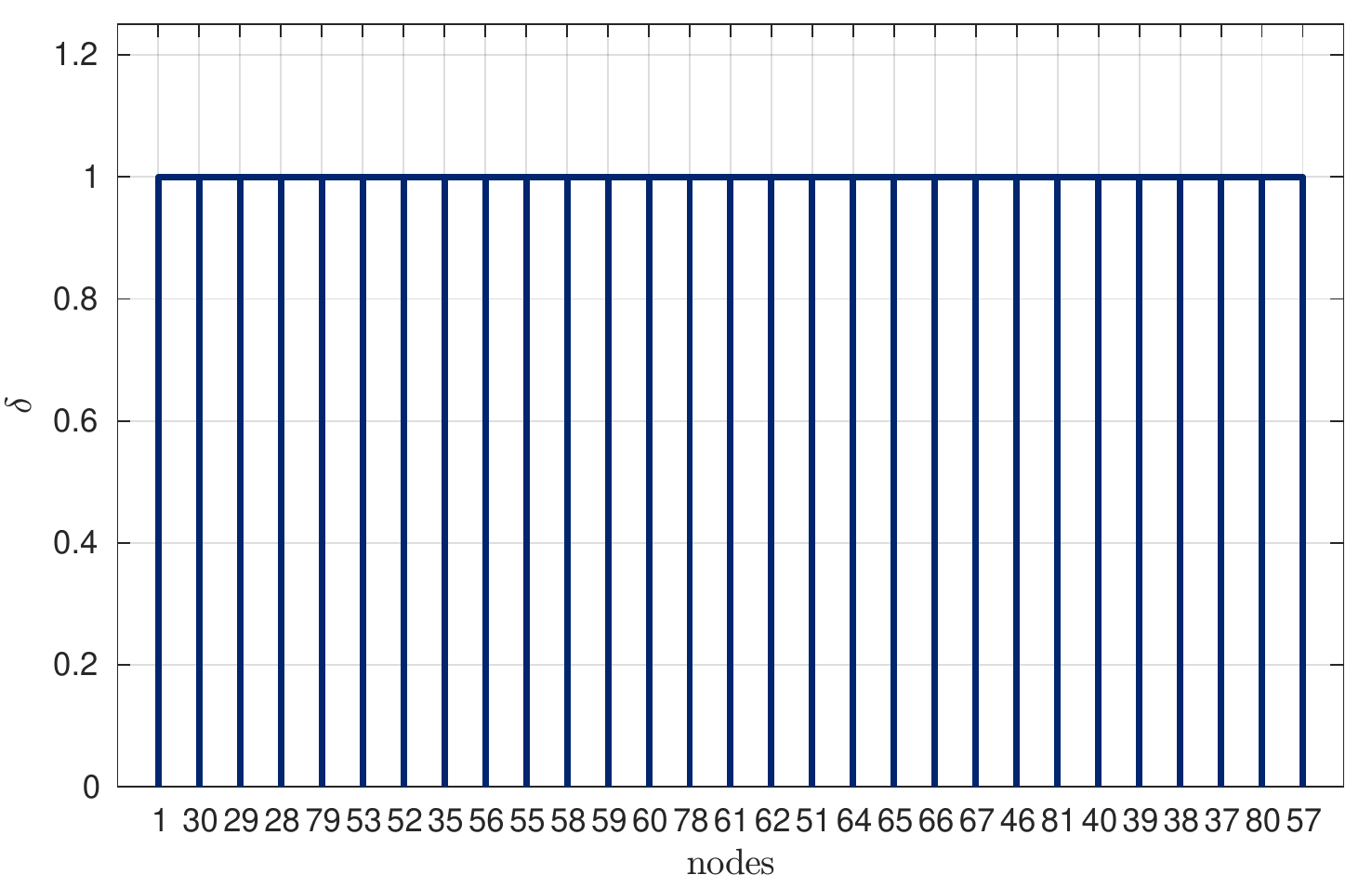}
		\caption{Dumbbell: dendrogram}
		\label{dumbbell-dendro}
	\end{subfigure}
	\caption{Synthetic experiments. \subref{two_clouds-nodes}-\subref{dumbbell-cover} Location of nodes for the three synthetic networks considered: \subref{two_clouds-nodes} two clouds, \subref{scaling-cover-01}-\subref{scaling-cover-02} multiple resolutions, and \subref{dumbbell-cover} dumbbell. Also shown are the coverings obtained at different values of $\delta$. \subref{two_clouds-ol}-\subref{dumbbell-ol} Overlapping functions for each one of the mentioned networks. Shown in red are the values of $\delta$ for which either a zero -- \subref{two_clouds-ol} and \subref{scaling-ol} -- or a minimum of local minima -- \subref{dumbbell-ol} -- is attained. \subref{two_clouds-dendro}-\subref{dumbbell-dendro} Reduced dendrograms resulting from applying single linkage to each of the three synthetic networks, respectively. These dendrograms have been simplified to show only a few illustrative nodes for the sake of clarity (i.e. nodes that are grouped at resolution levels close to $0$ are not shown, and only a representative node in each of these groups is shown).} 
	\label{fig:synthetic}
\end{figure*}

The first three illustrations of Algorithm~\ref{a:overlapping-clustering} are conducted on synthetic networks $N =(X , A_{X })$ where $X $ is a set of points embedded in $\reals^2$ and $A_{X }$ is the Euclidean distance between them. As a hierarchical clustering method $\ccalH$ we apply single linkage \cite{carlsson10} and average the clustering output of $J=100$ different noisy realizations of $N $. The noisy $\tdN$ are obtained by perturbing the positions of the nodes in $\reals^2$ with zero-mean gaussian noise with standard deviation $\sigma$ obtained as $10^{-1}$ of the smallest distance between nodes unless otherwise stated.

For comparison, we include the output of single linkage hierarchical clustering. We observe that, for clusterable datasets, the cut metric approach proposed in this work coincides with the ultrametric clustering (Sections~\ref{sub:two_clouds} and \ref{sub:scaling}). Furthermore, for the dumbbell network (Section~\ref{sub:dumbbell}), the proposed method effectively overcomes the chaining effect.

\subsection{Two clouds}
	\label{sub:two_clouds}

Consider the simple network portrayed in Fig.~\ref{two_clouds-nodes}, where there are $121$ nodes on each cloud. The average distance between adjacent nodes within each of the clouds is $d_{1}=1$ whereas the distance between the centers of both clouds is $d_{2}=22$. The overlapping function is found in Fig.~\ref{two_clouds-ol}.
As expected, this network is clusterable (see Definition~\ref{def:clusterability}), i.e., for $\delta=1.11$ a non-trivial (neither all nodes together nor all separated) partition $K_{X }(\delta)=\{C_{1},C_{2}\}$ is obtained; see Fig.~\ref{two_clouds-nodes}. The resulting dendrogram obtained from applying single linkage to this network is shown in Fig.~\ref{two_clouds-dendro}. By looking at resolution $\delta=1.26$ we obtain the same set of covers as for the cut metric algorithm; see Fig.~\ref{two_clouds-nodes}.
Observe that the value of $\delta$ that generates this partition that identifies both clouds is in the order of $d_{1}$.

\subsection{Multiple resolutions}
	\label{sub:scaling}

Consider now the network depicted in Fig.~\ref{scaling-cover-01} which consists of five clouds of $81$ points each. Within each of the four clouds on the left, the average distance between adjacent nodes is $d_{1}=1$ whereas the average distance between neighboring nodes in different clouds is $d_{2}=2$. The fifth cloud, on the right, has an average distance between adjacent nodes within the cloud of $d_{2}$ and the distance between the center of the right cloud and that of the other four clouds is $d_3 = 22$. The resulting overlapping function is found in Fig.~\ref{scaling-ol}. We have also applied single linkage to this network and obtained the dendrogram shown in Fig.~\ref{scaling-dendro}.

The network is found to be clusterable. First, for $\delta=1.11$ we obtain the covering $K_{X }(\delta)=\{C_{1},C_{2},\ldots,C_{85}\}$, as illustrated in Fig.~\ref{scaling-cover-01}. In this case, each subset $C_{1},\ldots,C_{4}$ contains one of the clouds on the left and each of the subsets $C_{5},\ldots,C_{85}$ contains a single point of the fifth cloud. Intuitively, since $\delta = 1.11$ is close to $d_{1}$ but smaller than $d_{2}$, only nodes that are at distances around or below $d_{1}$ are grouped together. The same result is obtained when looking at the output of the single linkage clustering for resolution $\delta=1.7387$.

Additionally, the network is also clusterable for $\delta=2.07$ with the resulting covering $K_{X }(\delta)=\{C_{1},C_{2}\}$ portrayed in Fig.~\ref{scaling-cover-02}. In this case, there are only two subsets: $C_{1}$ contains the four clouds on the left, and $C_{2}$ contains the right cloud. This is reasonable since $\delta=2.07$ is slightly greater than both the inter-cloud distance on the left and the inter-point distance on the right but smaller than the minimum distance between points in the left and right clouds. In the dendrogram, for resolution $\delta=2.1198$ we get the same partition.

\subsection{A solution to single linkage's chaining effect}
	\label{sub:dumbbell}

Consider the network formed by the set of $171$ nodes in Fig.~\ref{dumbbell-cover} where the distance between adjacent nodes is $d=1$. For this simulation, the value of $\sigma$ used is equal to the minimum distance between nodes. Unlike the previous examples considered, the expected clustering output is not unequivocal since it is not clear how to group the nodes along the line connecting both point clouds. More specifically, it is not clear whether these points should constitute a cluster in itself or if this putative group should have overlap with the ones corresponding to the point clouds at its extremes. 
As a matter of fact, if one applies (non-overlapping) single linkage clustering to the network of interest, the output dendrogram shown in Fig.~\ref{dumbbell-dendro} only consists of the two trivial partitions: every node in a different cluster for resolutions $\delta <1$ or all nodes clustered together for $\delta \ge 1$. This phenomenon is known as chaining effect \cite{lance67general} and is an undesirable concomitant of the definition of single linkage. Nevertheless, via the introduction of dithering, non-trivial coverings of the network of interest can be recovered, thus, overcoming the aforementioned chaining effect.

The overlapping function is presented in Fig.~\ref{dumbbell-ol} and, in accordance with the intuition explained, the network is found not to be clusterable, i.e., no zeros are attained for intermediate resolutions. In this case, we focus on the resolution $\delta$ leading to the minimum number of overlapping nodes, i.e., we look for a covering associated with the minimum of all local minima of the overlapping function. For $\delta=2.3854$, the resulting covering is depicted in Fig.~\ref{dumbbell-cover} where there are only two subsets $C_{1}$ and $C_{2}$, one containing all the left cloud as well as the connecting line (except for the rightmost point in the line), while the other subset contains all the right cloud and the connecting line (except for the leftmost point). By admitting overlap between clusters and using the overlapping function to guide our analysis, useful insight on the structure of the network was recovered.

\begin{figure*}
	\centering
	\begin{subfigure}{0.66\columnwidth}
		\centering
		\input{figures/mnist_classif_1_7/mnist-all-17}
		\caption{Digits \texttt{1}, \texttt{7}: images}
		\label{digits-1-7}
	\end{subfigure}
	\hfill
	\begin{subfigure}{0.66\columnwidth}
		\centering
		\input{figures/mnist_classif_6_8_9/mnist-all-689}
		\caption{Digits \texttt{6}, \texttt{8}, \texttt{9}: images}
		\label{digits-6-8-9}
	\end{subfigure}
	\hfill
	\begin{subfigure}{0.66\columnwidth}
		\centering
		\input{figures/mnist_classif_0_1_2_7/mnist-all-0127}
		\caption{Digits \texttt{0}, \texttt{1}, \texttt{2}, \texttt{7}: images}
		\label{digits-0-1-2-7}
	\end{subfigure}
		
	\vspace{0.5cm}
	
	\begin{subfigure}{0.66\columnwidth}
		\centering
		\includegraphics[width=0.9\columnwidth]{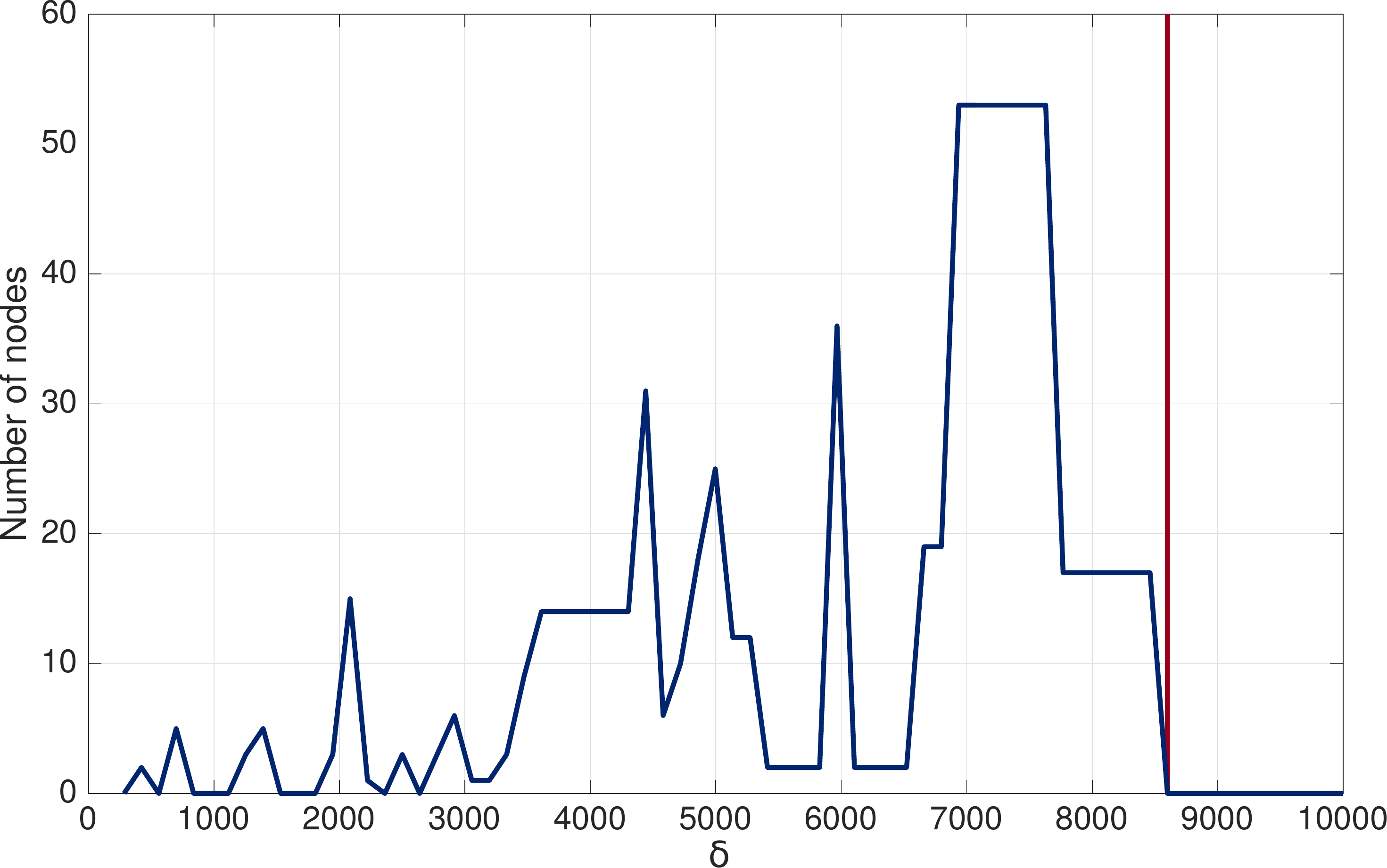}
		\caption{Digits \texttt{1}, \texttt{7}: overlapping function}
		\label{mnist_17-ol}
	\end{subfigure}
	\hfill
	\begin{subfigure}{0.66\columnwidth}
		\centering
		\includegraphics[width=0.9\columnwidth]{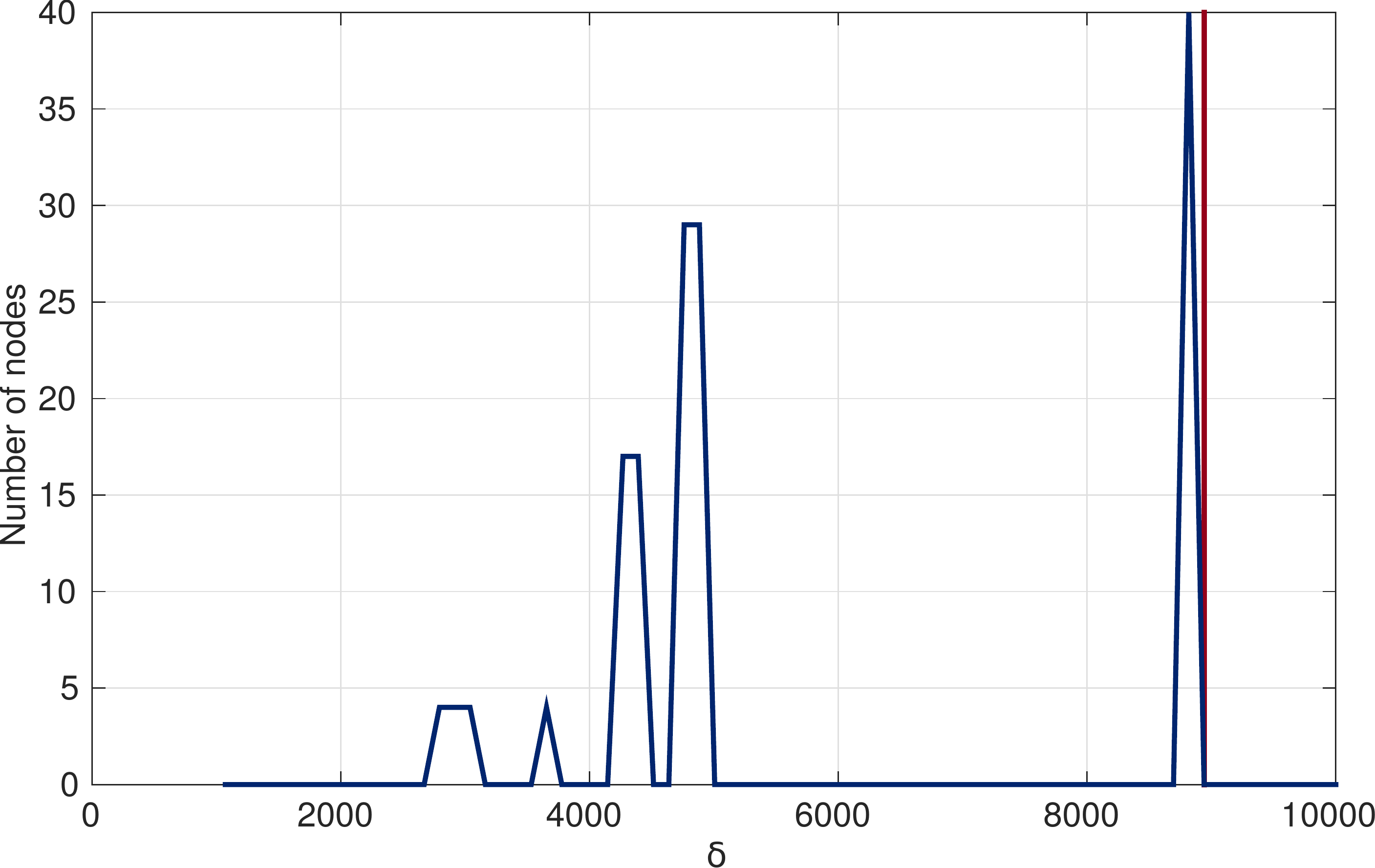}
		\caption{Digits \texttt{6}, \texttt{8}, \texttt{9}: overlapping function}
		\label{mnist_689-ol}
	\end{subfigure}
	\hfill
	\begin{subfigure}{0.66\columnwidth}
		\centering
		\includegraphics[width=0.9\columnwidth]{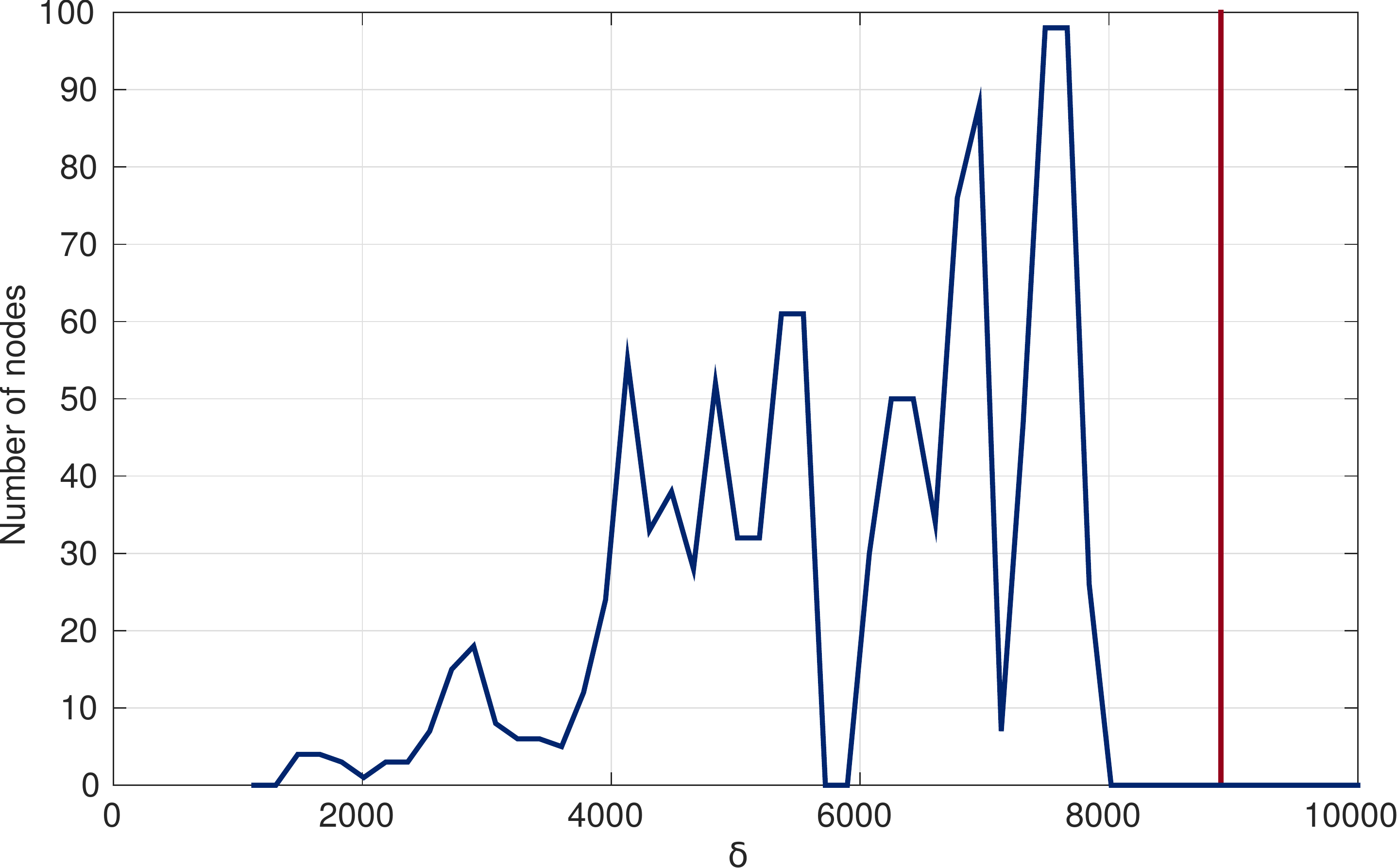}
		\caption{Digits \texttt{0}, \texttt{1}, \texttt{2}, \texttt{7}: overlapping function}
		\label{mnist_0127-ol}
	\end{subfigure}
	
	\vspace{0.5cm}
	
	\begin{subfigure}{0.66\columnwidth}
		\centering
		\includegraphics[width=0.9\columnwidth]{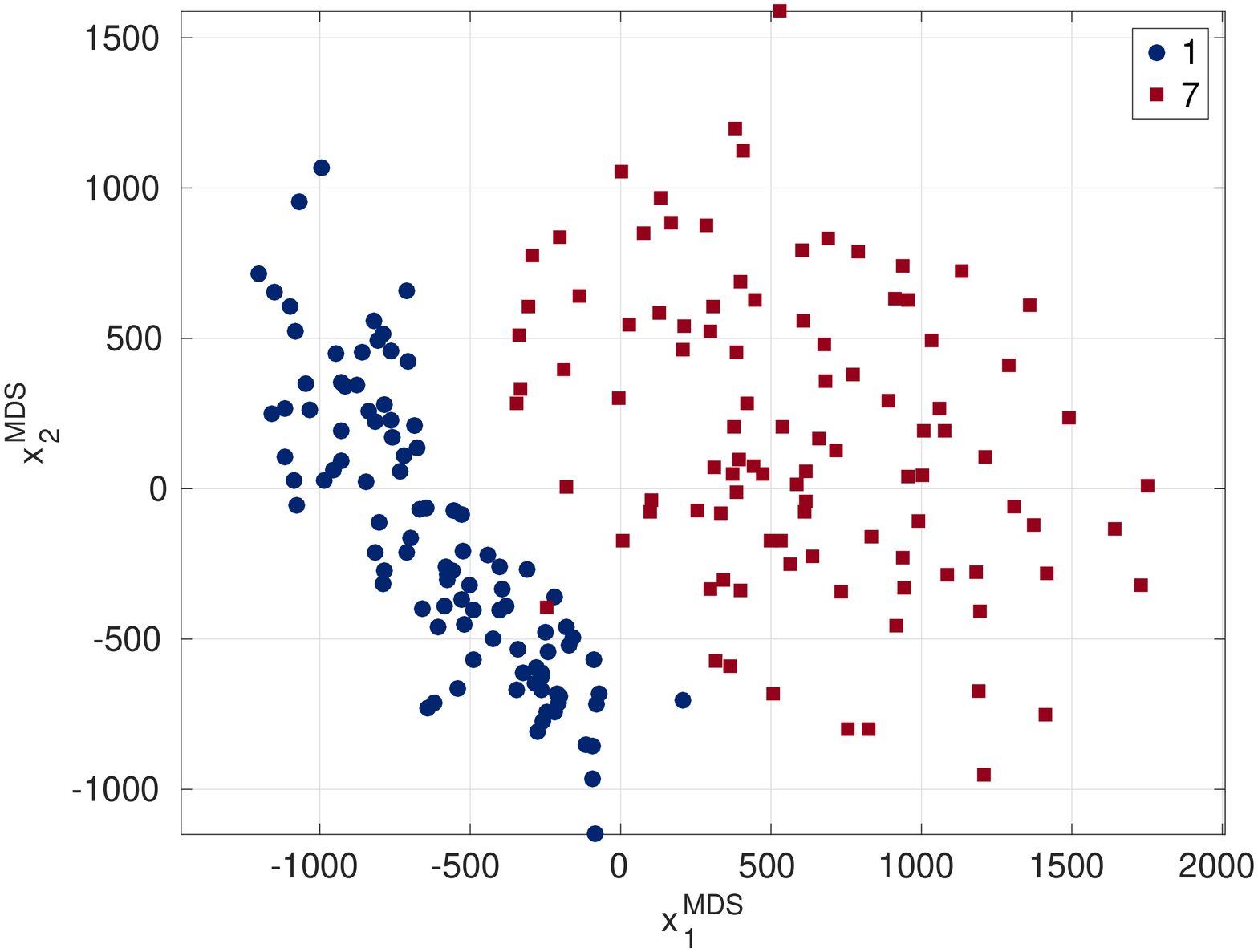}
		\caption{Digits \texttt{1}, \texttt{7}: MDS}
		\label{mnist_17-mds}
	\end{subfigure}
	\hfill
	\begin{subfigure}{0.66\columnwidth}
		\centering
		\includegraphics[width=0.9\columnwidth]{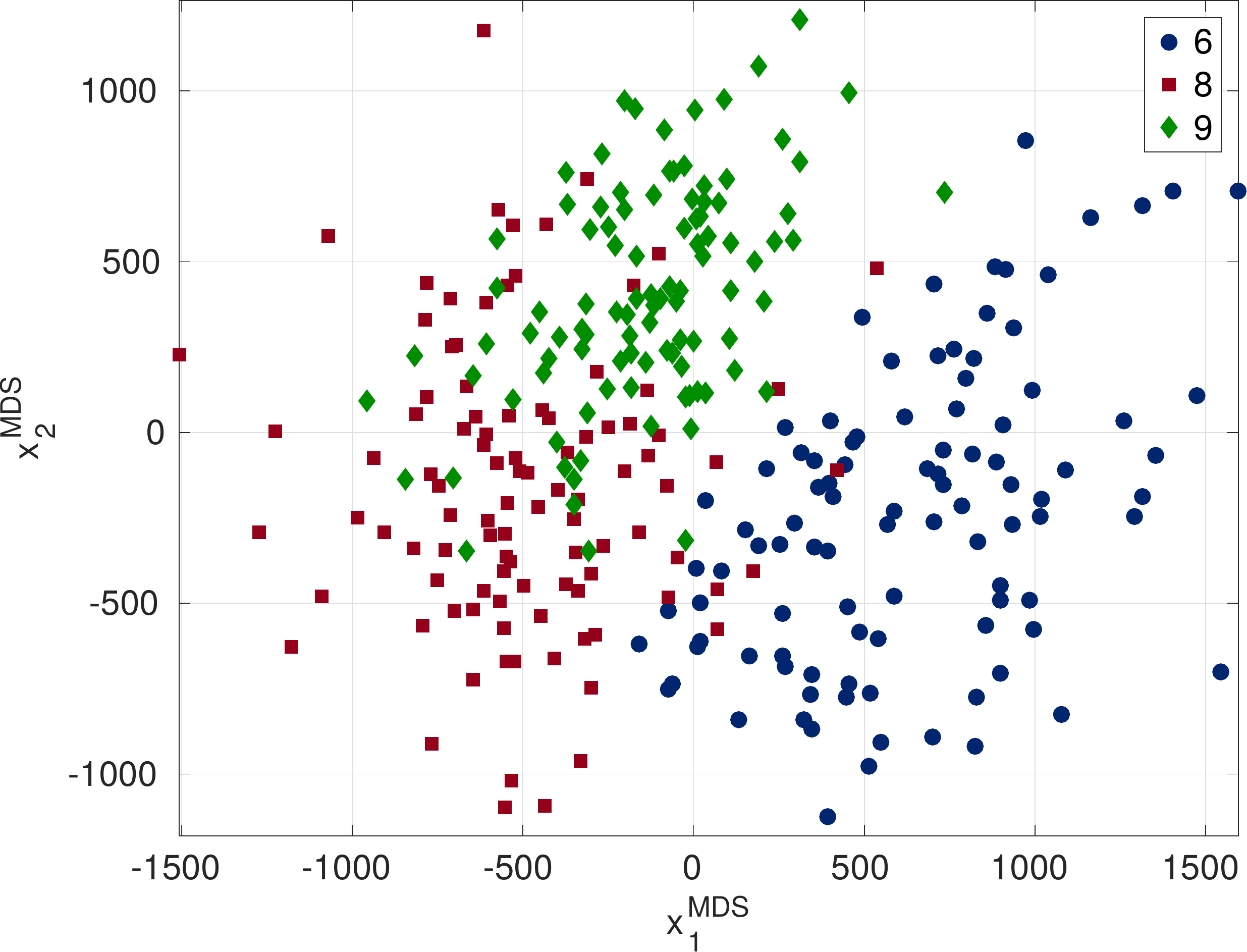}
		\caption{Digits \texttt{6}, \texttt{8}, \texttt{9}: MDS}
		\label{mnist_689-mds}
	\end{subfigure}
	\hfill
	\begin{subfigure}{0.66\columnwidth}
		\centering
		\includegraphics[width=0.9\columnwidth]{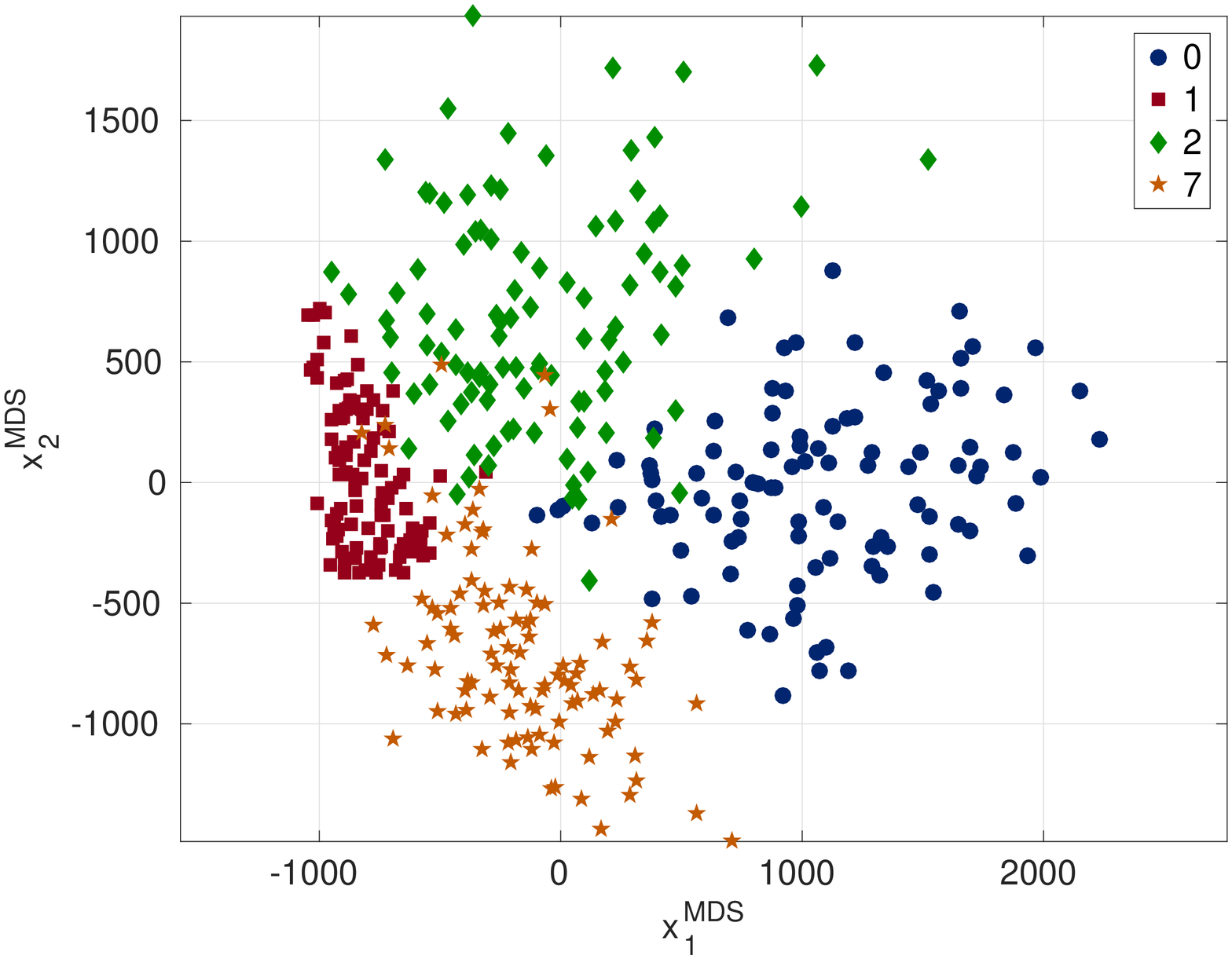}
		\caption{Digits \texttt{0}, \texttt{1}, \texttt{2}, \texttt{7}: MDS}
		\label{mnist_0127-mds}
	\end{subfigure}
	
	\vspace{0.5cm}
	
	\begin{subfigure}{0.66\columnwidth}
		\centering
		\includegraphics[width=0.9\columnwidth]{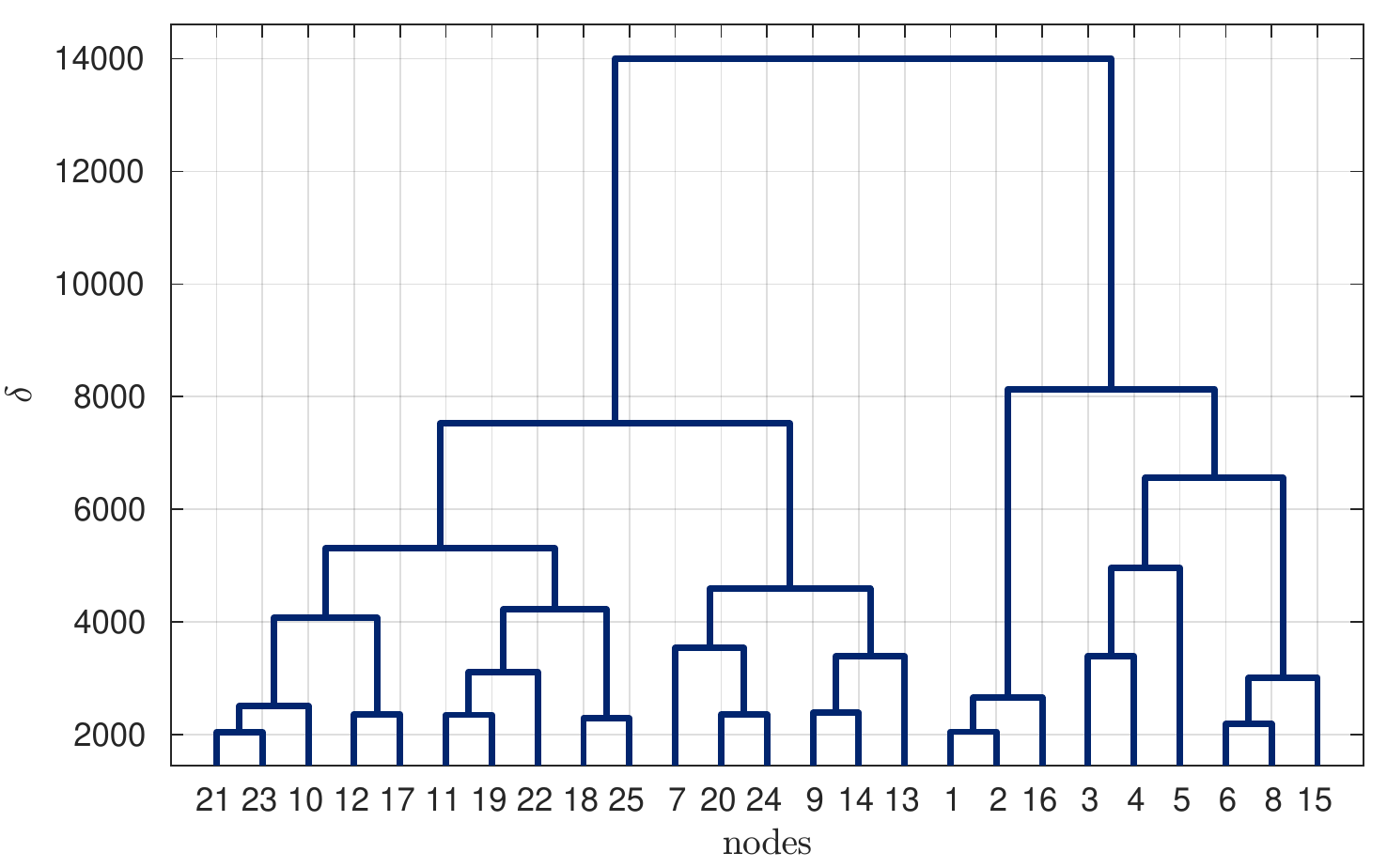}
		\caption{Digits \texttt{1}, \texttt{7}: dendrogram}
		\label{mnist_17-dendro}
	\end{subfigure}
	\hfill
	\begin{subfigure}{0.66\columnwidth}
		\centering
		\includegraphics[width=0.9\columnwidth]{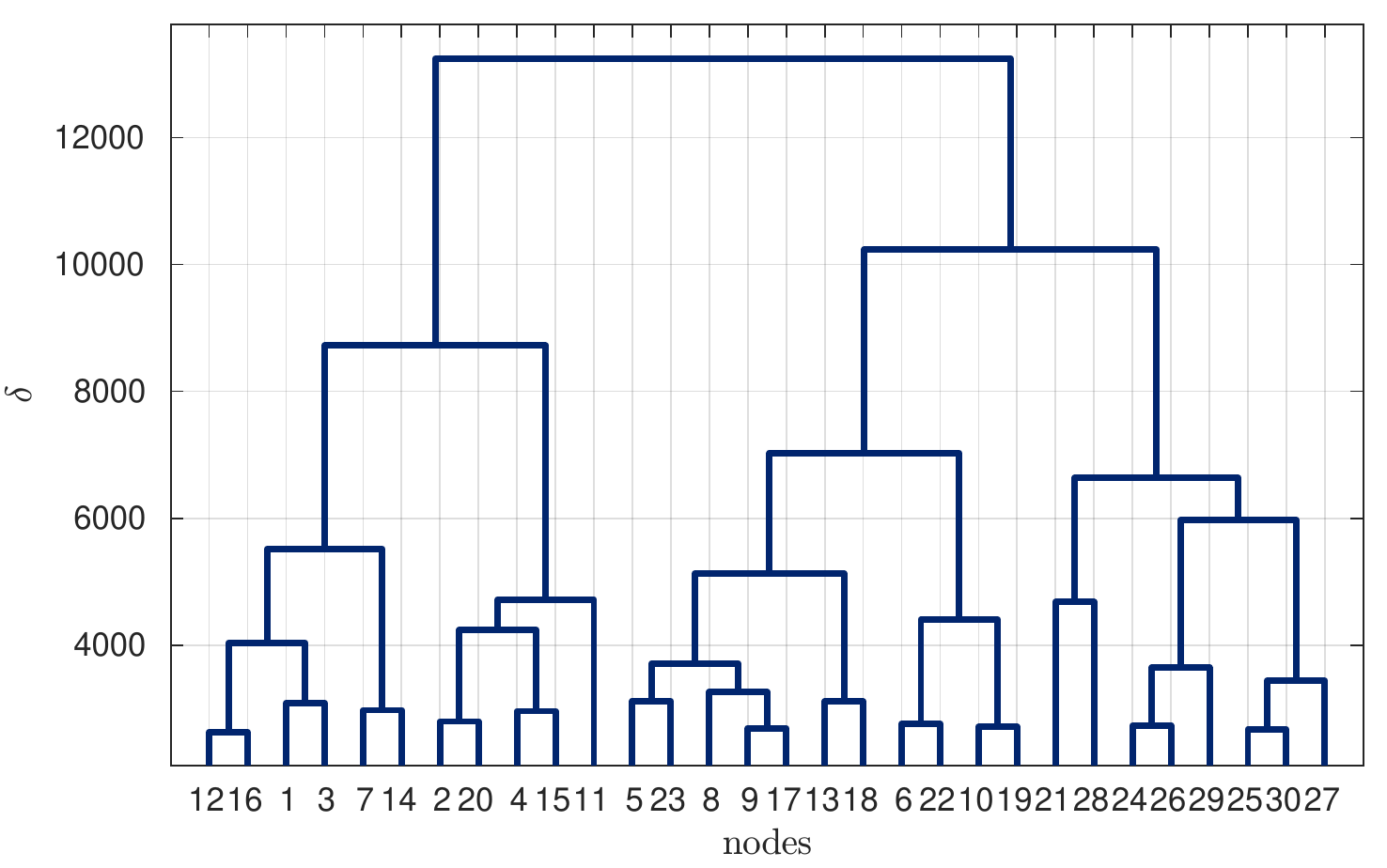}
		\caption{Digits \texttt{6}, \texttt{8}, \texttt{9}: dendrogram}
		\label{mnist_689-dendro}
	\end{subfigure}
	\hfill
	\begin{subfigure}{0.66\columnwidth}
		\centering
		\includegraphics[width=0.9\columnwidth]{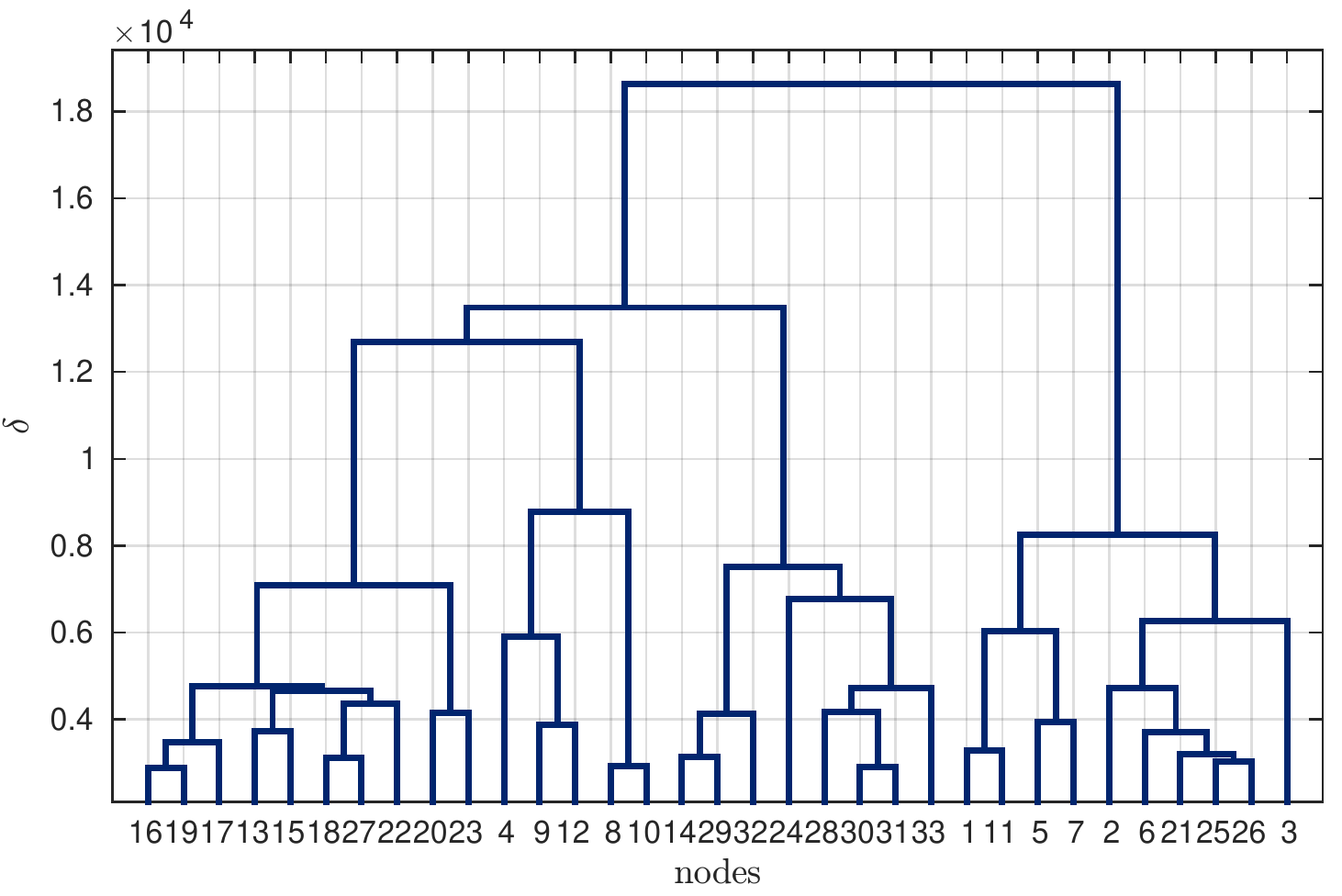}
		\caption{Digits \texttt{0}, \texttt{1}, \texttt{2}, \texttt{7}: dendrogram}
		\label{mnist_0127-dendro}
	\end{subfigure}
	\caption{Classification of handwritten digits. \subref{digits-1-7}-\subref{digits-0-1-2-7} Samples of images of the digits for each one of the three classification problems: \subref{digits-1-7} digits \texttt{1} and \texttt{7}, \subref{digits-6-8-9} digits \texttt{6}, \texttt{8}, and \texttt{9}, \subref{digits-0-1-2-7} digits \texttt{0}, \texttt{1}, \texttt{2}, and \texttt{7}. The first row shows samples of numbers that are typically classified correctly by the proposed method whereas the second row presents numbers that are hard to classify. \subref{mnist_17-ol}-\subref{mnist_0127-ol} Overlapping functions for each of the classification problems considered. \subref{mnist_17-mds}-\subref{mnist_0127-mds} Multidimensional Scaling (MDS) representation in $2$ dimensions of the points in the $20$-PCA domain. \subref{mnist_17-dendro}-\subref{mnist_0127-dendro} Simplified dendrograms resulting from applying Ward's linkage to each of the three classification problems.}
	\label{fig:mnist}
\end{figure*}

\section{Applications}
	\label{sec:apps}

We consider two real-world classification problems: determining handwritten digits from the MNIST image database (Section~\ref{subsec:mnist}), and identifying the authors and co-authors of famous plays (Section~\ref{subsec:author}).

\begin{figure*}
	\centering
	\begin{subfigure}{0.475\columnwidth}
		\centering
		\vspace{0.9cm}
		\includegraphics[width=0.9\columnwidth]{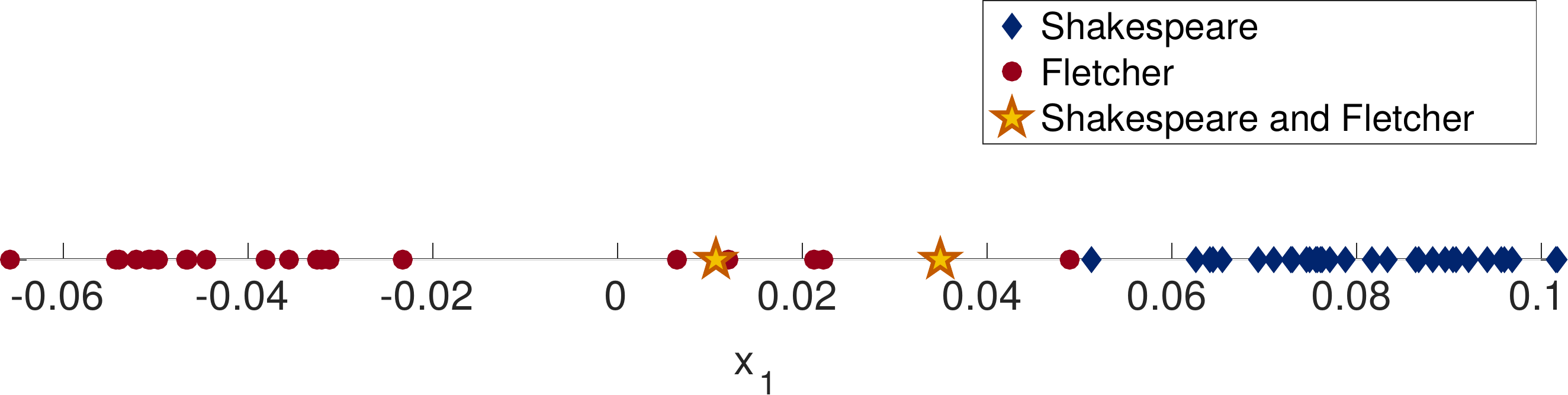}
		\vspace{0.9cm}
		\caption{SF: nodes}
		\label{s-f-mds}
	\end{subfigure}
	\hfill
	\begin{subfigure}{0.475\columnwidth}
		\centering
		\begin{subfigure}{\columnwidth}
			\centering
			\includegraphics[width=0.9\columnwidth]{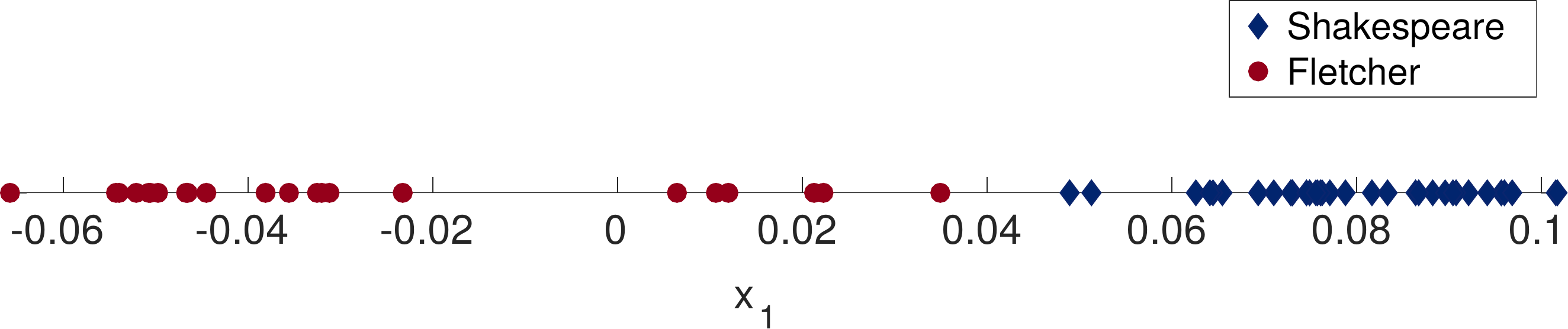}
			\caption{SF: covers for $\delta=0.2005$}
			\label{s-f-cover-02}
		\end{subfigure}
		\begin{subfigure}{\columnwidth}
			\centering
			\includegraphics[width=0.9\columnwidth]{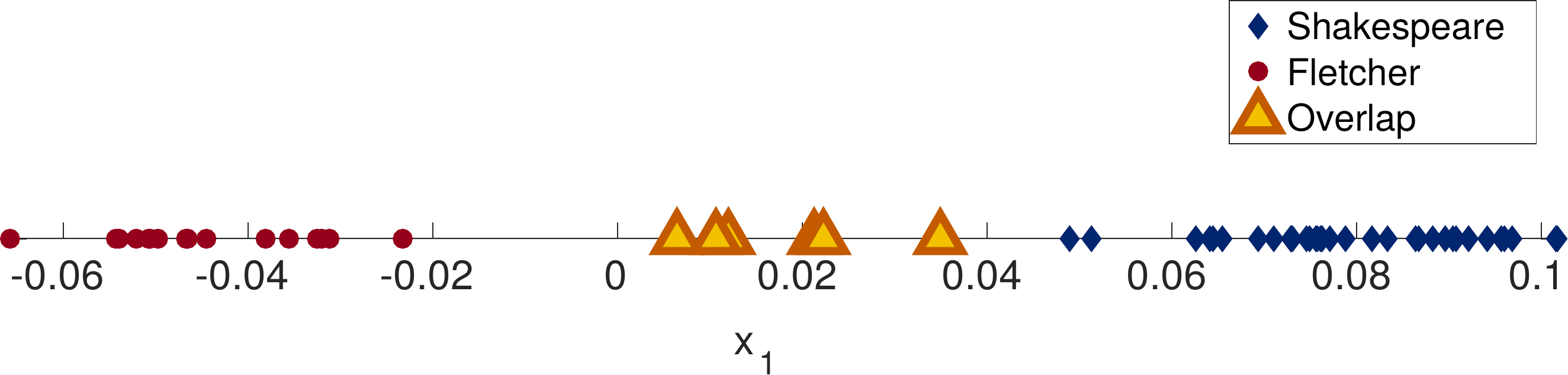}
			\caption{SF: covers for $\delta=0.5459$}
			\label{s-f-cover-01}
		\end{subfigure}
	\end{subfigure}
	\hfill
	\begin{subfigure}{0.475\columnwidth}
		\centering
		\includegraphics[width=0.9\columnwidth]{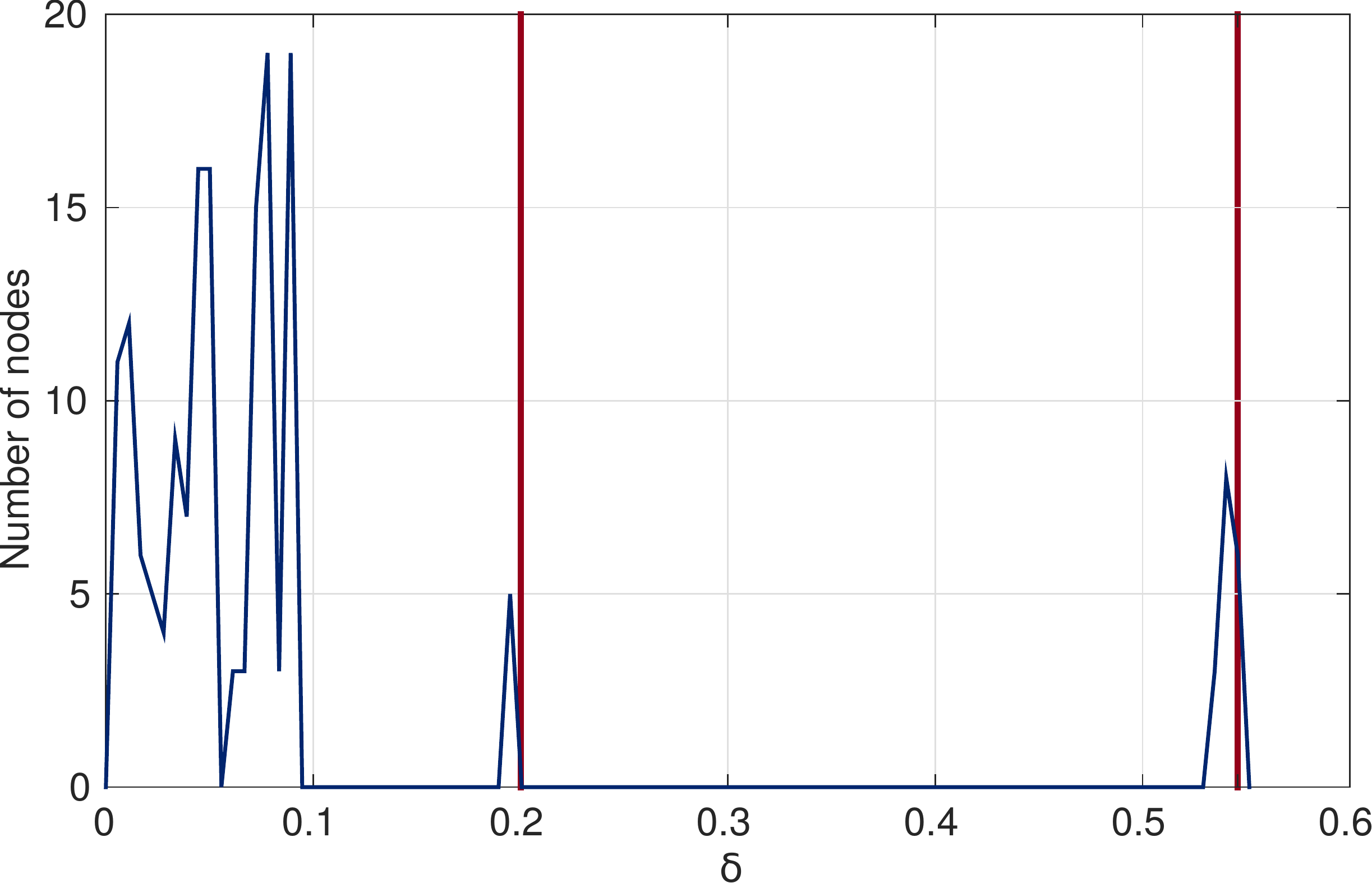}
		\caption{SF: overlapping function}
		\label{s-f-ol}
	\end{subfigure}
	\hfill
	\begin{subfigure}{0.475\columnwidth}
		\centering
		\includegraphics[width=0.9\columnwidth]{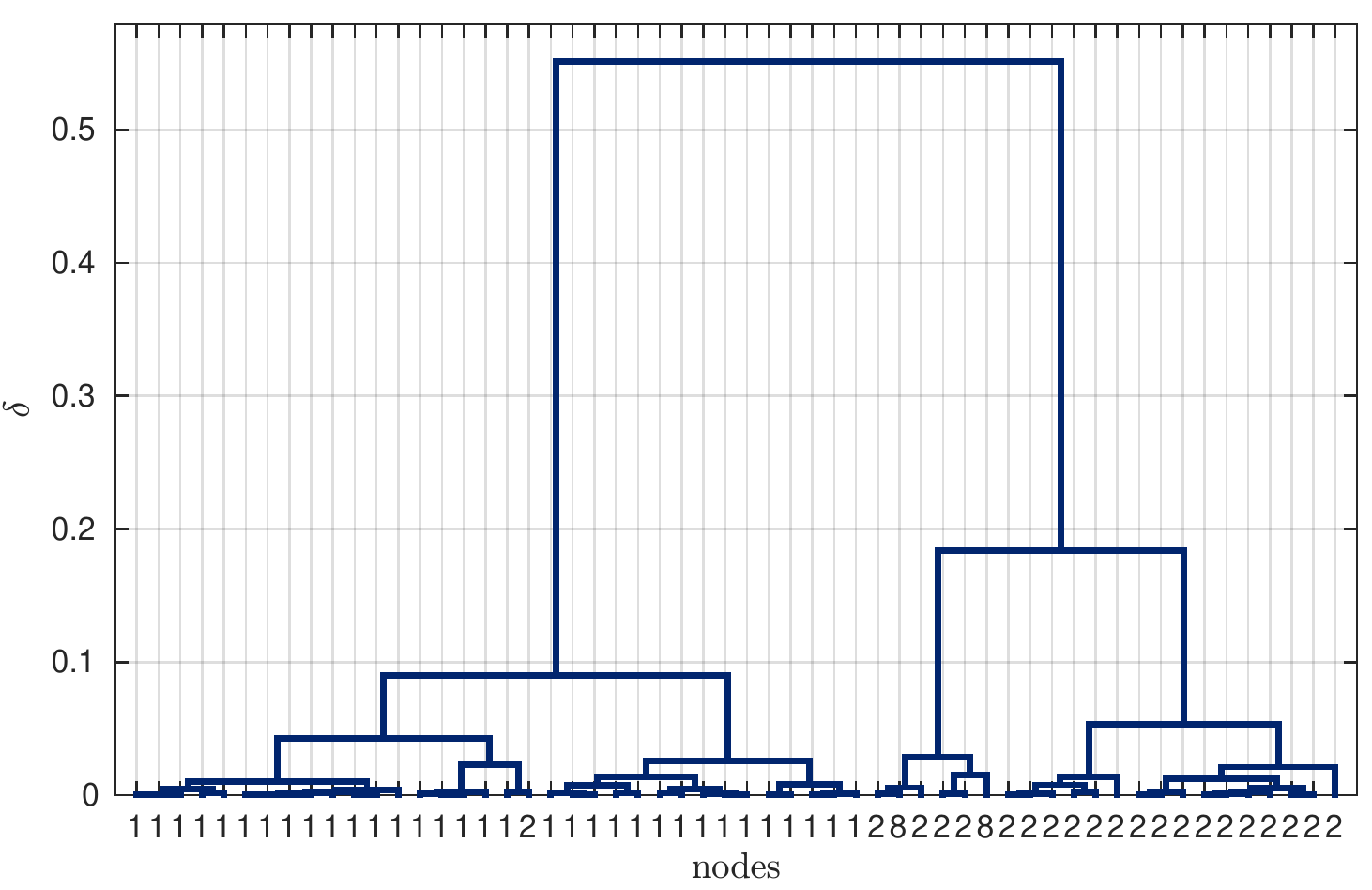}
		\caption{SF: dendrogram}
		\label{s-f-dendro}
	\end{subfigure}
	
	\vspace{0.5cm}
	
	\begin{subfigure}{0.475\columnwidth}
		\centering
		\includegraphics[width=0.9\columnwidth]{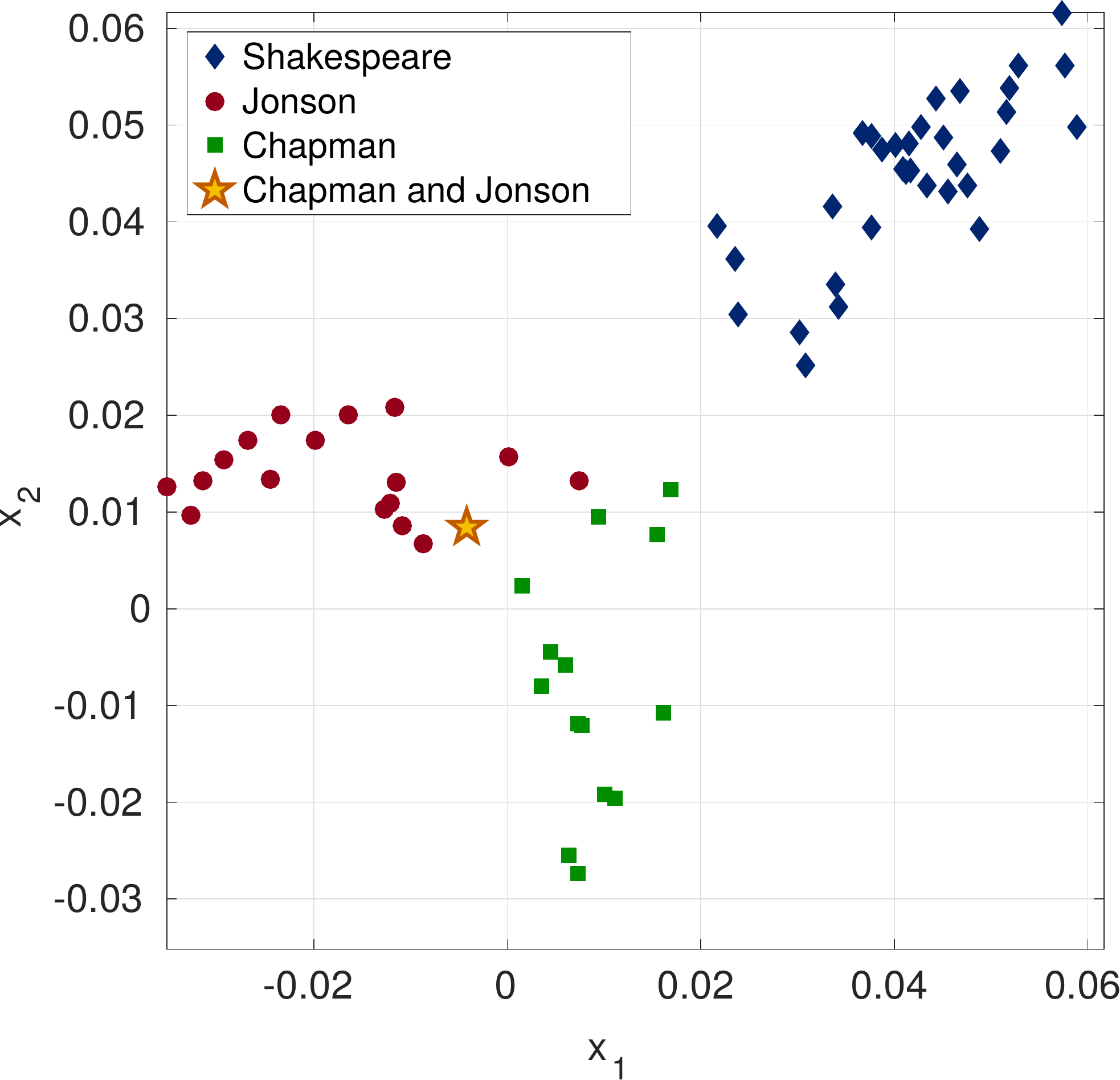}
		\caption{SCJ: nodes}
		\label{s-c-j-mds}
	\end{subfigure}
	\hfill
	\begin{subfigure}{0.475\columnwidth}
		\centering
		\includegraphics[width=0.9\columnwidth]{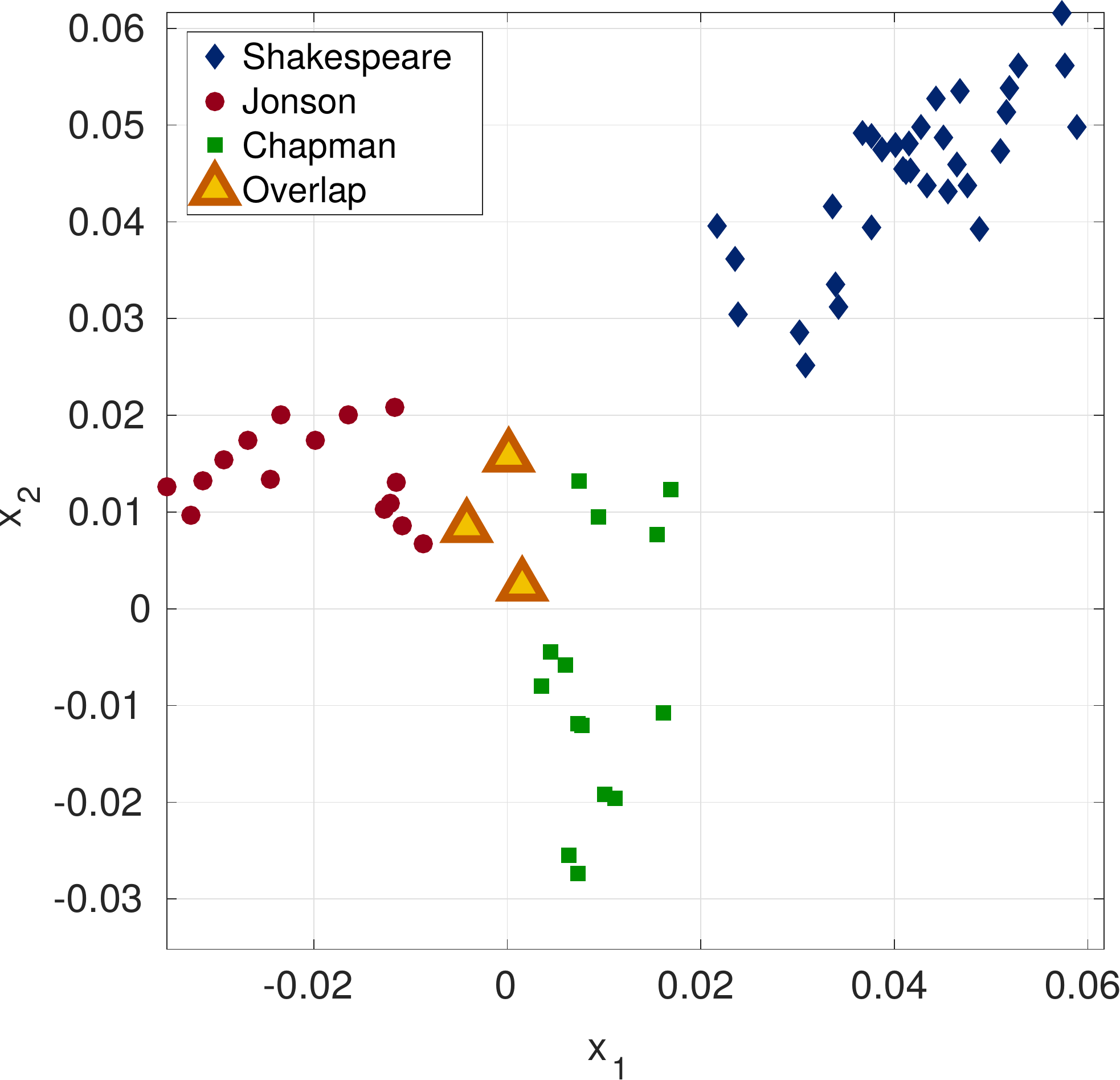}
		\caption{SCJ: covers for $\delta=0.0336$}
		\label{s-c-j-cover}
	\end{subfigure}
	\hfill
	\begin{subfigure}{0.475\columnwidth}
		\centering
		\vspace{0.6cm}
		\includegraphics[width=0.9\columnwidth]{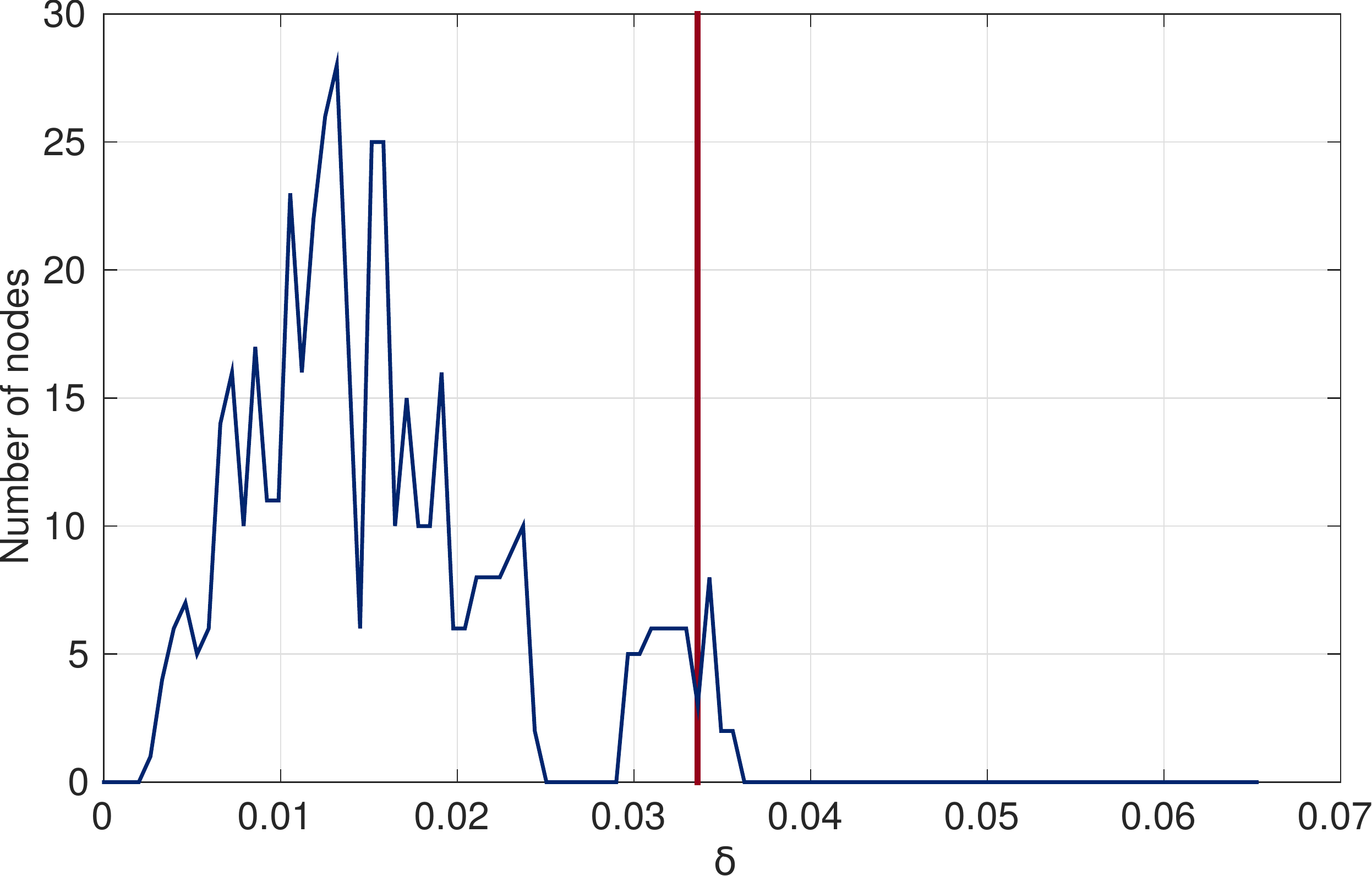}
		\vspace{0.6cm}
		\caption{SCJ: overlapping function}
		\label{s-c-j-ol}
	\end{subfigure}
	\hfill
	\begin{subfigure}{0.475\columnwidth}
		\centering
		\vspace{0.6cm}
		\includegraphics[width=0.9\columnwidth]{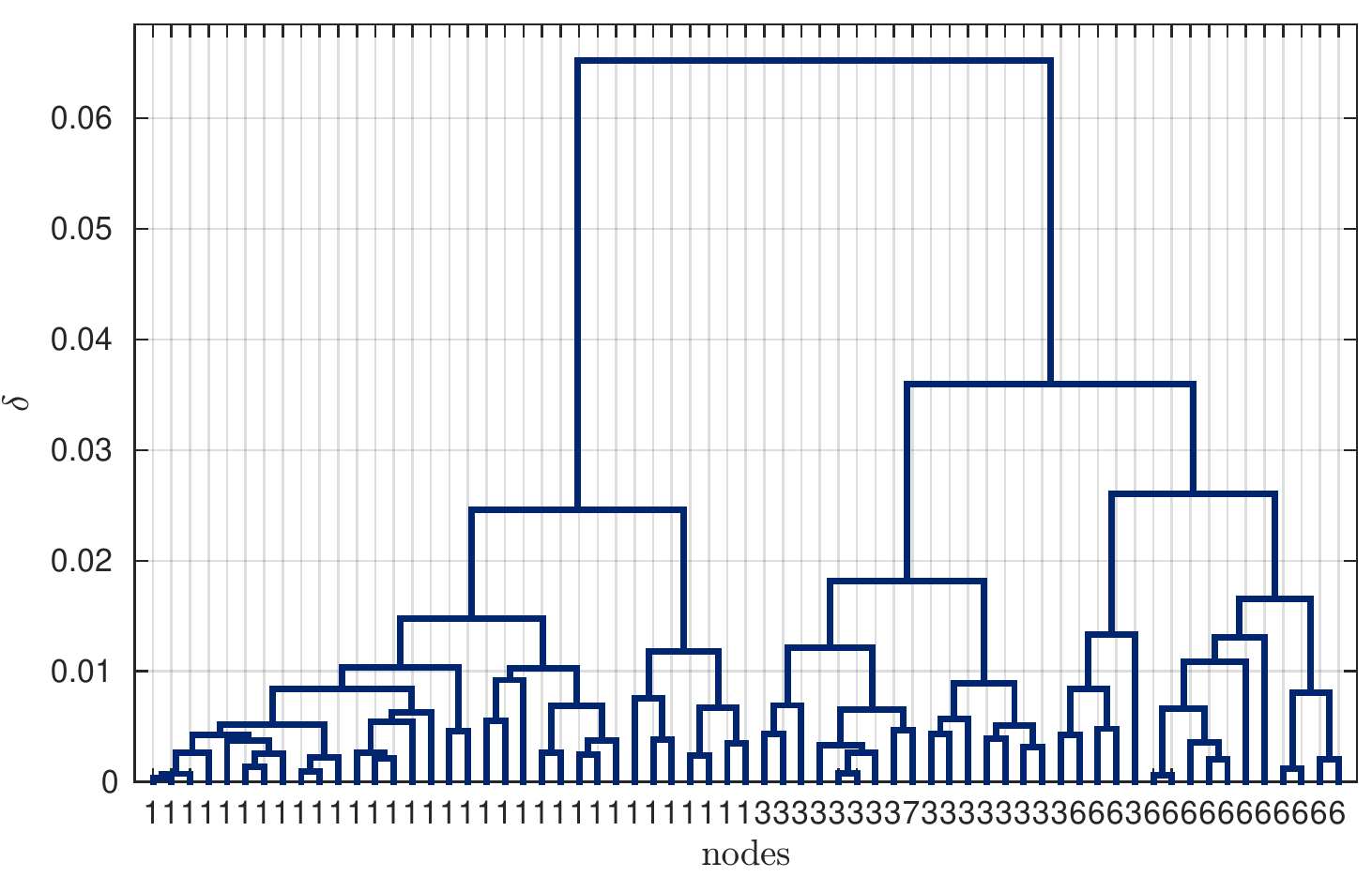}
		\vspace{0.6cm}
		\caption{SCJ: dendrogram}
		\label{s-c-j-dendro}
	\end{subfigure}
	\caption{Authorship attribution. \subref{s-f-mds}-\subref{s-f-dendro} Classifying plays from Shakespeare and Fletcher (SF). \subref{s-f-mds} Location of plays with correct labeling. \subref{s-f-cover-02} Covers obtained for the proposed algorithm at resolution $\delta = 0.1968$ and for Ward's linkage at resolution $\delta=0.1885$ showing no overlap. \subref{s-f-cover-01} Covers for $\delta=0.5460$ obtained by the proposed algorithm present overlap between Shakespeare and Fletcher. \subref{s-f-ol} Overlapping function showing the two resolutions of interest. \subref{s-f-dendro} Dendrogram resulting from applying Ward's linkage in which label $1$ represent Shakespeare's plays, label $2$ are Fletcher's plays and $8$ are the co-authored plays. \subref{s-c-j-mds}-\subref{s-c-j-dendro} Classifying plays from Shakespeare, Chapman and Jonson (SCJ). \subref{s-c-j-mds} Location of plays with correct labeling. \subref{s-c-j-cover} Covers obtained for the proposed method at resolution $\delta=0.0331$ showing overlap between Chapman and Jonson. \subref{s-c-j-ol} Overlapping function highlighting the local minimum of interest. \subref{s-c-j-dendro} Dendrogram obtained from applying UPGMA's linkage in which labels $1$, $3$ and $6$ correspond to plays authored solely by Shakespeare, Jonson and Chapman, respectively and where label $7$ is for the co-authored play by Chapman and Jonson.}
	\label{fig:authorship}
\end{figure*}

\subsection{Handwritten digit classification}
	\label{subsec:mnist}

The proposed hierarchical overlapping clustering algorithm is implemented to classify digits of the MNIST database \cite{mnist}. The database consists of black and white images $\{\bbG_{i}\}_{i \in \ccalI}$ of size $28 \times 28$ pixels.

In order to construct a network from the images, we first reshape each image $\bbG_{i}$ into a column vector $\bby_{i}$ of size $784$ and then apply a principal component analysis (PCA) transform \cite{pca} to each vector. For computing the PCA, $5,000$ training samples of each digit are used to estimate the mean $\hat{\bbmu}$ and the covariance matrix $\hat{\bbC}$. The resulting transformed vectors are denoted $\bby^{\textrm{PCA}}_{i}$ for $i \in \ccalI$. A total of $n$ images are considered for classification, and only the first $20$ components of the associated transformed vectors $\bby^{\textrm{PCA}}_{i}$ are used while the rest are discarded. The reduced vectors $\bbx_{i} \in \reals^{20}$ for $i=1,\ldots,n$ constitute the node set $X=\{\bbx_{1},\ldots,\bbx_{n}\}$ of our network. Each element of the (symmetric) dissimilarity function is computed as the Euclidean distance between vectors in the $20$-PCA space, i.e., $A_{X }(\bbx_{i},\bbx_{j})=\|\bbx_{i}-\bbx_{j}\|_{2}$.

In the application of Algorithm \ref{a:overlapping-clustering}, $J=100$ noise realizations are carried out for the dithering step in all simulations. The standard deviation of the noise $\sigma = k \, \min_{\bbx_i \neq \bbx_j} A_{X }(\bbx_i, \bbx_j)$ is computed as a factor $k$ of the minimum distance between nodes. The hierarchical (non-overlapping) clustering method $\ccalH$ used to obtain the ultrametrics is Ward's clustering method \cite{ward63}. We also include the dendrograms obtained from applying $\ccalH$ directly to the network. It is observed that the proposed method based on cut metrics outputs the same clusters as the hierarchical (non-overlapping) method since the MNIST handwritten digit database is a clusterable set (i.e., each image is associated to only one digit).

\subsubsection{Two-digit classification. Digits \texttt{1} and \texttt{7}}

First, the performance of the algorithm is analyzed for classification of just two digits: \texttt{1} and \texttt{7}. The total number of images considered is $n=200$ counting $100$ images of each digit, selected at random from the test set of the MNIST database. Images of the digits \texttt{1} and \texttt{7} from this database are illustrated in Fig.~\ref{digits-1-7}, where the top row shows typical images that are correctly classified when using the proposed method, and the bottom row shows numbers that are often confused.

The overlapping function obtained is displayed in Fig. \ref{mnist_17-ol} for $k = 0.01$ in the computation of $\sigma$. For $\delta=8656$, the covering $K_{X }(\delta)=\{C_{1},C_{2}\}$ with $C_{1}=\{\texttt{1} (\times 100), \texttt{7}\}$ and $C_{2}=\{\texttt{7} (\times 99)\}$ is obtained. Same result is obtained when applying the Ward's linkage method at resolution $\delta=8234$. In this case, the proposed algorithm makes one classification mistake, hence, the error rate is $0.5\%$ ($1$ misclassified image out of $200$). To help illustrate this example, a non-metric Multidimensional Scaling (MDS) \cite{mds} is performed on the vectors in $\reals^{20}$ to be able to depict them on $\reals^2$. This representation is shown in Fig. \ref{mnist_17-mds}. Notice that points representing images of different digits are clearly separated, except for one instance of $\texttt{7}$ that is projected on top of the points corresponding to digit $\texttt{1}$.

\subsubsection{Three-digit classification. Digits \texttt{6}, \texttt{8} and \texttt{9}}

As done before, $100$ images are chosen at random from the test dataset for each digit, resulting in $n=300$ images. As can be observed from Fig.~\ref{digits-6-8-9} that contains images of the numbers being classified, digits \texttt{8} and \texttt{9} are much harder to distinguish than digit \texttt{6}. The overlapping function is portrayed in Fig. \ref{mnist_689-ol} for $k=0.001$. The covering $K_{X }(\delta)=\{C_{1},C_{2},C_{3}\}$ for $\delta=9247$ yields an error rate of $2.67\%$ since $C_{1}=\{\texttt{6} (\times 100)\}$, $C_{2}=\{\texttt{8} (\times 95), \texttt{9} (\times 3)\}$ and $C_{3}=\{\texttt{8} (\times 5), \texttt{9} (\times 97)\}$. Application of Ward's linkage at resolution $8854$ yields the same three clusters. The MDS representation is depicted in Fig. \ref{mnist_689-mds}, confirming that the major concern is separating \texttt{8} and \texttt{9} since these two digits occupy approximately half of the plane and \texttt{6} occupies the other half.

\subsubsection{Four-digit classification. Digits \texttt{0}, \texttt{1}, \texttt{2} and \texttt{7}}

We again consider $100$ images per digit totalizing $n=400$ images randomly selected from the test dataset. Examples of these digits are in Fig.~\ref{digits-0-1-2-7} and the plot of the overlapping function can be found in Fig. \ref{mnist_0127-ol} for $k=0.005$. For $\delta=8981$ the resulting covering is $K_{X }(\delta)=\{C_{1},C_{2},C_{3},C_{4}\}$ with $C_{1}=\{\texttt{0} (\times 100), \texttt{2} (\times 3)\}$, $C_{2}=\{ \texttt{1} (\times 99), \texttt{7} (\times 2) \}$, $C_{3}=\{\texttt{1}, \texttt{2} (\times 86), \texttt{7} (\times 3)\}$, and $C_{4} = \{ \texttt{2} (\times 11), \texttt{7} (\times 95) \}$, resulting in an error rate of $5\%$. Same results are obtained after applying Ward's linkage at resolution $\delta=8887$. The MDS representation is shown in Fig. \ref{mnist_0127-mds} where we see that \texttt{2} is as sparsely located as \texttt{0} but presents larger overlap with the clouds corresponding to \texttt{1} and \texttt{7}. This constitutes the main source of error as can be seen by inspecting blocks $C_{3}$ and~$C_{4}$.

To wrap up this first application it is worth pointing out that, while both the proposed hierarchical overlapping clustering method using cut metrics as well as the hierarchical (non-overlapping) clustering method using ultrametrics yield the same coverings, in the former the overlapping function was the key tool used to determine which resolutions yield reasonable results. In other words, the zeros of the overlapping function were used to select adequate resolution values to obtain the coverings for the hierarchical overlapping clustering method.

\subsection{Authorship Attribution}
	\label{subsec:author}

We address the problem of authorship attribution \cite{segarra15}, where our objective is to attribute a given play to its rightful author. The authors considered in this experiment are William Shakespeare, John Fletcher, Ben Jonson, and George Chapman. We consider $33$ plays written by Shakespeare, $21$ by Fletcher, $17$ by Jonson, $14$ by Chapman, as well as $2$ plays co-authored by Shakespeare and Fletcher and $1$ co-authored by Jonson and Chapman. Following the procedure described in \cite{segarra15} based on word adjacency networks, for each play we are able to obtain a dissimilarity measure to the four authors of interest. Denote by $y_{ij}$ the dissimilarity of play $i$ to author $j$. The index $j$ represents the initial of each author, $j \in \{S,F,C,J\}$. To be more specific, if $y_{ij}$ is small, then play $i$ follows a function-word structure that resembles the one typically used by author $j$; see \cite{segarra15} for details. In what follows, the node set $X$ is comprised of the subset of plays corresponding to the authors under study. Our goal is to cluster together plays written by the same author, as well as identifying co-authored plays in the overlap between author-specific clusters. By comparing with hierarchical non-overlapping clustering, we observe that our method successfully identifies the co-authored plays as being part of overlapping clusters, while for clusterable resolutions, using cut metrics amounts to the same result as using ultrametrics.

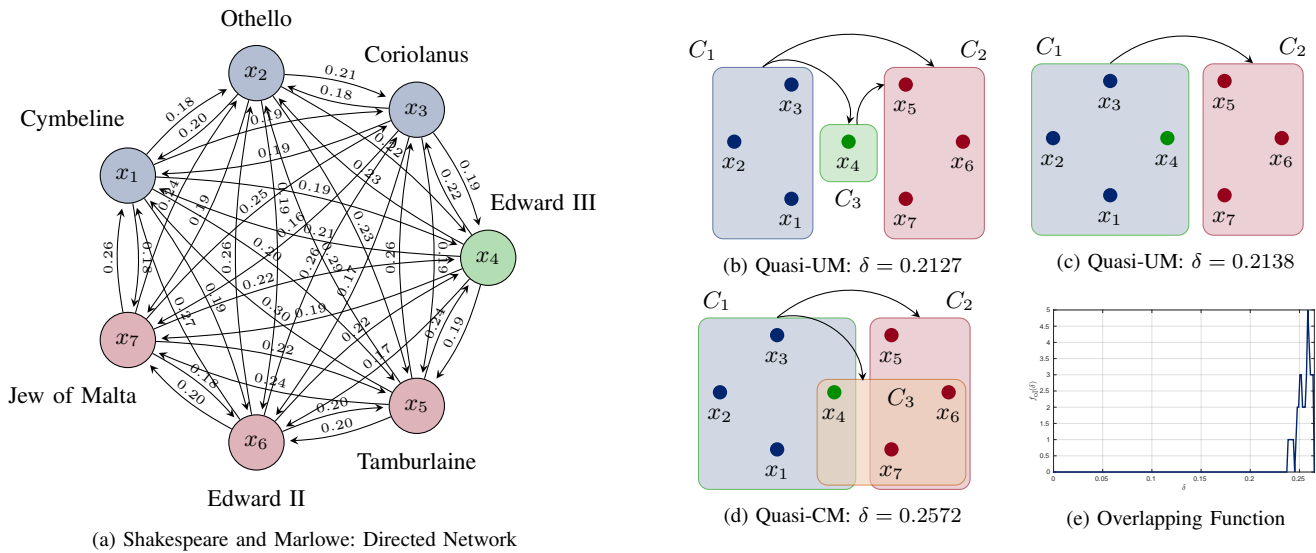
\begin{figure*}
	\centering
	\begin{subfigure}{0.99\columnwidth}
		\centering

\def \thisplotscale {0.42}
\def \unit {\thisplotscale cm}

\def \yheight{0.75}
\def \xdisplaced{0}

\tikzstyle{blue vertex} = [ellipse, 
                   inner sep=0pt, 
                   fill=pennblue!30,
                   draw=black,
                   anchor = center,
                   minimum height = 1.75*\unit, 
                   minimum width  = 1.75*\unit]

\tikzstyle{green vertex} = [ellipse, 
                   inner sep=0pt, 
                   fill=penngreen!30,
                   draw=black,
                   anchor = center,
                   minimum height = 1.75*\unit, 
                   minimum width  = 1.75*\unit]

\tikzstyle{red vertex} = [ellipse, 
                   inner sep=0pt, 
                   fill=pennred!30,
                   draw=black,
                   anchor = center,
                   minimum height = 1.75*\unit, 
                   minimum width  = 1.75*\unit]

\def\xpos{1.5}
\def\ypos{1.5}
\def\xlbl{1.75}
\def\ylbl{1.75}

{\tiny
\begin{tikzpicture}[sloped, bend left=10, -stealth, shorten >=2, scale = \thisplotscale]

{\small
	
	\path(-3.604*\xpos,1.736*\ypos) node[blue vertex] (s1) {$x_{1}$} ++ (-\xlbl,\ylbl) node {Cymbeline};
	\path(-0.89*\xpos,3.9*\ypos) node[blue vertex] (s2) {$x_{2}$} ++ (0,\ylbl) node {Othello};
	\path(2.494*\xpos,3.127*\ypos) node[blue vertex] (s3) {$x_{3}$} ++ (0,\ylbl) node {Coriolanus};
	
	\path(4*\xpos,0) node[green vertex] (sm) {$x_{4}$} ++ (\xlbl,\ylbl) node {Edward III};
	
	\path(2.494*\xpos,-3.127*\ypos) node[red vertex] (m1) {$x_{5}$} ++ (0,-\ylbl) node {Tamburlaine};
	\path(-0.89*\xpos,-3.9*\ypos) node[red vertex] (m2) {$x_{6}$} ++ (0,-\ylbl) node {Edward II};
	\path(-3.604*\xpos,-1.736*\ypos) node[red vertex] (m3) {$x_{7}$} ++ (-\xlbl,-\ylbl) node {Jew of Malta};
	
}

	\path (s3) edge [midway, above] node {$0.19$} (s1);
	\path (s1) edge [midway, above] node {$0.19$} (s3);
	
	\path (s3) edge [midway, above] node {$0.18$} (s2);
	\path (s2) edge [midway, above] node {$0.21$} (s3);
	
	\path (s3) edge [midway, above] node {$0.17$} (m2);
	\path (m2) edge [midway, above] node {$0.26$} (s3);
	
	\path (s3) edge [pos=0.45, above] node {$0.16$} (m3);
	\path (m3) edge [midway, above] node {$0.25$} (s3);
	
	\path (s3) edge [midway, above] node {$0.19$} (m1);
	\path (m1) edge [midway, above] node {$0.26$} (s3);
	
	\path (s3) edge [midway, above] node {$0.19$} (sm);
	\path (sm) edge [midway, above] node {$0.22$} (s3);

	\path (s1) edge [midway, above] node {$0.18$} (s2);
	\path (s2) edge [midway, above] node {$0.20$} (s1);
	
	\path (s1) edge [midway, above] node {$0.19$} (m2);
	\path (m2) edge [midway, above] node {$0.27$} (s1);
	
	\path (s1) edge [midway, above] node {$0.18$} (m3);
	\path (m3) edge [midway, above] node {$0.26$} (s1);
	
	\path (s1) edge [pos=0.4, above] node {$0.20$} (m1);
	\path (m1) edge [pos=0.45, above] node {$0.30$} (s1);
	
	\path (s1) edge [midway, above] node {$0.19$} (sm);
	\path (sm) edge [pos=0.45, above] node {$0.21$} (s1);

	\path (s2) edge [pos=0.3, above] node {$0.19$} (m2);
	\path (m2) edge [midway, above] node {$0.26$} (s2);
	
	\path (s2) edge [midway, above] node {$0.19$} (m3);
	\path (m3) edge [midway, above] node {$0.24$} (s2);
	
	\path (s2) edge [midway, above] node {$0.23$} (m1);
	\path (m1) edge [pos=0.45, above] node {$0.29$} (s2);
	
	\path (s2) edge [midway, above] node {$0.22$} (sm);
	\path (sm) edge [midway, above] node {$0.23$} (s2);

	\path (m2) edge [midway, above] node {$0.20$} (m3);
	\path (m3) edge [midway, above] node {$0.18$} (m2);
	
	\path (m2) edge [midway, above] node {$0.20$} (m1);
	\path (m1) edge [midway, above] node {$0.20$} (m2);
	
	\path (m2) edge [midway, above] node {$0.22$} (sm);
	\path (sm) edge [midway, above] node {$0.17$} (m2);

	\path (m3) edge [midway, above] node {$0.22$} (m1);
	\path (m1) edge [midway, above] node {$0.24$} (m3);
	
	\path (m3) edge [pos=0.35, above] node {$0.22$} (sm);
	\path (sm) edge [midway, above] node {$0.19$} (m3);

	\path (m1) edge [midway, above] node {$0.24$} (sm);
	\path (sm) edge [midway, above] node {$0.19$} (m1);

\end{tikzpicture}
}
		\caption{Shakespeare and Marlowe: Directed Network}
		\label{s-m-net}
	\end{subfigure}
	\hfill
	\begin{subfigure}{0.99\columnwidth}
		\begin{subfigure}{0.495\textwidth}
			\centering

\def \thisplotscale {3.8}
\def \unit {\thisplotscale cm}

\def \yheight{0.75}
\def \xdisplaced{0}

\tikzstyle{blue dot} = 
	[ellipse,
	 inner sep = 0pt,
	 fill = pennblue,
	 anchor = center,
	 minimum height = 0.05*\unit,
	 minimum width  = 0.05*\unit]

\tikzstyle{green dot} = 
	[ellipse,
	 inner sep = 0pt,
	 fill = penngreen,
	 anchor = center,
	 minimum height = 0.05*\unit,
	 minimum width  = 0.05*\unit]

\tikzstyle{red dot} = 
	[ellipse,
	 inner sep = 0pt,
	 fill = pennred,
	 anchor = center,
	 minimum height = 0.05*\unit,
	 minimum width  = 0.05*\unit]

\tikzstyle{cover} = 
	[rectangle, rounded corners,
	 inner sep = 0pt,
	 anchor = center]

\def\xpos{0.2}
\def\ypos{0.2}
\def\xlbl{0.07}
\def\ylbl{0.08}

{\small
\begin{tikzpicture}[scale = \thisplotscale]

	
	\path (-1.5*\xpos,-0.5*\ylbl) node [cover,minimum height = 0.6*\unit,
		minimum width = 0.35*\unit,
		opacity=0.6,
		fill=pennblue!30,
		draw=pennblue] (CS) {} ++ (-3*\xlbl,4.5*\ylbl) node {$C_{1}$};
	
	
	\path (0,-0.5*\ylbl) node [cover,minimum height = 0.2*\unit,
		minimum width = 0.2*\unit,
		opacity=0.6,
		fill=penngreen!30,
		draw=penngreen] (CSM) {} ++ (0,-2*\ylbl) node {$C_{3}$};

	
	\path (1.5*\xpos,-0.5*\ylbl) node [cover,minimum height = 0.6*\unit,
		minimum width = 0.35*\unit,
		opacity=0.6,
		fill=pennred!30,
		draw=pennred] (CM) {} ++ (2*\xlbl,4.5*\ylbl) node {$C_{2}$};
	
	
	\path (CS.north) edge[bend left=40, -stealth] (CM.north);
	\path (CS.north) edge[bend left=70, -stealth] (CSM.north);
	\path (0.15*\xpos,0.75*\ylbl) edge[bend left=40, -stealth] (0.63*\xpos,\ypos);

	\path (0,0) node [green dot] (sm) {} ++ (0,-\ylbl) node {$x_{4}$};
	
	\path (-\xpos,-\ypos) node [blue dot] (s1) {} ++ (0,-\ylbl) node {$x_{1}$};
	\path (-2*\xpos,0) node [blue dot] (s2) {} ++ (0,-\ylbl) node {$x_{2}$};
	\path (-\xpos,\ypos) node [blue dot] (s3) {} ++ (0,-\ylbl) node {$x_{3}$};
	
	\path (\xpos,\ypos) node [red dot] (m1) {} ++ (0,-\ylbl) node {$x_{5}$};
	\path (2*\xpos,0) node [red dot] (m2) {} ++ (0,-\ylbl) node {$x_{6}$};
	\path (\xpos,-\ypos) node [red dot] (m3) {} ++ (0,-\ylbl) node {$x_{7}$};

\end{tikzpicture}
}
			\caption{Quasi-UM: $\delta=0.2127$} 
			\label{s-m-qum-01}
		\end{subfigure}
		\hfill
		\begin{subfigure}{0.495\textwidth}
			\centering

\def \thisplotscale {3.8}
\def \unit {\thisplotscale cm}

\def \yheight{0.75}
\def \xdisplaced{0}

\tikzstyle{blue dot} = 
	[ellipse,
	 inner sep = 0pt,
	 fill = pennblue,
	 anchor = center,
	 minimum height = 0.05*\unit,
	 minimum width  = 0.05*\unit]

\tikzstyle{green dot} = 
	[ellipse,
	 inner sep = 0pt,
	 fill = penngreen,
	 anchor = center,
	 minimum height = 0.05*\unit,
	 minimum width  = 0.05*\unit]

\tikzstyle{red dot} = 
	[ellipse,
	 inner sep = 0pt,
	 fill = pennred,
	 anchor = center,
	 minimum height = 0.05*\unit,
	 minimum width  = 0.05*\unit]

\tikzstyle{cover} = 
	[rectangle, rounded corners,
	 inner sep = 0pt,
	 anchor = center]

\def\xpos{0.2}
\def\ypos{0.2}
\def\xlbl{0.07}
\def\ylbl{0.08}

{\small
\begin{tikzpicture}[scale = \thisplotscale]

	
	\path (-\xpos,-0.5*\ylbl) node [cover,minimum height = 0.6*\unit,
		minimum width = 0.55*\unit,
		opacity=0.6,
		fill=pennblue!30,
		draw=penngreen] (CS) {} ++ (-3*\xlbl,4.5*\ylbl) node {$C_{1}$};

	
	\path (1.5*\xpos,-0.5*\ylbl) node [cover,minimum height = 0.6*\unit,
		minimum width = 0.35*\unit,
		opacity=0.6,
		fill=pennred!30,
		draw=pennred] (CM) {} ++ (2*\xlbl,4.5*\ylbl) node {$C_{2}$};
	
	
	\path (CS.north) edge[bend left=40, -stealth] (CM.north);

	\path (0,0) node [green dot] (sm) {} ++ (0,-\ylbl) node {$x_{4}$};
	
	\path (-\xpos,-\ypos) node [blue dot] (s1) {} ++ (0,-\ylbl) node {$x_{1}$};
	\path (-2*\xpos,0) node [blue dot] (s2) {} ++ (0,-\ylbl) node {$x_{2}$};
	\path (-\xpos,\ypos) node [blue dot] (s3) {} ++ (0,-\ylbl) node {$x_{3}$};
	
	\path (\xpos,\ypos) node [red dot] (m1) {} ++ (0,-\ylbl) node {$x_{5}$};
	\path (2*\xpos,0) node [red dot] (m2) {} ++ (0,-\ylbl) node {$x_{6}$};
	\path (\xpos,-\ypos) node [red dot] (m3) {} ++ (0,-\ylbl) node {$x_{7}$};

\end{tikzpicture}
}
			\caption{Quasi-UM: $\delta=0.2138$} 
			\label{s-m-qum-02}
		\end{subfigure}
		
		\begin{subfigure}{0.495\textwidth}
			\centering

\def \thisplotscale {3.8}
\def \unit {\thisplotscale cm}

\def \yheight{0.75}
\def \xdisplaced{0}

\tikzstyle{blue dot} = 
	[ellipse,
	 inner sep = 0pt,
	 fill = pennblue,
	 anchor = center,
	 minimum height = 0.05*\unit,
	 minimum width  = 0.05*\unit]

\tikzstyle{green dot} = 
	[ellipse,
	 inner sep = 0pt,
	 fill = penngreen,
	 anchor = center,
	 minimum height = 0.05*\unit,
	 minimum width  = 0.05*\unit]

\tikzstyle{red dot} = 
	[ellipse,
	 inner sep = 0pt,
	 fill = pennred,
	 anchor = center,
	 minimum height = 0.05*\unit,
	 minimum width  = 0.05*\unit]

\tikzstyle{cover} = 
	[rectangle, rounded corners,
	 inner sep = 0pt,
	 anchor = center]

\def\xpos{0.2}
\def\ypos{0.2}
\def\xlbl{0.07}
\def\ylbl{0.08}

{\small
\begin{tikzpicture}[scale = \thisplotscale]

	
	\path (-\xpos,-0.5*\ylbl) node [cover,minimum height = 0.6*\unit,
		minimum width = 0.55*\unit,
		opacity=0.6,
		fill=pennblue!30,
		draw=penngreen] (CS) {} ++ (-3*\xlbl,4.5*\ylbl) node {$C_{1}$};

	
	\path (1.5*\xpos,-0.5*\ylbl) node [cover,minimum height = 0.6*\unit,
		minimum width = 0.35*\unit,
		opacity=0.6,
		fill=pennred!30,
		draw=pennred] (CM) {} ++ (2*\xlbl,4.5*\ylbl) node {$C_{2}$};
	
	
	\path (\xpos,-0.5*\ypos-0.5*\ylbl) node [cover,minimum height = 0.37*\unit,
		minimum width = 0.52*\unit,
		opacity=0.6,
		fill=pennorange!30,
		draw=pennorange] (CO) {} ++ (0.5*\xlbl,1.5*\ylbl) node {$C_{3}$};
	
	
	\path (CS.north) edge[bend left=40, -stealth] (CM.north);
	\path (CS.north) edge[bend left=50, -stealth] (0.5*\xpos,0.5*\ylbl);

	\path (0,0) node [green dot] (sm) {} ++ (0,-\ylbl) node {$x_{4}$};
	
	\path (-\xpos,-\ypos) node [blue dot] (s1) {} ++ (0,-\ylbl) node {$x_{1}$};
	\path (-2*\xpos,0) node [blue dot] (s2) {} ++ (0,-\ylbl) node {$x_{2}$};
	\path (-\xpos,\ypos) node [blue dot] (s3) {} ++ (0,-\ylbl) node {$x_{3}$};
	
	\path (\xpos,\ypos) node [red dot] (m1) {} ++ (0,-\ylbl) node {$x_{5}$};
	\path (2*\xpos,0) node [red dot] (m2) {} ++ (0,-\ylbl) node {$x_{6}$};
	\path (\xpos,-\ypos) node [red dot] (m3) {} ++ (0,-\ylbl) node {$x_{7}$};

\end{tikzpicture}
}
			\caption{Quasi-CM: $\delta=0.2572$} 
			\label{s-m-qcm}
		\end{subfigure}
		\hfill
		\begin{subfigure}{0.495\textwidth}
			\centering
			\vspace{0.35cm}
			\includegraphics[width=0.9\textwidth]{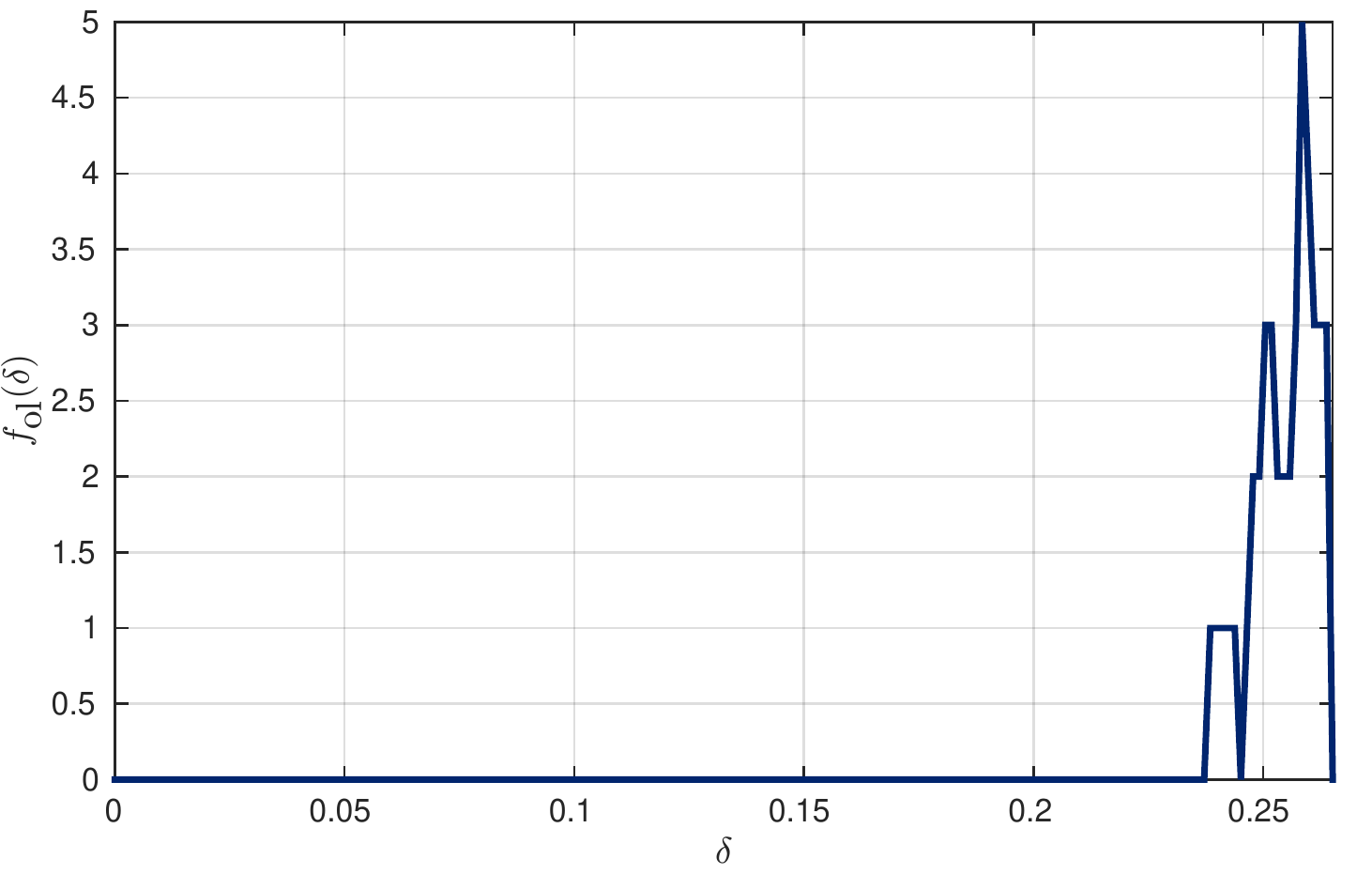}
			\caption{Overlapping Function}
			\label{s-m-ol}
		\end{subfigure}
		
	\end{subfigure}
	\caption{Authorship attribution as a directed network. \subref{s-m-net} Plays by Shakespeare in blue and Marlowe in red, and one (allegedly) co-authored play, in green. The weight of the directed edges is also shown. \subref{s-m-qum-01}-\subref{s-m-qum-02} Coverings obtained by applying directed single linkage (DSL) to the network. For resolution $\delta=0.2127$ the co-authored play forms a separate cluster, while for resolution $\delta=0.2138$ this play is incorporated to Shakespeare's plays. \subref{s-m-qcm} Coverings obtained by applying the proposed hierarchical overlapping quasi-clustering method with DSL as the quasi-ultrametric. We observe that for $\delta=0.2572$ covers are generated such that the co-authored play correctly belongs to both Shakespeare's and Marlowe's clusters. \subref{s-m-ol} Overlapping function resulting from applying the hierarchical overlapping  quasi-clustering algorithm.}
	\label{fig:authorship-directed}
\end{figure*}

\subsubsection{Shakespeare and Fletcher}

Denote by $\ccalS$ the set of $33$ plays written solely by Shakespeare, by $\ccalF$ the set of $21$ plays solely authored by Fletcher and by $\ccalS\ccalF$ the set containing the two co-authored plays. Define by $x_{i}$ the difference of dissimilarities of play $i$ to each author, $x_{i}=y_{iF}-y_{iS}$, for all $i \in \ccalS \cup \ccalF \cup \ccalS\ccalF$. The node set is then comprised of these $56$ points on $\reals$ representing the plays under study, i.e., $X=\{x_{i}\}_{i \in \ccalS \cup \ccalF \cup \ccalS\ccalF}$. The location of these points is depicted in Fig.~\ref{s-f-mds}. The dissimilarity function considered is the Euclidean distance $A_{X}(x_{k},x_{\ell})=|x_{k}-x_{\ell}|$ for every pair of plays $k$ and $l$.

In the application of Algorithm~\ref{a:overlapping-clustering}, we use Ward's linkage \cite{ward63} as the method $\ccalH$, we set $J=100$ dithering steps, and we fix $k=4$ in the determination of the noise amplitude; see Section~\ref{subsec:mnist}. The resulting overlapping function is displayed in Fig.~\ref{s-f-ol}, while the dendrogram obtained from just applying Ward's linkage is found in Fig.~\ref{s-f-dendro}. Two coverings of interest, achieved for resolutions $\delta=0.1968$ and $\delta=0.5460$ are presented in Figs.~\ref{s-f-cover-02} and \ref{s-f-cover-01}, respectively. Additionally, the clusters shown in Fig.~\ref{s-f-cover-02} are obtained by applying Ward's linkage at resolution $\delta=0.1185$.

In Fig.~\ref{s-f-cover-02} it is observed that the data does not have any overlap and that there are two resulting clusters $C_{1}=$\{Shakespeare (x 33), Fletcher\} and $C_{2}=$\{Fletcher (x 20), Shakespeare and Fletcher (x 2)\}. Thus, as a non-overlapping clustering method, it is seen that the proposed algorithm makes only one mistake but fails to identify the co-authored plays, assigning them to just one of the authors. By contrast, in Fig.~\ref{s-f-cover-01} the coverings $C_{1}=$\{Shakespeare (x 33), Fletcher (x 5), Shakespeare and Fletcher (x 2)\} and $C_{2}=$\{Fletcher (x 20), Shakespeare and Fletcher (x 2)\} \emph{do} overlap. More specifically, one error is committed and 6 plays belong to the overlap, among which the two co-authored plays are contained.

\subsubsection{Shakespeare, Chapman, and Jonson}

In this example, plays by Shakespeare, Chapman and Jonson are classified, including one play co-authored by the last two authors. Denote by $\ccalS$, $\ccalC$, $\ccalJ$, and $\ccalC\ccalJ$, the sets containing the plays solely authored by Shakespeare, Chapman, and Jonson, and the co-authored play, respectively. Mimicking the construction in the previous experiment, for a given play $i$ let $x_{i1}=y_{iC}-y_{iS}$ and $x_{i2}=y_{iJ}-y_{iS}$ be the differences in similarities between Chapman and Shakespeare, and Jonson and Shakespeare, respectively. Define the vector $\bbx_{i}=[x_{i1},x_{i2}]^{T} \in \reals^{2}$ for each play $i$ to obtain the node set $X=\{\bbx_{i}\}_{i \in \ccalS \cup \ccalC \cup \ccalJ \cup \ccalC\ccalJ}$ containing the $65$ plays from all three authors. The position of these nodes in $\reals^{2}$ is illustrated in Fig.~\ref{s-c-j-mds}. In this case, the dissimilarity function is given by the Euclidean distance between the points, i.e., $A_{X}(\bbx_{k},\bbx_{\ell})=\|\bbx_{k}-\bbx_{\ell}\|_{2}$, for plays $k$ and $l$.

In running Algorithm~\ref{a:overlapping-clustering} we select UPGMA \cite{upgma58} as the hierarchical method $\ccalH$, we set $J=100$ dithering steps and the noise amplitude is chosen for $k =1$. The resulting overlapping function is depicted in Fig.~\ref{s-c-j-ol} and the dendrogram obtained from directly applying UPGMA linkage is shown in Fig.~\ref{s-c-j-dendro}. A covering obtained from the proposed method for $\delta=0.0336$ is displayed in Fig.~\ref{s-c-j-cover}.
This covering is comprised of the blocks $C_{1}=$\{Shakespeare (x 33)\}, $C_{2}=$\{Jonson (x 16), Chapman and Jonson\}, $C_{3}=$\{Jonson (x 2), Chapman (x 14)\} and $C_{4}=$\{Jonson, Chapman, Chapman and Jonson\}. Note that block $C_{1}$ contains all of Shakespeare's plays and nothing else, thus making this classification perfect. Block $C_{2}$ contains all Jonson plays but one and also the co-authored play, and $C_{3}$ contains all Chapman plays and 2 Jonson plays, one of which overlaps with $C_{2}$. If these were all the outputs, then the algorithm would be making one full mistake (Jonson play classified as Chapman's), the co-authored play would not be recognized (as it is classified as Jonson's) and attention would be drawn to a Jonson play that is being classified as both Chapman and Jonson. However, the algorithm outputs a fourth covering $C_{4}$ that contains one play by Jonson (\emph{``Sejanus''}), one play by Chapman (\emph{``May Day''}) and the co-authored play (\emph{``Eastward Ho''}). Thus, by outputting this fourth cluster, now there are three overlapping plays: one by each author and the co-authored one. Therefore, the algorithm successfully recognizes the co-authored play and draws attention to two other plays that are hard to classify. By applying UPGMA hierarchical (non-overlapping) method we obtain at resolution $\delta=0.0357$ just three clusters $D_{X}(\delta)=\{C_{1},C_{2},C_{3}\}$ which are $C_{1}=\{$Shakespeare (x 33)$\}$, $C_{2}=\{$Jonson (x 15), Chapman and Jonson (x 1)$\}$ and $C_{3}=\{$Jonson (x2), Chapman (x14)$\}$. This clustering incurs in an error rate of $4.6\%$ and fails to identify the co-authored play.

\subsubsection{Shakespeare and Marlowe}

As a last example, we use plays by Shakespeare and Marlowe to illustrate the application of the hierarchical overlapping quasi-clustering method (Section~\ref{sec:quasi}). In this case, the network is constructed using directed edges based on the relative entropy between plays, which is asymmetric in nature, instead of using the distance to an author profile; see \cite{segarra15}. The resulting network and the corresponding directed dissimilarities $A_{X}$ are shown in Fig.~\ref{s-m-net}. For applying the algorithm we use quasi-ultrametrics obtained from applying directed single linkage (DSL) \cite{carlsson14icml} to dithered versions of the network. We carry out $J=100$ perturbations with noise given by $0.5$ times the smallest nonzero entry of the dissimilarity function $A_{X}$. For the sake of comparison, we also include results obtained from directly applying DSL to the original network.

By applying DSL we observe two relevant quasi-partitions, one for resolution $\delta=0.2127$ (Fig.~\ref{s-m-qum-01}) and one for $\delta=0.2138$ (Fig.~\ref{s-m-qum-02}). In the first case, we have three clusters, one consisting of Shakespeare's plays, one containing all Marlowe's plays, and a separate third cover including only the (allegedly) co-authored play \cite{merriam1993marlowe}. We observe that there is an edge between Shakespeare's plays and the other two covers. This can be intuitively interpreted as an indication that Shakespeare's writing is rich enough to contain part of Marlowe's stylistic fingerprint, but not vice versa; see \cite{segarra15}. Same happens for the edge between the co-authored play and Marlowe's plays, since the co-authored play also contains Shakespeare's style. When increasing the resolution to $\delta=0.2138$ we note that the co-authored play is incorrectly included as a Shakespeare play. However, when applying the hierarchical overlapping quasi-clustering algorithm, for resolution $\delta=0.2572$ we obtain the quasi-coverings illustrated in Fig.~\ref{s-m-qcm}. We observe that we have three covers, one containing all Shakespeare's plays and the co-authored play, one containing all Marlowe's plays, and the last cover containing two Marlowe's plays and the co-authored play. This third cover indicates that the co-authored play can be classified both as Shakespeare's and as Marlowe's. The resulting overlapping function is shown in Fig.~\ref{s-m-ol}.


\section{Conclusions}
	\label{sec:conclusions}

The proposed hierarchical overlapping clustering method obtains a collection of coverings of the data in such a way that a nested structure exists among the coverings and nodes can simultaneously belong to more than one cluster. The essence of the method is drawn from the connection of cut metrics, tolerance relations, and nested collections of coverings, that are respective generalizations of ultrametrics, equivalence relations, and dendrograms, typical of hierarchical (non-overlapping) clustering methods. A systematic method for obtaining cut metrics by dithering the dissimilarity function of a network is also proposed. The overlapping function is introduced as a tool for gaining insight on the structure of the data. It is also helpful in identifying meaningful resolutions within the nested collection of coverings.
Additionally, a hierarchical overlapping \emph{quasi}-clustering method was proposed. The objective of such method is to design a grouping algorithm that accommodates for asymmetries in the data, inherent to directed networks. We devised an algorithm for constructing quasi-cut metrics and showed how to obtain the output of the quasi-clustering method from these.

The hierarchical overlapping clustering method was applied to three synthetic networks to establish its performance in controlled and intuitive scenarios, as well as to illustrate how overlap can solve single linkage's chaining effect. Moreover, the method was applied to classification of handwritten digits obtained from the MNIST database showing satisfactory results in clusterable datasets. Finally, the proposed algorithm was used to classify plays by author, and it was shown that it succeeds in recognizing co-authored plays. By comparing with pertinent hierarchical non-overlapping algorithms, it was observed that for clusterable resolutions, the proposed algorithm using cut metrics coincides with the use of ultrametrics.




\bibliographystyle{IEEEtran}
\bibliography{myIEEEabrv,bib-cut-metrics}
%

%




\end{document}